\documentclass[11pt]{article}
\usepackage{epsfig,amsmath,amssymb,mathtools,amsthm}
\usepackage{euscript}
\usepackage[authoryear,round]{natbib}

\usepackage[usenames,dvipsnames]{color}
\usepackage{pdflscape}
\usepackage{multirow}
\usepackage{booktabs}
\usepackage{adjustbox}
\usepackage{graphicx}
\usepackage{subcaption}

\usepackage[table]{xcolor}
\usepackage{colortbl}

\makeatletter

\newcommand{\be}{\begin{eqnarray}}
\newcommand{\ee}{\end{eqnarray}}
\newcommand{\ben}{\begin{eqnarray*}}
\newcommand{\een}{\end{eqnarray*}}

\newtheorem{theorem}{Theorem}
\newtheorem{lemma}{Lemma}

\begin{document}

\title{\bf Multilevel Matrix Factor Model 
}

\author{Yuteng Zhang \\ {\small School of Mathematics and Statistics, Xian Jiaotong University, Xian, China
} \\
Yongchang Hui \footnote{  Corresponding author
	 } \\
{\small School of Economics and Finance, Xian Jiaotong University, Xian, China} \\
Junrong Song \\
{\small School of Mathematics and Statistics, Xian Jiaotong University, Xian, China
}\\
Shurong Zheng \\
{\small School of Mathematics and Statistics, Northeast Normal University, Changchun, China}
}
\date{October, 2023}
\maketitle

\begin{abstract}
    Large-scale matrix data has been widely discovered and continuously studied in various fields recently. 
Considering  the multi-level factor structure  and utilizing the matrix structure, we propose a multilevel matrix factor model with both global and local factors. The global factors can affect all
matrix times series, whereas the local factors are only allow to affect within each specific matrix time series. The estimation procedures can consistently estimate the factor loadings and determine
the number of factors. We
 establish  the asymptotic properties of the estimators. The simulation is presented to illustrate the performance of the proposed estimation method. We utilize the model to analyze eight indicators across 200 stocks from ten distinct industries, demonstrating the empirical utility of our proposed approach. 
\end{abstract}

\vspace{3em}

\noindent
{\it Keywords:}  Matrix-variate time series; Multilevel factor model; Eigenanalysis; Dimension reduction
\newpage
\section{Introduction}

With the rapid improvement of information collection and processing capabilities, matrix-variate time series have become increasingly prevalent in various fields, including business, ecology, finance, and social networks.
To model this kind of data, \citet{wang2019factor} proposed a matrix factor model designed to preserve the inherent matrix structure to achieve dimensionality reduction in two directions. The estimation procedures in their work is to make use of 
auto-cross-covariances of the time series to construct statistics. \citet{chen2021statistical} introduced $\alpha$-PCA, which relies on eigenanalysis performed on a combination of the sample mean matrix and the sample covariance matrices. This estimation method can be applied to 
iid matrix-variate data. \citet{yu2022projected} proposed a projection 
method that achieves faster convergence than existing estimators. The factor loading matrices were estimated by minimizing the Huber loss function in \citet{he2023matrix}, which exhibited better performance with heavy-tailed data. 
The matrix factor model was further extended to the constrained version by \citet{doi:10.1080/01621459.2019.1584899}. Additional extensions include the threshold matrix factor model studied by \citet{liu2022identification}. 
Furthermore, \citet{Chen2023TimeVaryingMF} investigated a time-varying matrix factor model allowing for smooth changes in the factor loading matrices over time. 

A common assumption in the present matrix factor models is there are only common factor processes which impact all variables in the system. However in factor analysis variables are usually  collected from different groups. For example, 
the case presented in \citet{wang2019factor} showed that the 200 companies are from different industries, and companies from the same industry are more likely to be grouped into the same categories, suggesting the existence of industry-specific factors. \citet{ando2016panel}
found that there are several groups in the Chinese mainland stock markets and provided interpretations for group-specific factors. 
 Ignoring the potential group-specific factors may result in unsufficiency and low efficiency. Therefore, we need to take into account a multi-level structure in the matrix factor model.

In this paper, we introduce a multilevel matrix-variate factor model with two types of factors: the global factors and the local factors. Both factors are assumed to be unobserved matrix time series.  Global factors are pervasive, affecting all matrix time series, while local factors impact only specific matrix time series. The proposed procedures for estimating the global and local factor loadings and latent dimensions can be achieved in three steps. In the initial step, we utilize the correlations between various matrix time series to construct statistics for estimating the global factor loading matrices. The latent dimensions of global factors can be estimated by eigenvalue ratio-based method proposed by \citet{10.1214/12-AOS970}. In the second step,  
we transform our model into a matrix factor model using the estimated global factor loading matrices. Subsequently, we apply \citet{wang2019factor}'s method to obtain estimates for the local factor loading matrices. The least square estimation is adopted to estimate the local factor process in the final step.

Early attempts of considering the multilevel structure in traditional factor model have studied by lots of scholars (e.g. \citet{wang2008large}, \citet{moench2013dynamic}, \citet{doi:10.1080/01621459.2016.1195743}, \citet{choi2018multilevel},
\citet{andreou2019inference}, \citet{han2021shrinkage}, \citet{CHOI202322}). Our work differs from the aforementioned studies in several ways. 
Firstly, these models can handle various matrix observations by vectorizing them into vectors and subsequently applying factor analysis. However, matrix observations may reveal strong correlations between rows and columns, making this approach inadequate for capturing the underlying data structure. 
Our model maintains the matrix structure. Secondly, all of these studies assume that the idiosyncratic noise component has the greatest impact on the dynamics of a few original time series.  The weak serial dependence present in idiosyncratic noise makes it challenging to identify the global and local factors separately. 
In our context, the matrix-variate time series comprise three elements: the global signal part, the local signal part, and a matrix white noise. The dynamic component is driven by lower-dimensional global and local factors time series, while the static component is represented by a matrix of white noise. The idiosyncratic noise process does not exhibit correlation over time.
    This setting is similar to those in \citet{pan2008modelling}, \citet{10.1093/biomet/asr048}, \citet{10.1214/12-AOS970}, \citet{chang2015high}, \citet{wang2019factor} and  \citet{doi:10.1080/01621459.2019.1584899}.

The rest of this article is organized as follows. In Section 2, we introduce the multilevel matrix-variate factor model. Section 3 shows the estimation procedures.
 The theoretical properties of the proposed model are included in Section 4. Section 5 reports simulation results. Real data analysis is in Section 6. Section 7 concludes.
  All proofs are presented in the Appendix. Throughout this paper, we use $\|\boldsymbol{A} \|_F$, $\| \boldsymbol{A} \|_2$ and $\| \boldsymbol{A} \|_{\text{min}}$ to denote the Frobenius norm, spectral norm and the smallest nonzero singular value of a matrix $\boldsymbol{A}$.
	Denote $\lambda_j(\boldsymbol{A})$ as the $j$-th largest eigenvalue of a nonnegative definite matrix $\boldsymbol{A}$,
 and let $\sigma_j(\boldsymbol{A})$ be the $j$th largest singular value of matrix $\boldsymbol{A}$. Let $\text{vec}(\cdot)$ be the operator that 
transforms matrix $\boldsymbol{A}$ into an vector by stacking the columns. We write $a \asymp b$, if $a= O(b)$ and $b = O(a)$.  

\section{The multilevel matrix factor model}

Let $\boldsymbol{X}_{mt}\ (t=1,...,T;m=1,...,M)$ be an observable $N_m \times p $ matrix-variate time series, each row in $\boldsymbol{X}_{mt}$ is the individual of index $m$,
\begin{equation}
\boldsymbol{X}_{mt}=\left(\begin{array}{c}
    \boldsymbol{x}_{mt,1}^{'}\\
    \vdots    \\
    \boldsymbol{x}_{mt,N_m}^{'}
     
\end{array} \right)=
\left(\begin{array}{ccc}
x_{mt, 11} & \cdots & x_{mt, 1 p} \\
\vdots & \ddots & \vdots \\
x_{mt, N_m 1} & \cdots & x_{mt, N_m p}
\end{array}\right),
\end{equation} 
where $m$ is the index for a group, $N_m$ is the number of individuals of group $m$.

We propose the following multilevel factor model for matrix-valued time series,
\begin{equation} \label{2}
\boldsymbol{X}_{mt}=\boldsymbol{R}_m \boldsymbol{G}_{t} \boldsymbol{C}_{m}^{'} +\boldsymbol{\Gamma}_{m} \boldsymbol{F}_{mt} \boldsymbol{\Lambda}_{m}^{'}+\boldsymbol{E}_{mt},\ (m=1,...,M;\ t=1,...,T)
\end{equation}
where $\boldsymbol{G}_t$ is a $k_1 \times k_2$ unobserved matrix-valued time series of global factors that affect all groups, $\boldsymbol{R}_m$ is a $N_m \times k_1$ front global loading matrix, $\boldsymbol{C}_m$ is a $p \times k_2$ 
back global loading matrix, $\boldsymbol{F}_{mt}$ is a $r_{m1} \times r_{m2}$ unobserved matrix-valued time series of local factors that affect individuals only in group $m$, $\boldsymbol{\Gamma}_m$ is a $N_m \times r_{m1}$ front local loading matrix, $\boldsymbol{\Lambda}_m$ is a $p \times r_{m2}$ back local loading matrix, and $\boldsymbol{E}_{mt}$ is a $N_m \times p$ error matrix.

\vspace{0.2in}

\noindent
{\bf Interpretation:} The model \eqref{2} can be viewed as two-step hierarchical model. This can be specified as the following three steps.

\noindent{\bf Step 1:} For each column $j = 1,...,p$, using data $\{ \boldsymbol{x}_{1t, \cdot j },\ \boldsymbol{x}_{2t, \cdot j },...,\ \boldsymbol{x}_{Mt, \cdot j },\ t=1,...,T\}$ to fit a typical multilevel factor model.
We can find $k_1$ dimensional global factors $\{ \boldsymbol{u}_{t,\cdot j} = (U_{t,1j},...,U_{t,k_1 j})^{\prime} , t=1,...,T\} $ and $r_{m_1}$ dimensional
local factors $ \{ \boldsymbol{v}_{mt,\cdot j} = (V_{mt,1j},...,V_{mt,r_{m_1} j})^{\prime} ,\ m=1,...,M,\ t=1,...,T\}. $ We have 
\begin{equation}
   \begin{pmatrix}
       X_{mt,1j} \\
       \vdots \\
       X_{mt, N_m j}
   \end{pmatrix} = \boldsymbol{R}_{m}^{(j)}  \begin{pmatrix}
       U_{t,1j} \\
       \vdots \\
       U_{t, k_1 j}
   \end{pmatrix} + \boldsymbol{\Gamma}_{m}^{(j)}
   \begin{pmatrix}
       V_{mt,1j} \\
       \vdots \\
       V_{mt, r_{m_1} j}
   \end{pmatrix}
   + \begin{pmatrix}
       H_{mt,1j} \\
       \vdots \\
       H_{mt, N_m j}
   \end{pmatrix} , t=1,2,...,T.
\end{equation}
Let 
\begin{equation*}
   \boldsymbol{U}_t =
   \begin{pmatrix}
       \boldsymbol{u}_{t,\cdot 1} &
       \cdots & 
       \boldsymbol{u}_{t,\cdot p}
   \end{pmatrix} ,\ 
   \boldsymbol{V}_{mt} =
   \begin{pmatrix}
       \boldsymbol{v}_{mt,\cdot 1} &
       \cdots & 
       \boldsymbol{v}_{mt,\cdot p}
   \end{pmatrix} ,\ \boldsymbol{H}_{mt} =
   \begin{pmatrix}
       \boldsymbol{h}_{mt,\cdot 1} &
       \cdots & 
       \boldsymbol{h}_{mt,\cdot p}
   \end{pmatrix} .
\end{equation*}
\noindent{\bf Step 2:}  Suppose each row $ i= 1,2,...,k_1$ of the matrix $\boldsymbol{U}_t$ assumes the factor structure with a $k_2$ dimensional factor $\boldsymbol{g}_{t,i \cdot}$. Also assume 
that each row $ i =1, 2,...,r_{m_1}$ of $\boldsymbol{V}_{mt}$ admits the following factor model with $r_{m_2}$ latent factors $\boldsymbol{f}_{mt, i\cdot}$. That is 
\begin{equation}
	\begin{aligned}
			(U_{t,i 1},...,U_{t, i p}) & = (G_{t,i 1},...,G_{t, i k_2}) \boldsymbol{C}^{(i)^{\prime}} + (H^*_{t, i1},...,H^*_{t, i p}), \\
			(V_{mt,i 1},...,V_{mt, i p}) & = (F_{mt,i 1},...,F_{mt, i r_{m_2}}) \boldsymbol{\Lambda}_m^{(i)^{\prime}} + (H^{**}_{mt, i1},...,H^{**}_{mt, i p}) . \\	
	\end{aligned}
\end{equation}
This step reveals the common factors that drive the co-moments in global factors $\boldsymbol{G}_t$ and local factors $\boldsymbol{F}_{mt}$.

\noindent{\bf Step 3:} Assume $\boldsymbol{R}_{m}^{(1)} = \cdots = \boldsymbol{R}_{m}^{(p)} = \boldsymbol{R}_m$, $\boldsymbol{\Gamma}_{m}^{(1)} = \cdots = \boldsymbol{\Gamma}_{m}^{(p)} = \boldsymbol{\Gamma}_m$, $\boldsymbol{C}_{m}^{(1)} = \cdots = \boldsymbol{C}_{m}^{(p)} = \boldsymbol{C}_m$,
$\boldsymbol{\Lambda}_{m}^{(1)} = \cdots = \boldsymbol{\Lambda}_{m}^{(p)} = \boldsymbol{\Lambda}_m$
we have
\begin{equation}
   \boldsymbol{X}_{mt} = \boldsymbol{R}_{m} \boldsymbol{U}_t + \boldsymbol{\Gamma}_m \boldsymbol{V}_{mt} + \boldsymbol{H}_{mt} ,\ \boldsymbol{U}_t = \boldsymbol{G}_t \boldsymbol{C}^{\prime} 
   + \boldsymbol{H}_t^{*}\ \text{and}\ \boldsymbol{V}_{mt} = \boldsymbol{F}_{mt} \boldsymbol{\Lambda}_m^{\prime} + \boldsymbol{H}_{mt}^{**}.
\end{equation}	
Hence 
\begin{equation}
    \boldsymbol{X}_{mt} = \boldsymbol{R}_{m } \boldsymbol{G}_t \boldsymbol{C}^{\prime} + \boldsymbol{\Gamma}_m \boldsymbol{F}_{mt} \boldsymbol{\Lambda}_m^{\prime} + \boldsymbol{R}_{m } \boldsymbol{H}^{*}_t + \boldsymbol{\Gamma}_m \boldsymbol{H}^{**}_{mt} +\boldsymbol{H}_{mt} = \boldsymbol{R}_{m } \boldsymbol{G}_t \boldsymbol{C}^{\prime} +
    \boldsymbol{\Gamma}_m \boldsymbol{F}_{mt} \boldsymbol{\Lambda}_m^{\prime} + \boldsymbol{E}_{mt}, 
\end{equation}
where $\boldsymbol{E}_{mt} =  \boldsymbol{R}_{m } \boldsymbol{H}^{*}_t + \boldsymbol{\Gamma}_m \boldsymbol{H}^{**}_{mt} +\boldsymbol{H}_{mt}$. This corresponding to one special case in model \eqref{2} with $\boldsymbol{C}_1 = \boldsymbol{C}_2 = \cdots = \boldsymbol{C}_M = \boldsymbol{C}$.

\noindent
{\bf Remark 1:\ }For each group $m$, the matrix-variate time series $\boldsymbol{X}_{mt}$ can be written as 
\begin{equation}
      \boldsymbol{X}_{mt}= \begin{pmatrix}\boldsymbol{R}_m & \boldsymbol{\Gamma}_{m}  \end{pmatrix}
     \begin{pmatrix}
         \boldsymbol{G}_{t} & 0 \\
         0 & \boldsymbol{F}_{mt}
         \end{pmatrix}
         \begin{pmatrix}\boldsymbol{C}_m^{'} \\ \boldsymbol{\Lambda}_{m}^{'}  \end{pmatrix} +\boldsymbol{E}_{mt}.
    \end{equation} 
It can be regarded as a matrix factor model  with constrains and ($2k_1 k_2+ 2 r_{m_1}r_{m_2}$) factors. Then if apply the matrix factor 
model to each $\boldsymbol{X}_{mt}$, $(2M k_1 k_2+ 2\sum_{m=1}^{M} r_{m_1} r_{m_2})$ factors are needed to determine whereas there are $k_1 k_2 +\sum_{m=1}^{M} r_{m_1} r_{m_2}$ 
factors using model \eqref{2}. Therefore the multilevel factor model uses fewer factors and reduces the dimension of the loading matrix by making use of
the correlation between matrix time series.

\medskip

\noindent
{\bf Remark 2:\ }
Let vec$(\cdot)$ be the vectorization operator. Model \eqref{2} can also be expressed in vector form as 
\begin{equation} \label{8}
    \text{vec}(\boldsymbol{X}_{mt})=(\boldsymbol{C}_m \otimes \boldsymbol{R}_m) \text{vec}(\boldsymbol{G}_t)+(\boldsymbol{\Lambda}_m \otimes \boldsymbol{\Gamma}_m) \text{vec}(\boldsymbol{F}_{mt}) + \text{vec}(\boldsymbol{E}_{mt}) ,\ (m=1,...,M;\ t=1,...,T)
\end{equation}
where $\otimes$ represents the Kronecker product. Model \eqref{8}  can be viewed as a specific instance of the multilevel factor model introduced in  \citet{choi2018multilevel} for vector time series, with its loading matrix given by a Kronecker
product. Comparing with model \eqref{2} and model \eqref{8}, the loading matrix contains a total of $N_m(k_1 + r_{m_1})+ p(k_2 +r_{m_2})$ parameters in model \eqref{2}, while model \eqref{8} needs $N_m p(k_1 k_2 +r_{m_1} r_{m_2})$ parameters for the loading matrices. When the matrix time series
have large dimensions, the model \eqref{2} reduces the dimension significantly.

\medskip

\noindent
{\bf Remark 3:\ }
Model \eqref{2} is applicable to the data $\boldsymbol{P}_{1t},...,\boldsymbol{P}_{Mt}$ of dimensions
$ (N, p_1), (N, p_2),...,(N, p_M)$, respectively. This type of data can be viewed the same set of individuals with different categories of characteristics. If one matrix-variate 
 time series $\boldsymbol{X}_t$ has dimensions $N*p$. It can be divided by columns or rows based on some prior knowledge or criteria. By doing so, we can adapt the resulting data to fit the model described in equation \eqref{2}.
\section{Estimation Procedures}
\label{Section 3}

Since the model \eqref{2} is unchanged if we replace triplets $(\boldsymbol{R}_{m},\ \boldsymbol{G}_t,\ \boldsymbol{C}_{m})$ by $(\boldsymbol{R}_{m} \boldsymbol{U}_1,\ \boldsymbol{U}_1^{-1}\boldsymbol{G}_t \boldsymbol{U}_2^{-1},\ \boldsymbol{C}_{m} \boldsymbol{U}_{2}^{'})$ or replace triplets $(\boldsymbol{\Gamma}_m,\ \boldsymbol{F}_{mt},\ \boldsymbol{\Lambda}_m)$ by $(\boldsymbol{\Gamma}_m \boldsymbol{U}_3,\ \boldsymbol{U}_3^{-1}\boldsymbol{F}_{mt} \boldsymbol{U}_4^{-1},\ \boldsymbol{\Lambda}_m \boldsymbol{U}_{4}^{'})$ for any invertible matrices $\boldsymbol{U}_1$, $\boldsymbol{U}_2$, $\boldsymbol{U}_3$, $\boldsymbol{U}_4$ of size $k_1 \times k_1$, $k_2 \times k_2$, $r_{m1} \times r_{m1}$ and $r_{m2} \times r_{m2}$, model \eqref{2} is not identifiable. 
However, the factor loading space spanned by the columns of loading matrix $\boldsymbol{R}_m,\ \boldsymbol{C}_m,\ \boldsymbol{\Gamma}_m$ and $\boldsymbol{\Lambda}_m$ are uniquely determined. So our estimation targets are the column space of $\boldsymbol{R}_m,\ \boldsymbol{C}_m,\ \boldsymbol{\Gamma}_m$ and $\boldsymbol{\Lambda}_m$, denoted by $\mathcal{M}(\boldsymbol{R}_m),\ \mathcal{M}(\boldsymbol{C}_m),\ \mathcal{M}(\boldsymbol{\Gamma
}_m)\ \text{and}\ \mathcal{M}(\boldsymbol{\Lambda}_m)$.

We further decompose $\boldsymbol{R}_m,\ \boldsymbol{C}_m,\ \boldsymbol{\Gamma}_m$ and $\boldsymbol{\Lambda}_m$ as follows:
$$ \boldsymbol{R}_m=\boldsymbol{Q}_{1m} \boldsymbol{K}_{1m},\quad \boldsymbol{C}_m=\boldsymbol{Q}_{2m} \boldsymbol{K}_{2m},\quad \boldsymbol{\Gamma}_m=\boldsymbol{Q}_{3m} \boldsymbol{K}_{3m},\ \text{and}\  \boldsymbol{\Lambda}_m=\boldsymbol{Q}_{4m} \boldsymbol{K}_{4m},\ m=1,...,M,$$
where $\boldsymbol{Q}_{im}$\ $(i=1,2,3,4;\ m=1,...,M)$ are semi-orthogonal matrices, $\boldsymbol{K}_{im}$\ $(i=1,2,3,4;\ m=1,...,M)$ are non-singular matrices. This can be achieved by QR decomposition. Then model \eqref{2} can be rewritten as
\begin{equation}
\boldsymbol{X}_{mt}=\boldsymbol{Q}_{1m} \boldsymbol{S}_{mt} \boldsymbol{Q}_{2m}^{'} +\boldsymbol{Q}_{3m} \boldsymbol{Z}_{mt} \boldsymbol{Q}_{4m}^{'}+\boldsymbol{E}_{mt},\ (m=1,...,M;\ t=1,...,T)   
\end{equation}
where $\boldsymbol{S}_{mt}=\boldsymbol{K}_{1m} \boldsymbol{G}_t \boldsymbol{K}_{2m}^{\prime},\ \boldsymbol{Z}_{mt}=\boldsymbol{K}_{3m} \boldsymbol{F}_{mt} \boldsymbol{K}_{4m}^{\prime}$. Note that $\mathcal{M}(\boldsymbol{R}_m)=\mathcal{M}(\boldsymbol{Q}_{1m})$, $\mathcal{M}(\boldsymbol{C}_m)=\mathcal{M}(\boldsymbol{Q}_{2m})$, $\mathcal{M}(\boldsymbol{\Gamma}_m)=\mathcal{M}(\boldsymbol{Q}_{3m})$ and $\mathcal{M}(\boldsymbol{\Lambda}_m)=\mathcal{M}(\boldsymbol{Q}_{4m})$. So the estimation of column spaces of $\boldsymbol{R}_m,\ \boldsymbol{C}_m,\ \boldsymbol{\Gamma}_m$ and $\boldsymbol{\Lambda}_m$ are equivalent to the estimation of column spaces of $\boldsymbol{Q}_{1m},\ \boldsymbol{Q}_{2m},\ \boldsymbol{Q}_{3m}$ and $\boldsymbol{Q}_{4m}$. 

Our estimation procedures are present in Section \ref{Section 3}.  In Section \ref{Section 3.1}, based on similar ideas presented in \citet{10.1093/biomet/asr048} and \citet{wang2019factor}, we rely on the covariance matrix to construct a non-negative definite matrix for estimating the global factor loadings.
 To determine the underlying dimensions of global factors, we use the ratio-based estimator. The local factor loadings and latent dimensions are estimated using the estimation procedures outlined in \citet{wang2019factor} in Section \ref{Section 3.2}.
  The normalized local factors are obtained through least squares estimation, the global and local signal parts estimation are presented in Section \ref{Section 3.3}.
   \subsection{Estimation of global factor loading matrix} \label{Section 3.1}
   Let $\boldsymbol{x}_{\cdot,j,mt}$, $\boldsymbol{e}_{\cdot,j,mt}$, $\boldsymbol{c}_{\cdot,j,m}$ and $\boldsymbol{\lambda}_{\cdot,j,m}$ denote the $j$th column of $\boldsymbol{X}_{mt}$, $\boldsymbol{E}_{mt}$, $\boldsymbol{C}_m$ and $\boldsymbol{\Lambda}_m$, respectively. Then 
   \begin{equation}
       \boldsymbol{x}_{\cdot,j,mt}=\boldsymbol{Q}_{1m} \boldsymbol{K}_{1m} \boldsymbol{G}_{t} \boldsymbol{c}_{j,\cdot,m}^{\prime}+ \boldsymbol{Q}_{3m} \boldsymbol{K}_{3m} \boldsymbol{F}_{mt} \boldsymbol{\lambda}_{j,\cdot,m}^{\prime} + \boldsymbol{e}_{\cdot,j,mt},\quad j=1,...,p.
   \end{equation}
   Define
   \begin{equation} \label{11}
       \begin{aligned}
           \boldsymbol{\Omega}_{x, m n, j_1 j_2} & =\frac{1}{T} \sum_{t=1}^{T} \text{Cov}(\boldsymbol{x}_{\cdot,j_1,mt}, \boldsymbol{x}_{\cdot,j_2,nt})\\
           & 
           \begin{split}
               = \frac{1}{T} \sum_{t=1}^T \text{Cov} ( & \boldsymbol{Q}_{1m} \boldsymbol{K}_{1m} \boldsymbol{G}_{t} \boldsymbol{c}_{j_1,\cdot,m}^{\prime}+ \boldsymbol{Q}_{3m} \boldsymbol{K}_{3m} \boldsymbol{F}_{mt} \boldsymbol{\lambda}_{j_1,\cdot,m}^{\prime} + \boldsymbol{e}_{\cdot,j_1,mt},\\ 
               & \boldsymbol{Q}_{1n} \boldsymbol{K}_{1n} \boldsymbol{G}_{t} \boldsymbol{c}_{j_2,\cdot,n}^{\prime}+ \boldsymbol{Q}_{3n} \boldsymbol{K}_{3n} \boldsymbol{F}_{nt} \boldsymbol{\lambda}_{j_2,\cdot,n}^{\prime} + \boldsymbol{e}_{\cdot,j_2,nt})
           \end{split}\\
           & = \frac{1}{T} \sum_{t=1}^T \boldsymbol{Q}_{1m} \text{Cov} (\boldsymbol{K}_{1m} \boldsymbol{G}_t \boldsymbol{c}_{j_1,\cdot,m}, \boldsymbol{K}_{1n} \boldsymbol{G}_t \boldsymbol{c}_{j_2,\cdot,n}) \boldsymbol{Q}_{1n}^{\prime},
       \end{aligned}
   \end{equation}
   the third equality is achieved by the Condition 6. Then follow the idea of \citet{10.1093/biomet/asr048} to construct the following statistics:
   \begin{equation}  \label{W1mdef}
    \boldsymbol{W}_{1,m}  =\sum_{i=1,i\neq m}^{M} \sum_{j_1=1}^{p} \sum_{j_2=1}^{p} \boldsymbol{\Omega}_{x, m i , j_1 j_2} \boldsymbol{\Omega}_{x, m i, j_1 j_2}^{\prime},
   \end{equation}
   $m=1,...,M.$ This is a $N_m \times N_m$ nonnegative definite matrix. From Equation \eqref{11} we can get $\boldsymbol{W}_{1,m}$ has a $\boldsymbol{Q}_{1m}$ on the left and a $\boldsymbol{Q}_{1m}^{\prime}$ on the right.
   
   Suppose the matrix $\boldsymbol{W}_{1,m}$ has rank $k_1$. Since each column of $\boldsymbol{W}_{1,m}$ can be expressed as a linear combination of columns from $\boldsymbol{R}_m$, the matrices $\boldsymbol{W}_{1,m}$ and $\boldsymbol{R}_m$ share the same column space. Let $\boldsymbol{q}_{1m,i}$ denote the unit eigenvectors of 
   $ \boldsymbol{W}_{1,m}$ corresponding to its $i$th largest eigenvalue. Define
   \begin{equation*}
       \boldsymbol{Q}_{1m} = (\boldsymbol{q}_{1m,1},...,\boldsymbol{q}_{1m,k_1}).
   \end{equation*}
   If we
    further assume that $\boldsymbol{W}_{1,m}$ has $k_1$ distinct nonzero eigenvalues and ignore the trivial sign changes of $\boldsymbol{q}_{1m,i}$, then $\boldsymbol{Q}_{1m}$ can be uniquely identified from the population version of $\boldsymbol{W}_{1,m}$. The factor loading space $\mathcal{M}(\boldsymbol{Q}_{1m})$ is spanned by the eigenvectors of
   $\boldsymbol{W}_{1,m}$ corresponding to its nonzero eigenvalues. 
   
   For the estimation of $\mathcal{M}(\boldsymbol{Q}_{1m})$, we construct the sample version of $\boldsymbol{W}_{1,m}$ as follows:
   \begin{equation} \label{hatW1mdef}
   \begin{aligned} 
         \widehat{\boldsymbol{\Omega}}_{x,mn,j_1 j_2} & = \frac{1}{T} \sum_{t=1}^{T} \boldsymbol{x}_{\cdot,j_1,mt} \boldsymbol{x}_{\cdot,j_2,nt}^{\prime} ,\\
         \widehat {\boldsymbol{W}}_{1,m} & =\sum_{i=1,i\neq m}^{M} \sum_{j_1=1}^{p} \sum_{j_2=1}^{p} \widehat{\boldsymbol{\Omega}}_{x,mi,j_1 j_2} \widehat{\boldsymbol{\Omega}}_{x,mi,j_1 j_2}^{\prime}, 
   \end{aligned}
   \end{equation}
   where $\boldsymbol{x}_{\cdot,j_1,mt}$ represents the $j_1$th column of matrix $\boldsymbol{X}_{mt}$. Then, $\mathcal{M}(\boldsymbol{Q}_{1m})$ can be estimated by $\mathcal{M}(\widehat{\boldsymbol{Q}}_{1m})$, where $\widehat{\boldsymbol{Q}}_{1m}=(\widehat{\boldsymbol{q}}_{1m,1},...,\widehat{\boldsymbol{q}}_{1m,k_1})$ 
   with $\widehat{\boldsymbol{q}}_{1m,1},...,\widehat{\boldsymbol{q}}_{1m,k_1}$ being the eigenvectors corresponding to the $k_1$ largest eigenvalues of $\widehat{\boldsymbol{W}}_{1,m}$.  Let $\widehat{\boldsymbol{B}}_{1m}=(\widehat{\boldsymbol{q}}_{1m,k_1+1},...,\widehat{\boldsymbol{q}}_{1m,N_m})$ represent the estimated orthogonal complement matrix of $\widehat{\boldsymbol{Q}}_{1m}$.
   
   The estimation procedure describes above assumes that the number of  global row factors $k_1$ is known. In practice, to estimate $k_1$ we employ the ratio-based estimator proposed by \cite{10.1214/12-AOS970}. Let $\widehat{\lambda}_{m,1} \geq \widehat{\lambda}_{m,2} \geq \cdots \geq \widehat{\lambda}_{m,N_m} \geq 0$ be the ordered eigenvalues of $\widehat{\boldsymbol{W}}_{1,m}$. Then 
   \begin{equation}
       \widehat{k}_1 = \text{argmin}_{1\leq i\leq N_m/3} \frac{\widehat{\lambda}_{m,i+1}}{\widehat{\lambda}_{m,i}}.
   \end{equation}
   
   The estimation procedure for $\mathcal{M}(\mathbf{C}_m)$ and $k_2$ can be applied to the transpose of $\boldsymbol{X}_{mt}$ in the same manner.
   \subsection{Estimation of local factor loading matrix} \label{Section 3.2}
Rewrite model \eqref{2} as 
\begin{equation} \label{15}
    \begin{aligned}
    \boldsymbol{X}_{mt} & =\boldsymbol{Q}_{1m}\boldsymbol{S}_{mt} \boldsymbol{Q}_{2m}^{'}  +\boldsymbol{\Gamma}_{m} \boldsymbol{F}_{mt} \boldsymbol{\Lambda}_{m}^{'}+\boldsymbol{E}_{mt},\\
    & = (\boldsymbol{Q}_{1m}-\widehat{\boldsymbol{Q}}_{1m})\boldsymbol{S}_{mt}  \boldsymbol{Q}_{2m}^{'} +\widehat{\boldsymbol{Q}}_{1m}\boldsymbol{S}_{mt}  \boldsymbol{Q}_{2m}^{'}
    +\boldsymbol{\Gamma}_{m} \boldsymbol{F}_{mt} \boldsymbol{\Lambda}_{m}^{'}+\boldsymbol{E}_{mt}\\
    & = \widehat{\boldsymbol{Q}}_{1m}\boldsymbol{S}_{mt}  \boldsymbol{Q}_{2m}^{'}
    +\boldsymbol{\Gamma}_{m} \boldsymbol{F}_{mt} \boldsymbol{\Lambda}_{m}^{'}+\boldsymbol{E}_{mt}^{(1)},
    \end{aligned}
\end{equation}
where $\boldsymbol{E}_{mt}^{(1)}=\boldsymbol{E}_{mt}+(\boldsymbol{Q}_{1m}-\widehat{\boldsymbol{Q}}_{1m})\boldsymbol{S}_{mt}  \boldsymbol{Q}_{2m}^{'}.$ 

Then left multiplying Equation \eqref{15} by $\widehat{\boldsymbol{B}}^{\prime}_{1m}$, we have
\begin{equation}
    \boldsymbol{Y}_{mt}^{(1)}:=\widehat{\boldsymbol{B}}^{\prime}_{1m} \boldsymbol{X}_{mt} = \widehat{\boldsymbol{B}}^{\prime}_{1m}\boldsymbol{\Gamma}_{m} \boldsymbol{F}_{mt} \boldsymbol{\Lambda}_{m}^{'}+\widehat{\boldsymbol{B}}^{\prime}_{1m}\boldsymbol{E}_{mt}^{(1)} = \widehat{\boldsymbol{B}}^{\prime}_{1m}\boldsymbol{Q}_{3m} \boldsymbol{Z}_{mt} \boldsymbol{Q}_{4m}^{'}+\widehat{\boldsymbol{B}}^{\prime}_{1m}\boldsymbol{E}_{mt}^{(1)} ,
\end{equation}
where $ \boldsymbol{Z}_{mt} = \boldsymbol{K}_{3m} \boldsymbol{F}_{mt} \boldsymbol{K}_{4m}^{\prime}$. Because matrix $\widehat{\boldsymbol{B}}^{\prime}_{1m}\boldsymbol{\Gamma}_{m}$ is full column-rank, the estimation  procedures outline in \cite{wang2019factor} can 
be applied to estimate $\boldsymbol{Q}_{4m}$ and the latent dimension $r_{m2}$. Similarly, let $\widehat{\boldsymbol{B}}_{2m} = (\widehat{\boldsymbol{c}}_{m,k_2+1},...,\widehat{\boldsymbol{c}}_{m,p})$ denote the estimated orthogonal complement
matrix of $\widehat{\boldsymbol{Q}}_{2m}$. We transpose $\boldsymbol{X}_{mt}$ then we have
\begin{equation} \label{Equation 17}
    \boldsymbol{X}^{\prime}_{mt}  =\widehat{\boldsymbol{Q}}_{2m}\boldsymbol{S}^{\prime}_{mt} \boldsymbol{Q}^{\prime}_{1m}  +\boldsymbol{Q}_{4m} \boldsymbol{Z}^{\prime}_{mt} \boldsymbol{Q}_{3m}^{'}+\boldsymbol{E}_{mt}^{(2)},
\end{equation}
where  $\boldsymbol{E}_{mt}^{(2)}=\boldsymbol{E}^{\prime}_{mt}+(\boldsymbol{Q}_{2m}-\widehat{\boldsymbol{Q}}_{2m})\boldsymbol{S}^{\prime}_{mt}  \boldsymbol{Q}_{1m}^{'}.$ Now it follows from \eqref{Equation 17} that
\begin{equation}
\boldsymbol{Y}_{mt}^{(2)}:=\widehat{\boldsymbol{B}}^{\prime}_{2m}\boldsymbol{X}^{\prime}_{mt} = \widehat{\boldsymbol{B}}^{\prime}_{2m} \boldsymbol{Q}_{4m} \boldsymbol{Z}^{\prime}_{mt} \boldsymbol{Q}_{3m}^{'}+\widehat{\boldsymbol{B}}^{\prime}_{2m} \boldsymbol{E}_{mt}^{(2)}.
\end{equation}
Similarly, we can estimate $\boldsymbol{Q}_{3m}$ and the latent dimension $\boldsymbol{r}_{m1}$ using the approaches outlined in \cite{wang2019factor}.
\subsection{Estimation of dynamic global and local signal part} \label{Section 3.3}
From Section \ref{Section 3.2} we have obtained the estimation of matrices $\boldsymbol{Q}_{3m}$ and $\boldsymbol{Q}_{4m}$, for $m=1,...,M$. Consequently, the normalized local factors $\boldsymbol{Z}_{mt} $ can be estimated by the least square estimation method:
\begin{equation}
    \widehat{\boldsymbol{Z}}_{mt} = (\boldsymbol{A}_m \boldsymbol{A}_m ^{\prime})^{-1} \boldsymbol{A}_m ^{\prime} \boldsymbol{Y}_{mt}^{(1)} \widehat{\boldsymbol{Q}}_{4m} ,
\end{equation}
where $\boldsymbol{A}_m = \widehat{\boldsymbol{B}}_{1m}^{\prime}  \widehat{\boldsymbol{Q}}_{3m}$.

Let $\boldsymbol{\Psi}_{mt}$ and $\boldsymbol{\Phi}_{mt}$ represent the signal part of global and local, respectively. That is, $ \boldsymbol{\Psi}_{mt} = \boldsymbol{R}_{m} \boldsymbol{G}_t \boldsymbol{C}_{m}^{\prime} = \boldsymbol{Q}_{1m} \boldsymbol{S}_{mt} \boldsymbol{Q}_{2m}^{\prime}
, \boldsymbol{\Phi}_{mt} = \boldsymbol{\Gamma}_{m} \boldsymbol{F}_{mt} \boldsymbol{\Lambda}_{m}^{\prime} = \boldsymbol{Q}_{3m} \boldsymbol{Z}_{mt} \boldsymbol{Q}_{4m}^{\prime}$. So $\boldsymbol{\Phi}_{mt}$ can be estimated by
\begin{equation}
    \widehat{\boldsymbol{\Phi}}_{mt} = \widehat{\boldsymbol{Q}}_{3m} \widehat{\boldsymbol{Z}}_{mt} \widehat{\boldsymbol{Q}}_{4m}^{\prime}.
\end{equation}
Consequently, we estimate the normalized global factors $\boldsymbol{S}_{mt}$ and $\boldsymbol{\Psi}_{mt}$, respectively, by
\begin{equation}
	\begin{aligned}
		\widehat{\boldsymbol{S}}_{mt} & = \widehat{\boldsymbol{Q}}_{1m}^{\prime} (\boldsymbol{X}_{mt} - \widehat{\boldsymbol{\Phi}}_{mt} )\widehat{\boldsymbol{Q}}_{2m} , \\
		\widehat{\boldsymbol{\Psi}}_{mt} & = \widehat{\boldsymbol{Q}}_{1m} \widehat{\boldsymbol{Q}}_{1m}^{\prime}(\boldsymbol{X}_{mt} - \widehat{\boldsymbol{\Phi}}_{mt} )\widehat{\boldsymbol{Q}}_{2m} \widehat{\boldsymbol{Q}}_{2m}^{\prime} .	
	\end{aligned}   
\end{equation}

\section{Theoretical Properties}
\label{Section 4}

In this section, we present the convergence rates of the estimators under the setting that all $T$, $N$ and $p$ tend towards infinity, while keeping the latent dimensions $(k_1, k_2, r_{m_1}, r_{m2})$ fixed for $m=1,...,M$.

We define the following notations first: $\boldsymbol{\Sigma}_G =\text{Var} (\text{vec}(\boldsymbol{G}_t)),$
\begin{equation}
    \boldsymbol{\Sigma}_{fm}(h ) = \frac{1}{T} \sum_{t=1}^{T-h } \text{Cov}(\text{vec}(\boldsymbol{F}_{mt}),\ \text{vec}(\boldsymbol{F}_{m,t+h})), \  \text{and}\ \boldsymbol{\Sigma}_{em} = \text{Var} (\text{vec}(\boldsymbol{E}_{mt})),\ m=1,...,M.
\end{equation}

Then we introduce some regularity conditions.
\vspace{.1in}

\noindent {\bf Condition 1.} The process $\{(\text{vec}(\boldsymbol{G}_t),\ \text{vec}(\boldsymbol{F}_{1t}),...,\ \text{vec}(\boldsymbol{F}_{Mt}))\}$ is 
$\alpha$-mixing with the mixing coefficients satisfying the condition $\sum_{k=1}^{\infty} \alpha(k)^{1-2/\gamma} < \infty$ for
some $\gamma >2$, where
\begin{equation}
   \alpha(k) =  \mathop{\sup}_{i} \mathop{\sup}_{A \in \mathcal{F}_{-\infty }^{i},\ B \in \mathcal{F}_{i+k}^{\infty}} | P(A) P(B) - P(AB) |, 
\end{equation}
and $\mathcal{F}_{i }^{j}$ is the $\sigma$-field generated by $\{(\text{vec}(\boldsymbol{G}_t),\ \text{vec}(\boldsymbol{F}_{1t}),...,\ \text{vec}(\boldsymbol{F}_{Mt})) :\ i \leq t\leq j\}$.
\vspace{.1in}

\noindent {\bf Condition 2.} For any $i= 1,...,k_1, j=1,...,k_2$ and $t=1,...,T$, we have $E(| \boldsymbol{G}_{t,ij}|^{2\gamma} |) \leq C$; for any $ i = 1,...,r_{m1}, j=1,...,r_{m2}, m=1,...,M$ and $t=1,...,T$, we have $E(| \boldsymbol{F}_{mt,ij}|^{2\gamma} |) \leq C$
, where $C$ is a positive constant, and $\gamma$ is given in Condition 1. $ \boldsymbol{\Sigma}_G$ is rank $ k = \text{max} (k_1, k_2)$ and $ \| \boldsymbol{\Sigma_G} \|_2 \asymp O(1) \asymp \sigma_k (\boldsymbol{\Sigma}_G)$.  For $i=1,...,r_{m1}$ 
and $j=1,...,r_{m2}$, $\frac{1}{T-h} \sum_{t=1}^{T-h} \text{Cov}(\boldsymbol{f}_{mt,i}, \boldsymbol{f}_{m(t+h),i}) \neq \boldsymbol{0}$, $\frac{1}{T-h} \text{Cov}(\boldsymbol{f}_{mt,j\cdot}, \boldsymbol{f}_{m(t+h),j\cdot}) \neq \boldsymbol{0}.$  In additional, 
there exists a $1\leq h\leq h_0$ such that $\text{rank}(\boldsymbol{\Sigma}_{fm}(h)) \geq r_m$, and $\| \boldsymbol{\Sigma}_{fm}(h)\|_2 \asymp O(1) \asymp \sigma_{r_m}(\boldsymbol{\Sigma}_{fm}(h))$ for $m=1,...,M$, where $r_m = \text{max} \{ r_{m_1}, r_{m_2}\}$.

\vspace{.1in}
\noindent {\bf Condition 3.} Each elements of $\boldsymbol{\Sigma}_{em}$ and $\|\boldsymbol{\Sigma}_{em}\|_2\ (m=1,...,M)$ remain bounded as $N$ and $p$ increase to infinity.

\vspace{.1in}
\noindent {\bf Condition 4.}  There exist constant $\delta_1,\ \delta_2,\ \delta_3$ and $\delta_4$ in $[0,1]$ such that $\|\boldsymbol{R}_m\|_2^2 \asymp N^{1-\delta_1} \asymp \|\boldsymbol{R}_m\|_{\text{min}}^2$,
$\|\boldsymbol{C}_m\|_2^2 \asymp p^{1-\delta_1} \asymp \|\boldsymbol{C}_m\|_{\text{min}}^2$, $\|\boldsymbol{\Gamma}_m\|_2^2 \asymp \|\boldsymbol{B}_{1m}^{\prime} \boldsymbol{\Gamma}_m\|_2^2
\asymp N^{1- \delta_3} \asymp \|\boldsymbol{\Gamma}_m\|_{\text{min}}^2 \asymp \|\boldsymbol{B}_{1m}^{\prime} \boldsymbol{\Gamma}_m\|_{\text{min}}^2$ and  $\|\boldsymbol{\Lambda}_m\|_2^2 \asymp \|\boldsymbol{B}_{2m}^{\prime} \boldsymbol{\Lambda}_m\|_2^2
\asymp p^{1- \delta_4} \asymp \|\boldsymbol{\Lambda}_m\|_{\text{min}}^2 \asymp \|\boldsymbol{B}_{2m}^{\prime} \boldsymbol{\Lambda}_m\|_{\text{min}}^2(m=1,...,M)$.

\vspace{.1in}

$\|\boldsymbol{B}_{1m}^{\prime} \boldsymbol{\Gamma}_m\|_2^2
\asymp N^{1- \delta_3}  \asymp \|\boldsymbol{B}_{1m}^{\prime} \boldsymbol{\Gamma}_m\|_{\text{min}}^2$ and $\|\boldsymbol{B}_{2m}^{\prime} \boldsymbol{\Lambda}_m\|_2^2
\asymp p^{1- \delta_4} \asymp \|\boldsymbol{B}_{2m}^{\prime} \boldsymbol{\Lambda}_m\|_{\text{min}}^2(m=1,...,M)$ are necessary for the estimation. As an example, consider the case where $\mathcal{M}(\boldsymbol{R}_m) \supseteq \mathcal{M}(\boldsymbol{\Gamma}_m)$, implying that
$\boldsymbol{B}_{1m}^{\prime} \boldsymbol{\Gamma}_m = \mathbf{0}$. Then the estimation procedures in \citet{wang2019factor} cannot be used to estimate the column space of $\mathcal{M}(\boldsymbol{\Lambda}_m)$. The rates $\delta_1$ and $\delta_2$ are called the strength for the global row factors and the global column factors, respectively. Similarly, $\delta_3$ and $\delta_4$ are the strengths of the local row factors and the local column factors, respectively. If $\delta_i = 0$, the corresponding factors are strong. If $\delta_i > 0$, the corresponding factors are weak. Smaller values of $\delta$'s indicate stronger factors. 

For the sake of simplicity in notation, we assume the factor strength tetrad $(\delta_1, \delta_2, \delta_3, \delta_4)$ is identical for all groups. However, it's important to note that the theoretical properties can still be achieved if different factor strengths are allowed for different groups.  

\vspace{.1in}
\noindent {\bf Condition 5.} $\boldsymbol{W}_{im}$ has $k_i$ distinct positive eigenvalues for $i=1,2$ and $m=1,2,...,M$.
\vspace{.1in}

Condition 5 implies that the $\boldsymbol{Q}_{im}(i=1,2;\ m=1,...,M)$ defined in Section 3 are unique. Conditions 1-5 are similarly to those in \citet{wang2019factor}. For a more detailed explanation, please refer to \citet{wang2019factor}.

\vspace{.1in}
\noindent {\bf Condition 6.} $\{\text{vec}(\boldsymbol{G}_t) \},\ \{\text{vec}(\boldsymbol{F}_{1t}) \},\ ...\ ,\{\text{vec}(\boldsymbol{F}_{Mt}) \}$ are uncorrelated. That is , ${\rm E}(\text{vec}(\boldsymbol{F}_{mt})\text{vec}(\boldsymbol{F}_{nt})^{\prime}) 
= \mathbf{0}$ for all $t$ and $m \neq n$; and 
${\rm E}(\text{vec}(\boldsymbol{G}_{t})\text{vec}(\boldsymbol{F}_{mt})^{\prime}) = \mathbf{0}$ for all $m$ and $t$. And ${\rm E}(\text{vec}(\boldsymbol{E}_{mt})\text{vec}(\boldsymbol{E}_{nt})^{\prime}) = \mathbf{0}$ for all $t$ and $m \neq n$.

\vspace{.1in}
\noindent {\bf Condition 7.} $N_1,...,N_M$ and $N = N_1+ \cdots +N_M $ are of the same order of magnitude.

\vspace{.1in}

Condition 6 can be viewed as an extension of Assumption 1 in \citet{choi2018multilevel} to the matrix regime. The uncorrelation between two kinds of factors is to make sure the global factors and the local
factors can be separately identified. Since we use the covariance structure to estimate the global factor loading matrices, the uncorrelation between local factors is needed. 

The following theorems establish the convergence rates of the estimators.
\begin{theorem} \label{T1}
    Under Conditions 1-7 and $ N^{\delta_1} p^{\delta_2} T^{-1/2} = o(1) $, it holds that
    \begin{equation*}
        \| \widehat{\boldsymbol{Q}}_{im} - {\boldsymbol{Q}}_{im} \|_2 = O_p(N^{\delta_1} p^{\delta_2} T^{-1/2}),\ for\ i=1, 2.
    \end{equation*}
\end{theorem}

The convergence rates presented in Theorem 1 are consistent with those in \citet{wang2019factor}. When all the global factors are strong, i.e., $\delta_1 = \delta_2 = 0$
, the convergence rates  become $O_p(T^{-1/2})$. It's the optimal rate. The convergence rate of global factor loadings does not depend on the strength of the local factors. This is due to the uncorrelation between global factors and local factors allows us to treat the local component as noise  when estimating the global factor loadings.
 Even when $\delta_3 = \delta_4 = 0$, indicating that the local signal part is as strong as the noise, the noise plays the dominate role.

\begin{theorem} \label{T2}
    Under Conditions 1-7 and $ N^{\delta_1} p^{\delta_2} T^{-1/2} = o(1) $, the eigenvalues $\{ \widehat{\lambda}_{1m,1},..., \widehat{\lambda}_{1m,k_1}\}$ of $\widehat{\boldsymbol{W}}_{1,m}(m=1,...,M)$ which
    are sorted in descending order, satisfy
    \begin{equation*}
        \begin{aligned}
             | \widehat{\lambda}_{1m,j} - {\lambda}_{1m,j} | & = O_p(N^{2-\delta_1} p^{2-\delta_2} T^{-1/2}),\ for\ j=1, 2,...,k_1, \\
             | \widehat{\lambda}_{1m,j} | & = O_p(N^{2} p^{2} T^{-1}),\ for\ j=k_1+1, 2,...,N_m,
        \end{aligned}
    \end{equation*}
    where ${\lambda}_{1m,1}  > \lambda_{1m,2} > \cdots > {\lambda}_{1m,N_m}$ are the eigenvalues of $\boldsymbol{W}_{1,m}$.
\end{theorem}

Theorem \ref{T2} shows that the convergence of estimators for the zero eigenvalues of $\boldsymbol{W}_{1,m}$ is faster compared to that of the nonzero eigenvalues. This observation forms the theoretical basis for the ratio-based estimator discussed in Section 3.

\begin{theorem} \label{T3}
    Under Conditions 1-7 , $ N^{\delta_1} p^{\delta_2} T^{-1/2} = o(1) $ and $ N^{\delta_3} p^{\delta_4} T^{-1/2} = o(1) $, we have
    \begin{equation*}
        \| \widehat{\boldsymbol{Q}}_{im} - {\boldsymbol{Q}}_{im} \|_2 = 
		O_p \left({\rm max} \{ N^{\delta_1} p^{\delta_2} T^{-1/2} , N^{\delta_3} p^{\delta_4} T^{-1/2}\}\right),\ for\ i=3, 4.
    \end{equation*}
\end{theorem}

    Theorem \ref{T3} demonstrates the convergence rate of the local factor loadings. The convergence rate of local factor loadings depends on both the strength of the global factors and the strength of the local factors.
 This is not unexpected because the estimation of the local factor loadings depends on the accuracy of the estimated global factor loadings. As it mentioned in Theorem \ref{T1} , $N^{\delta_1} p^{\delta_2} T^{-1/2}$ is the convergence rate of 
 global factor loadings. And \citet{wang2019factor} established the convergence rate of  $N^{\delta_3} p^{\delta_4} T^{-1/2}$ when no global factors are present in the model \eqref{2}. Hence ${\rm max} \{ N^{\delta_1} p^{\delta_2} T^{-1/2} , N^{\delta_3} p^{\delta_4} T^{-1/2}\}$
  gives the convergence rate for the local factor loadings. The convergence rate $N^{\delta_3} p^{\delta_4} T^{-1/2}$ can be achieved when the condition ${\rm E}(\text{vec}(\boldsymbol{G}_{t_1}) \text{vec}(\boldsymbol{F}_{m t_2})^{\prime}) = \boldsymbol{0}$ is satisfied for all $t_1,\ t_2$
   and $m$.

   The following theorems show the theoretical properties of the estimated signal parts $\widehat{\boldsymbol{\Phi}}_{mt}$ and $\widehat{\boldsymbol{\Psi}}_{mt}$.

   \begin{theorem} \label{T4}
    Under Conditions 1-7 , $ N^{\delta_1} p^{\delta_2} T^{-1/2} = o(1) $ and $ N^{\delta_3} p^{\delta_4} T^{-1/2} = o(1) $, we have
    \begin{equation*}
       N^{-1/2} p^{-1/2} \| \widehat{\boldsymbol{\Phi}}_{mt} - {\boldsymbol{\Phi}}_{mt} \|_2 = O_p( N^{\delta_1/2} p^{\delta_2/2} T^{-1/2} + N^{\delta_3/2} p^{\delta_4/2} T^{-1/2} + N^{-1/2} p^{-1/2}).
    \end{equation*}
\end{theorem}

\begin{theorem} \label{T5}
    Under Conditions 1-7 , $ N^{\delta_1} p^{\delta_2} T^{-1/2} = o(1) $ and $ N^{\delta_3} p^{\delta_4} T^{-1/2} = o(1) $, we have
    \begin{equation*}
       N^{-1/2} p^{-1/2} \| \widehat{\boldsymbol{\Psi}}_{mt} - {\boldsymbol{\Psi}}_{mt} \|_2 = O_p( N^{\delta_1/2} p^{\delta_2/2} T^{-1/2} + N^{\delta_3/2} p^{\delta_4/2} T^{-1/2} + N^{-1/2} p^{-1/2}).
    \end{equation*}
\end{theorem}
The consistency of the global and local signal estimate requires $N, p \rightarrow \infty$. Importantly, both the global signal part and the local signal part exhibit the same convergence rate.

\section{Simulation}
\begin{table}[!htbp]
    \tabcolsep  = 0.15cm
    \begin{center}
        \caption{Means and standard deviations of $\mathcal{D}(\widehat{\boldsymbol{Q}}_{i1}, \boldsymbol{Q}_{i1}),\ i = 1,2,3,4$, over 200 simulation runs. For ease of presentation, all numbers in this table are the true numbers multipied by 10.}
    \begin{tabular}{cccccccccc}
        \hline 
        \multicolumn{4}{c}{\textbf{Panel\ A}: $i=1,2$} & \multicolumn{2}{c}{$T = n*p $} & \multicolumn{2}{c}{$ T = 1.5*n*p $} &\multicolumn{2}{c}{$ T= 2*n*p $}\\
        \hline\rule{0pt}{15pt}
        $ \delta_1 $ & $ \delta_2 $ & $ n $ & $ p $ & $ \mathcal{D}(\widehat{\boldsymbol{Q}}_{11}, \boldsymbol{Q}_{11}) $ & $ \mathcal{D}(\widehat{\boldsymbol{Q}}_{21}, \boldsymbol{Q}_{21}) $ & $ \mathcal{D}(\widehat{\boldsymbol{Q}}_{11}, \boldsymbol{Q}_{11}) $ & $ \mathcal{D}(\widehat{\boldsymbol{Q}}_{21}, \boldsymbol{Q}_{21}) $ & $ \mathcal{D}(\widehat{\boldsymbol{Q}}_{11}, \boldsymbol{Q}_{11}) $ & $ \mathcal{D}(\widehat{\boldsymbol{Q}}_{21}, \boldsymbol{Q}_{21}) $ \\
        \hline
        0.5 & 0.5 & 20 & 20 &5.48(0.04) &6.67(0.41) &5.00(0.09)&6.88(0.41)&4.87(0.07)&4.38(0.41) \\
        && 20 & 40 &5.74(0.04)&6.32(0.34)&5.67(0.03)&2.35(0.84)&5.74(0.02)&2.75(1.26)\\
        && 40 & 40 &5.77(0.01)&6.89(0.13)&5.78(0.01)&6.53(1.41)&5.57(0.03)&2.34(0.21)\\
        \hline
        0.5 & 0 & 20 & 20 & 1.56(0.47)&0.85(0.27) &0.75(0.19)&0.37(0.08)&0.73(0.17)&0.46(0.11) \\
        && 20 & 40 &0.77(0.21)&0.35(0.07)&0.36(0.07)&0.22(0.03)&0.31(0.06)&0.39(0.13)\\
        && 40 & 40 &0.25(0.04)&0.21(0.03)&0.26(0.04)&0.14(0.02)&0.17(0.03)&0.13(0.02)\\
        \hline
        0 & 0 & 20 & 20 & 0.40(0.13)&0.18(0.05) &1.05(0.56)&0.30(0.10)&0.24(0.08)&0.15(0.04) \\
        && 20 & 40 &0.21(0.08)&0.15(0.03)&0.14(0.05)&0.15(0.04)&0.18(0.06)&0.13(0.04)\\
        && 40 & 40 &0.11(0.02)&0.11(0.03)&0.07(0.02)&0.06(0.02)&0.08(0.02)&0.04(0.01)\\
        \hline \rule{0pt}{15pt}
    \end{tabular}
    \begin{tabular}{cccccccccc}
        \hline 
        \multicolumn{4}{c}{\textbf{Panel\ B}: $i=3,4$}  & \multicolumn{2}{c}{$T = n*p $} & \multicolumn{2}{c}{$ T = 1.5*n*p $} &\multicolumn{2}{c}{$ T= 2*n*p $}\\
        \hline \rule{0pt}{15pt}
        $ \delta_3 $ & $ \delta_4 $ & $ n $ & $ p $ & $ \mathcal{D}(\widehat{\boldsymbol{Q}}_{31}, \boldsymbol{Q}_{31}) $ & $ \mathcal{D}(\widehat{\boldsymbol{Q}}_{41}, \boldsymbol{Q}_{41}) $ & $ \mathcal{D}(\widehat{\boldsymbol{Q}}_{31}, \boldsymbol{Q}_{31}) $ & $ \mathcal{D}(\widehat{\boldsymbol{Q}}_{41}, \boldsymbol{Q}_{41}) $ & $ \mathcal{D}(\widehat{\boldsymbol{Q}}_{31}, \boldsymbol{Q}_{31}) $ & $ \mathcal{D}(\widehat{\boldsymbol{Q}}_{41}, \boldsymbol{Q}_{41}) $ \\
        \hline
        0.5 & 0.5 & 20 & 20 & 6.95(0.10)& 6.98(0.11)&6.14(0.26)&6.54(0.27)&7.01(0.08)&5.09(1.07) \\
        && 20 & 40 &7.02(0.01)&4.35(1.58)&6.40(0.07)&5.42(0.27)&7.06(0.01)&1.09(0.33)\\
        && 40 & 40 &7.08(0.02)&6.37(0.35)&6.87(0.10)&5.63(0.58)&6.26(0.17)&6.98(0.03)\\
        \hline
        0.5 & 0 & 20 & 20 &1.62(1.18) &1.88(0.53) &0.97(0.70)&1.09(0.42)&1.96(0.52)&0.87(0.37) \\
        && 20 & 40 &2.24(0.57)&0.51(0.11)&0.51(0.14)&0.27(0.04)&0.20(0.04)&0.28(0.05)\\
        && 40 & 40 &0.22(0.05)&0.23(0.04)&0.21(0.05)&0.18(0.03)&0.24(0.07)&0.19(0.03)\\
        \hline
        0 & 0 & 20 & 20 & 0.36(0.07)&0.28(0.05) &0.23(0.06)&0.21(0.06)&0.21(0.04)&0.16(0.03) \\
        && 20 & 40 &0.12(0.02)&0.15(0.02)&0.10(0.02)&0.13(0.02)&0.09(0.02)&0.14(0.02)\\
        && 40 & 40 &0.09(0.02)&0.10(0.01)&0.09(0.01)&0.07(0.01)&0.06(0.01)&0.05(0.00)\\
        \hline
    \end{tabular}
    \label{Table 1}
\end{center}
\end{table} 
\begin{table}[!htbp]
    \tabcolsep  = 0.15cm
    \begin{center}
        \caption{Means and standard deviations of $\mathcal{D}(\widehat{\boldsymbol{Q}}_{i2}, \boldsymbol{Q}_{i2}),\ i = 1,2,3,4$, over 200 simulation runs. For ease of presentation, all numbers in this table are the true numbers multipied by 10.}
    \begin{tabular}{cccccccccc}
        \hline 
        \multicolumn{4}{c}{\textbf{Panel\ A}: $i=1,2$} & \multicolumn{2}{c}{$T = n*p $} & \multicolumn{2}{c}{$ T = 1.5*n*p $} &\multicolumn{2}{c}{$ T= 2*n*p $}\\
        \hline\rule{0pt}{15pt}
        $ \delta_1 $ & $ \delta_2 $ & $ n $ & $ p $ & $ \mathcal{D}(\widehat{\boldsymbol{Q}}_{12}, \boldsymbol{Q}_{12}) $ & $ \mathcal{D}(\widehat{\boldsymbol{Q}}_{22}, \boldsymbol{Q}_{22}) $ & $ \mathcal{D}(\widehat{\boldsymbol{Q}}_{12}, \boldsymbol{Q}_{12}) $ & $ \mathcal{D}(\widehat{\boldsymbol{Q}}_{22}, \boldsymbol{Q}_{22}) $ & $ \mathcal{D}(\widehat{\boldsymbol{Q}}_{12}, \boldsymbol{Q}_{12}) $ & $ \mathcal{D}(\widehat{\boldsymbol{Q}}_{22}, \boldsymbol{Q}_{22}) $ \\
        \hline
        0.5 & 0.5 & 20 & 20 &5.75(0.01) &6.75(0.41) &5.82(0.02)&4.38(0.59)&5.73(0.05)&3.06(0.53) \\
        && 20 & 40 &5.75(0.05)&6.66(0.40)&5.85(0.04)&5.18(1.86)&4.77(0.09)&4.25(1.61)\\
        && 40 & 40 &5.50(0.02)&6.30(0.28)&5.63(0.02)&5.39(0.60)&5.65(0.02)&1.83(0.84)\\
        \hline
        0.5 & 0 & 20 & 20 & 3.98(0.67)&0.51(0.13) &0.67(0.12)&0.49(0.12)&0.97(0.18)&0.45(0.10) \\
        && 20 & 40 &0.45(0.11)&0.36(0.09)&0.32(0.06)&0.19(0.03)&0.21(0.04)&0.27(0.07)\\
        && 40 & 40 &0.25(0.05)&0.30(0.06)&0.27(0.04)&0.11(0.02)&0.15(0.02)&0.10(0.02)\\
        \hline
        0 & 0 & 20 & 20 & 0.62(0.28)&0.54(0.22) &1.4(0.80)&0.20(0.07)&0.34(0.14)&0.16(0.04) \\
        && 20 & 40 &0.27(0.10)&0.12(0.03)&0.19(0.08)&0.13(0.04)&0.16(0.06)&0.14(0.04)\\
        && 40 & 40 &0.15(0.05)&0.08(0.02)&0.10(0.03)&0.05(0.01)&0.08(0.02)&0.04(0.01)\\
        \hline \rule{0pt}{15pt}
    \end{tabular}
    \begin{tabular}{cccccccccc}
        \hline 
        \multicolumn{4}{c}{\textbf{Panel\ B}: $i=3,4$}  & \multicolumn{2}{c}{$T = n*p $} & \multicolumn{2}{c}{$ T = 1.5*n*p $} &\multicolumn{2}{c}{$ T= 2*n*p $}\\
        \hline \rule{0pt}{15pt}
        $ \delta_3 $ & $ \delta_4 $ & $ n $ & $ p $ & $ \mathcal{D}(\widehat{\boldsymbol{Q}}_{32}, \boldsymbol{Q}_{32}) $ & $ \mathcal{D}(\widehat{\boldsymbol{Q}}_{42}, \boldsymbol{Q}_{42}) $ & $ \mathcal{D}(\widehat{\boldsymbol{Q}}_{32}, \boldsymbol{Q}_{32}) $ & $ \mathcal{D}(\widehat{\boldsymbol{Q}}_{42}, \boldsymbol{Q}_{42}) $ & $ \mathcal{D}(\widehat{\boldsymbol{Q}}_{32}, \boldsymbol{Q}_{32}) $ & $ \mathcal{D}(\widehat{\boldsymbol{Q}}_{42}, \boldsymbol{Q}_{42}) $ \\
        \hline
        0.5 & 0.5 & 20 & 20 & 6.88(0.06)& 7.08(0.06)&6.83(0.01)&6.98(0.01)&7.05(0.02)&7.16(0.18) \\
        && 20 & 40 &7.00(0.07)&6.04(0.40)&6.36(0.20)&6.73(0.61)&6.91(0.03)&6.49(0.18)\\
        && 40 & 40 &7.12(0.03)&6.21(0.26)&6.99(0.01)&6.68(0.20)&6.93(0.02)&4.05(0.45)\\
        \hline
        0.5 & 0 & 20 & 20 &1.54(1.32) &3.09(1.35) &0.67(0.19)&4.24(0.80)&1.24(0.26)&0.51(0.14) \\
        && 20 & 40 &0.46(0.14)&0.54(0.14)&0.19(0.04)&0.26(0.04)&0.74(0.24)&0.49(0.13)\\
        && 40 & 40 &0.50(0.14)&0.30(0.06)&0.37(0.09)&0.17(0.03)&0.16(0.03)&0.17(0.03)\\
        \hline
        0 & 0 & 20 & 20 & 0.29(0.06)&0.33(0.06) &0.20(0.03)&0.25(0.07)&0.18(0.03)&0.15(0.03) \\
        && 20 & 40 &0.16(0.03)&0.20(0.03)&0.11(0.02)&0.14(0.02)&0.11(0.02)&0.13(0.02)\\
        && 40 & 40 &0.11(0.02)&0.08(0.01)&0.06(0.01)&0.05(0.01)&0.06(0.01)&0.05(0.00)\\
        \hline
    \end{tabular}
    \label{Table 2}
\end{center}
\end{table} 
\begin{table}[!htbp]
    \tabcolsep  = 0.15cm
    \begin{center}
        \caption{Means and standard deviations of $\mathcal{D}(\widehat{\boldsymbol{Q}}_{i3}, \boldsymbol{Q}_{i3}),\ i = 1,2,3,4$, over 200 simulation runs. For ease of presentation, all numbers in this table are the true numbers multipied by 10.}
    \begin{tabular}{cccccccccc}
        \hline 
        \multicolumn{4}{c}{\textbf{Panel\ A}: $i=1,2$} & \multicolumn{2}{c}{$T = n*p $} & \multicolumn{2}{c}{$ T = 1.5*n*p $} &\multicolumn{2}{c}{$ T= 2*n*p $}\\
        \hline\rule{0pt}{15pt}
        $ \delta_1 $ & $ \delta_2 $ & $ n $ & $ p $ & $ \mathcal{D}(\widehat{\boldsymbol{Q}}_{13}, \boldsymbol{Q}_{13}) $ & $ \mathcal{D}(\widehat{\boldsymbol{Q}}_{23}, \boldsymbol{Q}_{23}) $ & $ \mathcal{D}(\widehat{\boldsymbol{Q}}_{13}, \boldsymbol{Q}_{13}) $ & $ \mathcal{D}(\widehat{\boldsymbol{Q}}_{23}, \boldsymbol{Q}_{23}) $ & $ \mathcal{D}(\widehat{\boldsymbol{Q}}_{13}, \boldsymbol{Q}_{13}) $ & $ \mathcal{D}(\widehat{\boldsymbol{Q}}_{23}, \boldsymbol{Q}_{23}) $ \\
        \hline
        0.5 & 0.5 & 20 & 20 &5.79(0.01) &7.04(0.04) &5.93(0.16)&6.89(0.16)&5.24(0.11)&3.21(0.43) \\
        && 20 & 40 &5.47(0.06)&6.73(0.15)&5.70(0.09)&4.07(0.55)&5.79(0.01)&3.02(1.79)\\
        && 40 & 40 &5.50(0.02)&7.01(0.08)&5.75(0.01)&4.32(0.92)&5.28(0.04)&2.91(0.21)\\
        \hline
        0.5 & 0 & 20 & 20 & 1.81(0.67)&0.43(0.10) &0.60(0.15)&0.33(0.08)&0.57(0.12)&0.32(0.06) \\
        && 20 & 40 &0.37(0.07)&0.30(0.06)&0.39(0.05)&0.21(0.03)&0.45(0.09)&0.28(0.07)\\
        && 40 & 40 &0.34(0.06)&0.17(0.03)&0.27(0.04)&0.12(0.02)&0.14(0.02)&0.11(0.02)\\
        \hline
        0 & 0 & 20 & 20 & 0.55(0.24)&0.29(0.10) &0.30(0.10)&0.21(0.07)&0.18(0.05)&0.15(0.04) \\
        && 20 & 40 &0.15(0.03)&0.20(0.06)&0.37(0.13)&0.17(0.04)&0.17(0.08)&0.19(0.06)\\
        && 40 & 40 &0.09(0.02)&0.06(0.01)&0.08(0.02)&0.04(0.01)&0.08(0.02)&0.04(0.01)\\
        \hline \rule{0pt}{15pt}
    \end{tabular}
    \begin{tabular}{cccccccccc}
        \hline 
        \multicolumn{4}{c}{\textbf{Panel\ B}: $i=3,4$}  & \multicolumn{2}{c}{$T = n*p $} & \multicolumn{2}{c}{$ T = 1.5*n*p $} &\multicolumn{2}{c}{$ T= 2*n*p $}\\
        \hline \rule{0pt}{15pt}
        $ \delta_3 $ & $ \delta_4 $ & $ n $ & $ p $ & $ \mathcal{D}(\widehat{\boldsymbol{Q}}_{33}, \boldsymbol{Q}_{33}) $ & $ \mathcal{D}(\widehat{\boldsymbol{Q}}_{43}, \boldsymbol{Q}_{43}) $ & $ \mathcal{D}(\widehat{\boldsymbol{Q}}_{32}, \boldsymbol{Q}_{33}) $ & $ \mathcal{D}(\widehat{\boldsymbol{Q}}_{43}, \boldsymbol{Q}_{43}) $ & $ \mathcal{D}(\widehat{\boldsymbol{Q}}_{33}, \boldsymbol{Q}_{33}) $ & $ \mathcal{D}(\widehat{\boldsymbol{Q}}_{43}, \boldsymbol{Q}_{43}) $ \\
        \hline
        0.5 & 0.5 & 20 & 20 & 6.55(0.21)& 6.91(0.15)&6.61(0.27)&6.58(1.28)&6.91(0.09)&5.66(0.09) \\
        && 20 & 40 &6.73(0.03)&5.51(0.16)&6.94(0.01)&6.59(0.13)&6.58(0.12)&6.41(0.61)\\
        && 40 & 40 &6.55(0.03)&6.40(0.31)&7.00(0.01)&7.00(0.01)&7.03(0.02)&2.92(0.74)\\
        \hline
        0.5 & 0 & 20 & 20 &0.57(0.16) &0.87(0.29) &1.81(0.77)&1.63(0.65)&0.93(0.23)&0.69(0.20) \\
        && 20 & 40 &1.71(0.40)&0.36(0.06)&0.25(0.05)&0.33(0.06)&0.26(0.06)&0.22(0.03)\\
        && 40 & 40 &0.38(0.10)&0.26(0.04)&0.24(0.05)&0.32(0.08)&0.20(0.03)&0.18(0.03)\\
        \hline
        0 & 0 & 20 & 20 & 0.24(0.04)&0.23(0.04) &0.21(0.04)&0.19(0.03)&0.23(0.04)&0.18(0.03) \\
        && 20 & 40 &0.15(0.02)&0.26(0.06)&0.09(0.02)&0.13(0.02)&0.10(0.02)&0.11(0.01)\\
        && 40 & 40 &0.10(0.02)&0.08(0.01)&0.07(0.01)&0.05(0.01)&0.06(0.01)&0.05(0.00)\\
        \hline
    \end{tabular}
    \label{Table 3}
\end{center}
\end{table} 
For two orthogonal matrices $\boldsymbol{O}_1$ and $\boldsymbol{O}_2$ of size $p \times q_1$ and $ p \times q_2$, we define 
\begin{equation}
    \mathcal{D}\left(\boldsymbol{O}_1, \boldsymbol{O}_2\right)=\left(1-\frac{1}{\max \left(q_1, q_2\right)} \operatorname{Tr}\left(\boldsymbol{O}_1 \boldsymbol{O}_1^{\top} \boldsymbol{O}_2 \boldsymbol{O}_2^{\top}\right)\right)^{1 / 2}.
    \end{equation}
Here $ \mathcal{D}\left(\boldsymbol{O}_1, \boldsymbol{O}_2\right) \in [0,1]$, which measures the distance between the two linear spaces $\mathcal{M}(\boldsymbol{O}_1)$ and $\mathcal{M}(\boldsymbol{O}_2)$.
It equals to 0 if and only if $ \mathcal{M}(\boldsymbol{O}_1) = \mathcal{M}(\boldsymbol{O}_2)$ and equals to 1 if and only if $\mathcal{M}(\boldsymbol{O}_1)  \perp \mathcal{M}(\boldsymbol{O}_2)$.

Three groups' case is chosen for our simulation. The observed data are generated by the following model:
\begin{equation}
    \left\{
    \begin{aligned}
   \boldsymbol{X}_{1t} & =\boldsymbol{R}_1 \boldsymbol{G}_{t} \boldsymbol{C}_{1}^{'} +\boldsymbol{\Gamma}_{1} \boldsymbol{F}_{1t} \boldsymbol{\Lambda}_{1}^{'}+\boldsymbol{E}_{1t},   \\
         \boldsymbol{X}_{2t} & =\boldsymbol{R}_2 \boldsymbol{G}_{t} \boldsymbol{C}_{2}^{'} +\boldsymbol{\Gamma}_{2} \boldsymbol{F}_{2t} \boldsymbol{\Lambda}_{2}^{'}+\boldsymbol{E}_{2t} ,\ (t=1,2,...,T)\\
         \boldsymbol{X}_{3t} & =\boldsymbol{R}_3 \boldsymbol{G}_{t} \boldsymbol{C}_{3}^{'} +\boldsymbol{\Gamma}_{3} \boldsymbol{F}_{3t} \boldsymbol{\Lambda}_{3}^{'}+\boldsymbol{E}_{3t} . 
    \end{aligned}
    \right.
\end{equation}
The dimensions of the matrix global factors are assumed to be $k_1=3$ and $k_2=2$. The true dimension
 of all groups are set to $(r_{m1},r_{m2})=(2,2),\ m=1,2,3$. The number of individuals in each group is chosen to be the same, i.e., $N_1=N_2=N_3= n$. The elements of $\boldsymbol{G}_t$ follow $k_1 k_2$ independent AR$(1)$
  processes with Gaussian white noise $\mathcal{N}(0,1)$ innovations. These AR coefficients are given by $[-0.5\ 0.6;\ 0.8\ -0.4;\ 0.7\ 0.3]$.   
  The entries of $\boldsymbol{F}_{mt}$ follow $r_{m1} r_{m2} $ independent AR$(1)$ processes with Gaussian white noise $\mathcal{N}(0,1)$ innovations.
   These AR coefficients are $[-0.5\ 0.6;\ 0.8\ -0.4]$. The entries of $\boldsymbol{R}_m$ and $\boldsymbol{\Gamma}_m$  are randomly sampled from 
   uniform distributions $U(-n^{-\delta_i/2},n^{-\delta_i/2})$ for $i=1,3$. Similarly, the entries of $\boldsymbol{C}_m$ and $\boldsymbol{\Lambda}_m$  are independently drawn from the 
   uniform distributions $U(-p^{-\delta_i/2},p^{-\delta_i/2})$ for $i=2,4$. The error process $\boldsymbol{E}_{mt}$ is a white noise process with mean $\boldsymbol{0}$ and a Kronecker product structure, that is,
    $\text{Cov}(\text{vec}(\boldsymbol{E}_{mt})) = \boldsymbol{\Sigma}_2 \otimes \boldsymbol{\Sigma}_{m1}$, where $\boldsymbol{\Sigma}_{m1}$ and $\boldsymbol{\Sigma}_2$ are of sizes $N_m \times N_m$ and $p \times p$.
     For $\boldsymbol{\Sigma}_{m1}$ and $\boldsymbol{\Sigma}_2$, the diagonal elements are set to $1$ and off-diagonal elements are set to $0.2$. Three pairs of $(\delta_1,\ \delta_2,\ \delta_3,\ \delta_4)$ combinations are chosen in our work: $(0.5, 0.5, 0.5, 0.5),\ (0.5, 0, 0.5, 0),\ (0, 0, 0, 0)$. The sample size $T$ is selected as $n p$, $1.5  n p $, $2 n  p $.
   
	Table \ref{Table 1}-\ref{Table 3} exhibit the mean and standard deviations of the estimation errors $\mathcal{D}(\widehat{\boldsymbol{Q}}_{im}, \boldsymbol{Q}_{im})$ for group $m =1,2,3$ and global row ($i=1$), global column ($i=2$), local row ($i=3$)
	and local column $(i=4)$ loading space. The results show that increasing the strength of column or row factors can enhance the accuracy of estimation.

    Secondly, we evaluate the performance of estimating the number of global factors and local factors, denoted as $k_1$, $k_2$, $r_{m_1}$ and $r_{m_2}$. We use the same data as presented in Table \ref{Table 1}-\ref{Table 3} with $(\delta_1,\ \delta_2,\ \delta_3,\ \delta_4) = \ (0, 0, 0, 0)$.
	We use $f_m$ to represent the relative frequency of correctly estimating the true number of factors $(3, 2, 2, 2)$ for each group $m$ $(m=1,2,3)$. Table \ref{Table 4} displays the relative frequency of the estimated global and local factor numbers. The accuracy improves when the
	 sample size $T$ increases. 
	\begin{table}[!htbp]
		\tabcolsep  = 0.15cm
		\begin{center}
			\caption{Relative frequencies of correctly estimating the number of global and local factors.}
		\begin{tabular}{ccccccccccccccc}
			\hline 
            \multicolumn{6}{c}{} & \multicolumn{3}{c}{$T = n*p $} & \multicolumn{3}{c}{$ T = 1.5*n*p $} &\multicolumn{3}{c}{$ T= 2*n*p $}\\
			\hline\rule{0pt}{15pt}
			$ \delta_1 $ & $ \delta_2 $ & $\delta_3$ & $\delta_4$ & $ n $ & $ p $ & $ f_1 $ & $ f_2 $ & $ f_3 $ &   $ f_1 $ & $ f_2 $ & $ f_3 $ &  $ f_1 $ & $ f_2 $ & $ f_3 $ \\
			0 & 0 & 0 & 0 & 20 & 20 & 0.935 & 0.035 &0.71&0.295 &0.035&0.665 & 0.6 & 0.59 & 0.97 \\
			&&&& 20 & 40 &0.7&0.565&1 &0.895&0.93&0.975& 1 & 0.5& 0.69 \\
			&&&& 40 & 40 &0.97&0.965&0.995&1&0.95 &0.97& 1 & 0.995 &1\\
			\hline \rule{0pt}{15pt}
		\end{tabular}
		\label{Table 4}
		\end{center}
	\end{table}
 
	Finally, we investigate the performance of recovering the global signal part $\boldsymbol{\Psi}_{mt}$ and the local signal part $\boldsymbol{\Phi}_{mt}$ for group $m=1,2,3.$ 
	In this evaluation, we consider the case where $(\delta_1,\ \delta_2,\ \delta_3,\ \delta_4)=(0, 0, 0, 0)$. To quantify the accuracy of recovery, we calculate the distances  $\mathcal{D}(\widehat{\boldsymbol{\Psi}}_m,\boldsymbol{\Psi}_m)$
	and $\mathcal{D}(\widehat{\boldsymbol{\Phi}}_m,\boldsymbol{\Phi}_m)$. These are estimated by the following equations:
	\begin{equation}
		\begin{aligned}
			\mathcal{D}(\widehat{\boldsymbol{\Psi}}_m,\boldsymbol{\Psi}_m) & = N_m^{-1/2} p^{-1/2} \left(\sum_{t=1}^{T} \| \widehat{\boldsymbol{\Psi}}_{mt} - \boldsymbol{\Psi}_{mt} \|_2/T \right) ,\\
			\mathcal{D}(\widehat{\boldsymbol{\Phi}}_m,\boldsymbol{\Phi}_m) & = N_m^{-1/2} p^{-1/2} \left(\sum_{t=1}^{T} \| \widehat{\boldsymbol{\Phi}}_{mt} - \boldsymbol{\Phi}_{mt} \|_2/T \right) .
		\end{aligned}
	\end{equation}
 
	Table \ref{Table 5} presents the mean and standard deviations of $\mathcal{D}(\widehat{\boldsymbol{\Psi}}_m,\boldsymbol{\Psi}_m)$
	and $\mathcal{D}(\widehat{\boldsymbol{\Phi}}_m,\boldsymbol{\Phi}_m)$. The estimated errors decrease when the dimensions $(p,q)$ or the sample size $T$ increase. 
	\begin{table}[!htbp]
		\tabcolsep  = 0.15cm
		\begin{center}
			\caption{Means and standard deviations of $\mathcal{D}(\widehat{\boldsymbol{\Phi}}_m,\boldsymbol{\Phi}_m)$ and $\mathcal{D}(\widehat{\boldsymbol{\Psi}}_m,\boldsymbol{\Psi}_m)$ for group $m = 1,2,3$, over 200 simulation runs.}
		\begin{tabular}{ccccccccc}
			\hline\rule{0pt}{15pt}
		    $ n $ & $ p $ & $ T $ & $ \mathcal{D}(\widehat{\boldsymbol{\Phi}}_1,\boldsymbol{\Phi}_1) $ & $ \mathcal{D}(\widehat{\boldsymbol{\Psi}}_1,\boldsymbol{\Psi}_1) $ & $ \mathcal{D}(\widehat{\boldsymbol{\Phi}}_2,\boldsymbol{\Phi}_2) $ & $ \mathcal{D}(\widehat{\boldsymbol{\Psi}}_2,\boldsymbol{\Psi}_2) $ & $ \mathcal{D}(\widehat{\boldsymbol{\Phi}}_3,\boldsymbol{\Phi}_3) $ & $ \mathcal{D}(\widehat{\boldsymbol{\Psi}}_3,\boldsymbol{\Psi}_3) $ \\
			\hline
			10 & 10 & 50 & 0.97(0.60) &0.75(0.43) & 1.12(0.95)&0.97(0.80)& 2.06(0.98)& 1.45(0.68) \\
			   &    & 200 & 1.82(1.35) & 1.32(1.03) & 0.91(1.21)& 0.77(1.06)& 0.46(0.21)& 0.33(0.11)\\
			   &    & 800 & 0.22(0.01) &0.21(0.01) & 0.27(0.00)& 0.24(0.00)& 0.32(0.02)& 0.22(0.01)\\
			   &    & 2000 & 0.20(0.01) &0.21(0.01) & 0.29(0.01)& 0.29(0.00)& 0.20(0.00)& 0.21(0.00)\\
			\hline
			20 & 20 & 50 & 1.11(1.15) &0.97(1.00) &1.28(1.00)&1.01(0.84)&1.01(0.82)&0.80(0.65) \\
			   &    & 200 & 0.18(0.38)&0.22(0.31)&0.27(0.32)&0.25(0.16)&0.29(0.60)&0.27(0.39)\\
			   &    & 800 & 0.10(0.00)&0.10(0.00)&0.09(0.00)&0.12(0.00)&0.08(0.00)&0.10(0.00)\\
			   &    & 2000 & 0.09(0.00)&0.10(0.00)&0.08(0.00)&0.09(0.00)&0.11(0.00)&0.10(0.00)\\
			\hline
			40 & 40 & 50 &  0.86(0.97)&0.76(0.80) &0.97(0.96)&0.76(0.75)&1.32(1.08)&1.06(0.92) \\
			   &    & 200 &  0.16(0.14)&0.19(0.10)&0.07(0.08)&0.11(0.08)&0.05(0.00)&0.07(0.00)\\
			   &    & 800 &  0.04(0.00)&0.05(0.00)&0.05(0.00)&0.05(0.00)&0.04(0.00)&0.05(0.00)\\
			   &    & 2000 &  0.04(0.00)&0.06(0.00)&0.05(0.00)&0.05(0.00)&0.04(0.00)&0.06(0.00)\\
			\hline \rule{0pt}{15pt}
		\end{tabular}
		\label{Table 5}
	\end{center}
	\end{table} 
 
\section{Real Data} \label{S6}

In this section, we collect multi-industry stock indices data from Wind database. It encompasses eight daily indicators ($q=8$) within the stock market, shared among twenty stocks ($n=20$) from each of the ten industries ($M=10$) from June 2021 to June 2023 for five hundred trading days ($T=500$). The industries include Energy, Materials, Industrials, Consumer Discretionary, Consumer Staples, Health Care, Financials, Information Technology, Utilities and Real Estate. The stock selection criterion for each industry involves excluding stocks with incomplete data and selecting 
 the top 20 stocks based on their total market capitalization. The indicators are listed in the following order: Return, Trading Volume, Trading Value, Price-Earnings Ratio (P/E Ratio), Price-to-Book Ratio (P/B Ratio), Volume Ratio, Historical Rolling Volatility and Turnover Rate.  Missing values processing and stationarity tests are carried out for each original time series. Non-stationary time series are transformed by taking the first difference to make them stationary. Finally, the processed sequences are standardized. The visualization results, exemplified using Energy data, are presented in Figure \ref{original time series}.

\begin{figure}[!htbp]
	\centering
	\includegraphics[width=0.98\linewidth]{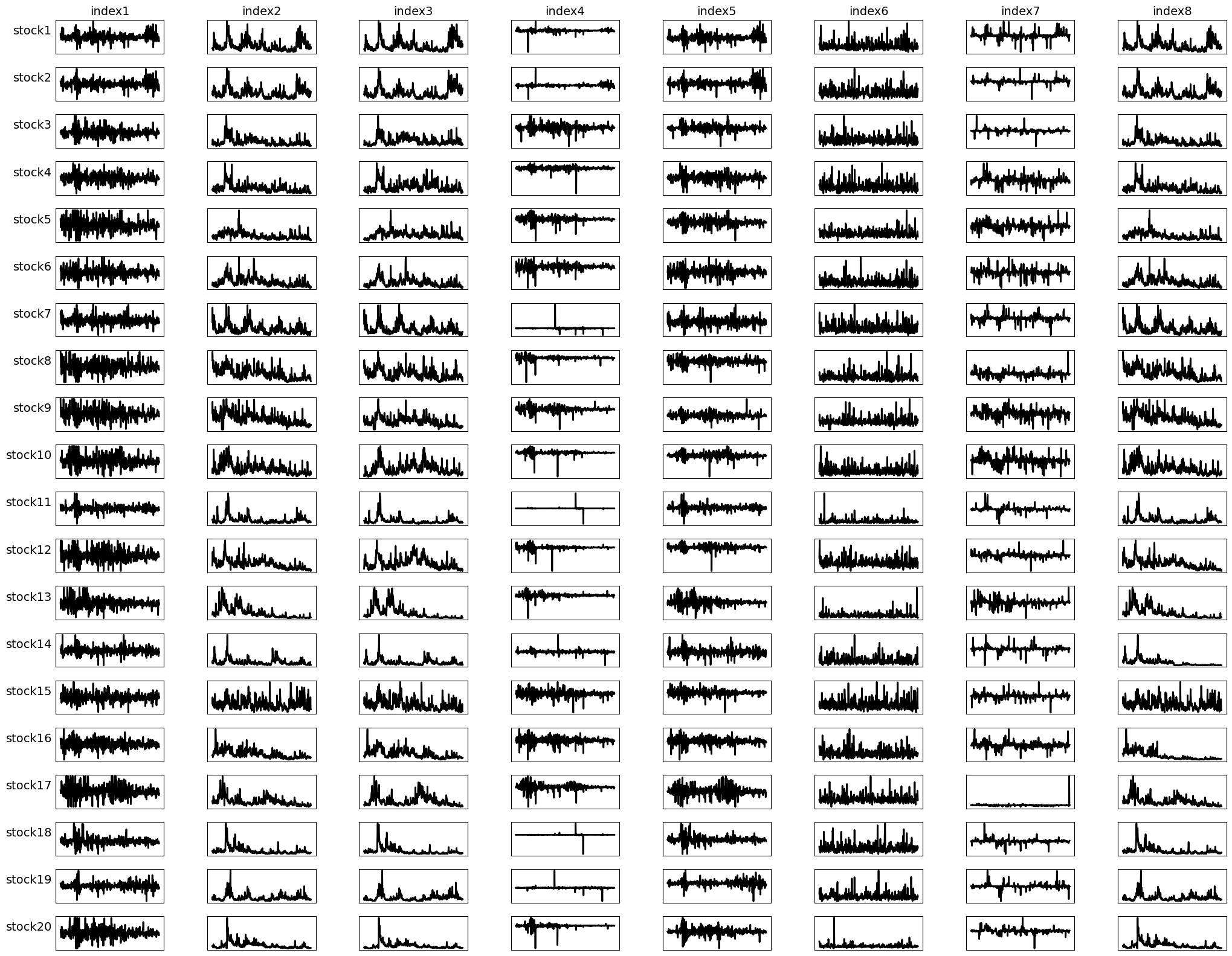}
	\caption{Time series plots of stock indicators of energy sector (after data transformation)}
	\label{original time series}
\end{figure}

We begin with an analysis of  within-correlations and between-correlations. The former is quantified as the average correlation among the eight indicators for different stocks within the same industry, whereas the latter assesses the average correlation spanning various industries. To visualize these relationships, we utilize a heat map, as shown 
in  Figure \ref{within correlations and between correlations of standardized data}. Here, the diagonal elements correspond to within-correlations, while the off-diagonal elements indicate between-correlations. It is noteworthy that both correlations exhibit positive and substantial trends across all industries. Remarkably, the within-correlations consistently surpasses the between-correlations across all industries. For instance, among the Financials, Information Technology, and Utilities, the respective within-correlations stand at 0.223, 0.116, and 0.174, while the between-correlations correspond to 0.068, 0.069, and 0.052.
 This discrepancy suggests the possible existence of local/industry-specific factors within the same industry.

\begin{figure}[!htbp]
	\centering
 \includegraphics[width=0.95\linewidth]{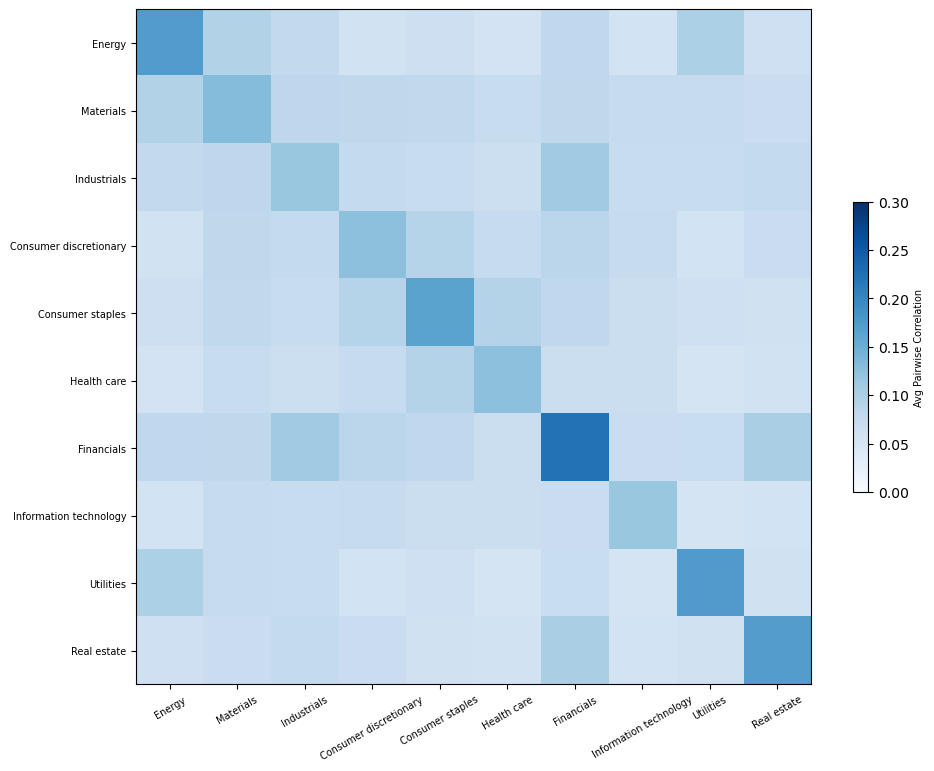}
	\caption{Average pairwise correlations of standardized data}
	\label{within correlations and between correlations of standardized data}
\end{figure}
Subsequently, we apply the multilevel matrix factor model to explore the inherent correlation structure of the data. To begin this exploration, it is imperative to utilize the eigenvalue-ratio method for estimating the latent dimensions of global and local factors within each industry. 
Detailed graphical representations of this process can be found in Figure \ref{selection of factors number of energy}, which shows the logarithms and ratios of eigenvalues for the energy data. 
For the number of global factors, many industry data show that $k_1 = 1 , k_2 = 2$, we use $k_1 = k_2 = 2$ here for illustration. 
The determination of the number of local factors in each industry is based on the outcomes of the eigenvalue-ratio method's estimation. 
A comprehensive summary of factor selection for each industry is presented in Table \ref{selection of factor number}. The plots for the remaining industries (excluding Energy) are included in the Appendix 2.

\begin{figure}[!htbp]
	\centering
        \includegraphics[width=0.95\linewidth]{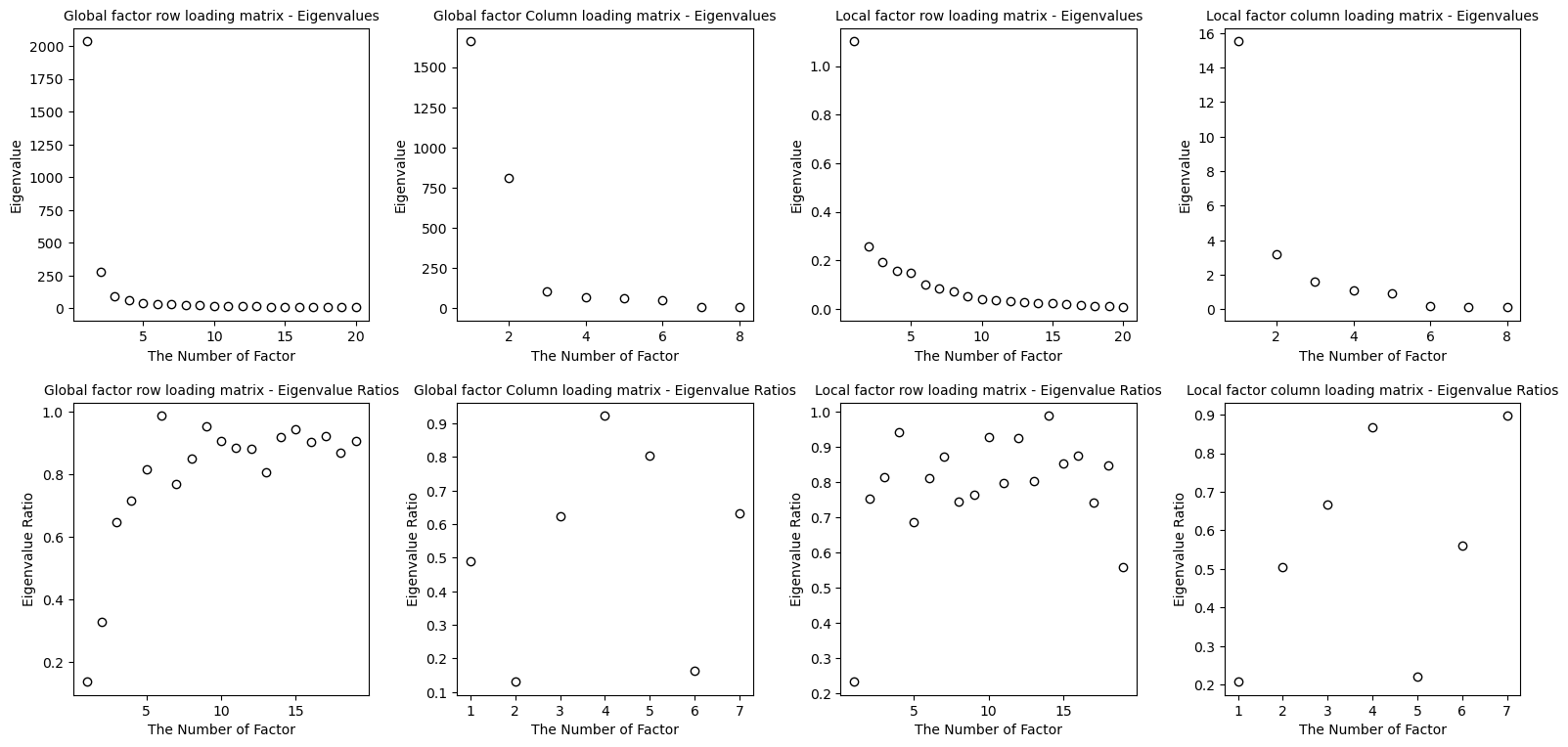}
	\caption{Energy series: eigenvalues and eigenvalues ratio for determining the number of global and local factors }
	\label{selection of factors number of energy}

\end{figure}

\begin{table}[!htbp]
    \centering
    \caption{Selection of Global and Local Factor Number}
    \begin{adjustbox}{width=0.98\textwidth}
        \begin{tabular}{llllll}
            \toprule
            Industry & Global Factor & Local Factor & Industry & Global Factor & Local Factor \\
            \midrule
            Energy & (1,2) & (1,1) & Health care & (1,2) & (1,1) \\
            Materials & (1,2) & (2,1) & Financials & (1,2) & (2,1) \\
            Industrials & (2,2) & (1,1) & Information technology & (2,2) & (1,1) \\
            Consumer discretionary & (1,2) & (4,1) & Utilities & (1,2) & (1,1) \\
            Consumer staples & (1,2) & (1,1) & Real estate & (1,2) & (1,1) \\
            \bottomrule
        \end{tabular}
    \end{adjustbox}
    \label{selection of factor number}
\end{table}

In the subsequent analysis, we investigate the within and between correlations after concentrating out global and local factors, respectively. Figure \ref{Average pairwise correlations}(a) presents the results derived from residual data obtained through a regression of the data solely concentrate out  the global factors. 
In contrast to Figure \ref{within correlations and between correlations of standardized data}, we observe a significant decrease in between correlations across all industries, indicating that the global factors effectively capture the market-wide co-movement of individual stocks.
However, it is noteworthy that the within correlations for the Materials, Industrials, Consumer Discretionary, Utilities, and Real Estate remain non-negligible, suggesting that these co-movements may be attributed to local factors. To delve deeper, we proceed to eliminate the influence of local factors, resulting in residuals that primarily represent idiosyncratic components.
Figure \ref{Average pairwise correlations}(b) illustrates that both within and between correlations become nearly negligible, underscoring the substantial role played by local-specific factors in driving higher within correlations among these industries.

\begin{figure}[!htbp]
    \centering
    \begin{subfigure}[t]{0.45\textwidth}
        \centering
        \includegraphics[width=\textwidth]{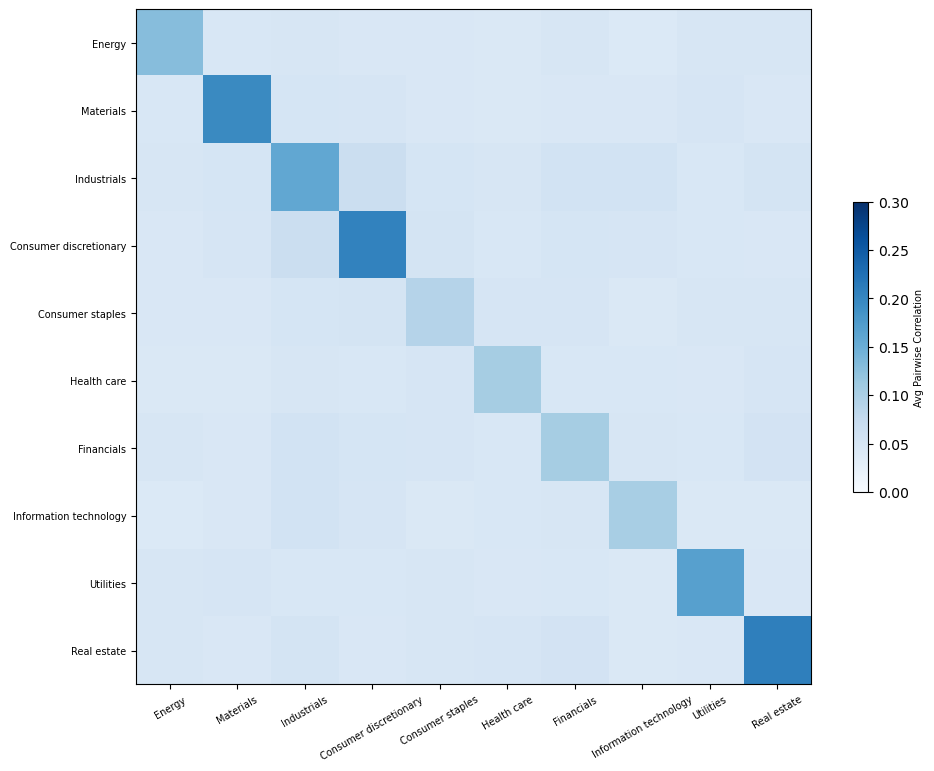}
        \caption{Average pairwise correlations of residuals after concentrating out global factors}
        \label{within correlations and between correlations concentrating out global factors}
    \end{subfigure}%
    \hspace{0.05\textwidth} 
    \begin{subfigure}[t]{0.45\textwidth}
        \centering
        \includegraphics[width=\textwidth]{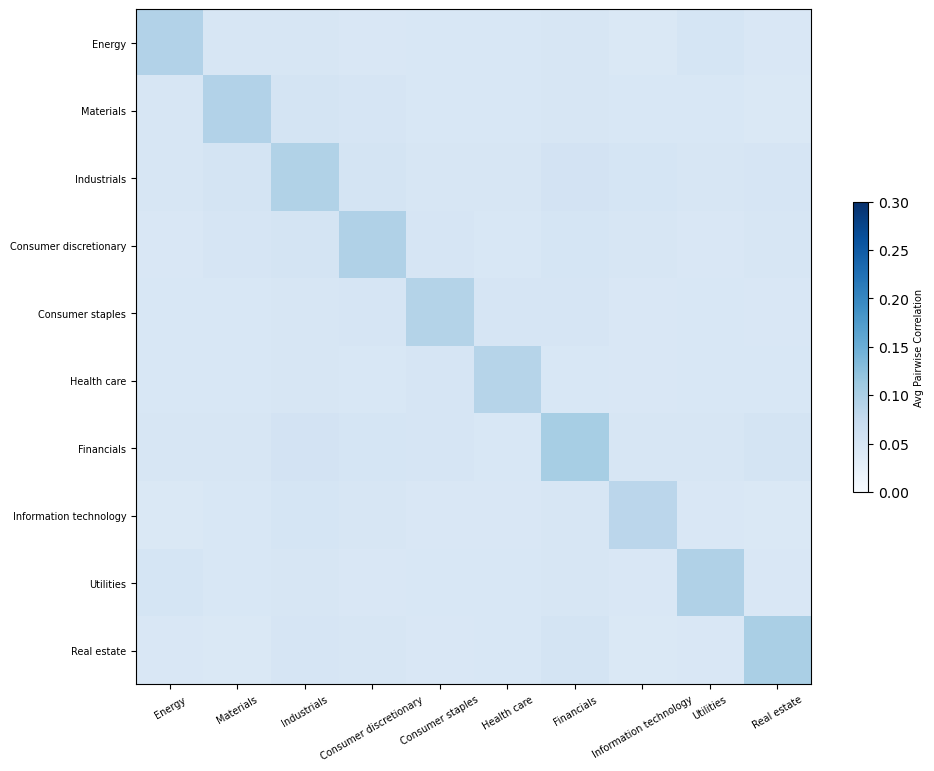}
        \caption{Average pairwise correlations of residuals after concentrating out global and local factors}
        \label{within correlations and between correlations concentrating out global and local factors}
    \end{subfigure}
    \caption{Average pairwise correlations}
    \label{Average pairwise correlations}
\end{figure}

Next, the latent structure of the China Stock Markets can be comprehended through an examination of the estimated global and local factor loading matrices. 
Within the context of our proposed model, our initial emphasis is placed on the loading matrices associated with global factors, using Energy data as an illustrative example.
 Table \ref{Row loading matrix of global factor within energy data} and Table \ref{Column loading matrix of global factor within energy data} show the estimated loading matrices undergo rotation to maximize their variance, following appropriate scaling (30-fold) to enhance visual clarity.
  It can be seen that potential groupings are discernible within the tables, as depicted by shaded regions. Regarding the stocks, a demarcation into two groups becomes evident. Group 1 encompasses stocks S5, S8-S10, S12-S13, and S16-S17, while Group 2 includes stocks S1-S3, S7, S11, S15, and S19.
   It is worth highlighting that stocks S4, S6, S14, S18, and S20 exhibit notably similar factor loading matrix weights, making them amenable to allocation within either Group 1 or Group 2. Furthermore, global column factors discern two primary groups. The first group comprises Trading Volume, Trading Value and Turnover Rate.
    The second group consists of Return, Price-Earnings Ratio (P/E Ratio) and Price-to-Book Ratio (P/B Ratio). Similarly, Volume Ratio and Historical Rolling Volatility exhibit comparable weights and can be flexibly assigned to either the first or the second group. This classification aligns harmoniously with the conventional definitions and distinctions associated with stock market indicators.

\begin{table}[!htbp]
    \centering
    \caption{Row loading matrix of global factor within energy data}
    \begin{tabular}{lllllllllll}
    \toprule
    \textbf{Global Row Factors} & S1 & S2 & S3 & S4 & S5 & S6 & S7 & S8 & S9 & S10 \\
    \midrule
    Factor 1 & 3 & 4 & -3 & \cellcolor{gray!45}-6 & \cellcolor{gray!20}-6 & \cellcolor{gray!45}-5 & -1 & \cellcolor{gray!20}-9 & \cellcolor{gray!20}-9 & \cellcolor{gray!20}-8 \\
    Factor 2 & \cellcolor{gray!20} -13 & \cellcolor{gray!20}-14 & \cellcolor{gray!20}-8 & \cellcolor{gray!45}-5 & -2 & \cellcolor{gray!45}-5 & \cellcolor{gray!20}-7 & -1 & -1 & -1 \\
    \midrule
    & S11 & S12 & S13 & S14 & S15 & S16 & S17 & S18 & S19 & S20 \\
    \midrule
    Factor 1 & -1 & \cellcolor{gray!20}-10 & \cellcolor{gray!20}-12 & \cellcolor{gray!45}-4 & -2 & \cellcolor{gray!20}-10 & \cellcolor{gray!20}-12 & \cellcolor{gray!45}-5 & 2 & \cellcolor{gray!45}-5 \\
    Factor 2 & \cellcolor{gray!20}-10 & 1 & 3 & \cellcolor{gray!45}-5 & \cellcolor{gray!20}-6 & -1 & 6 & \cellcolor{gray!45}-5 & \cellcolor{gray!20}-10 & \cellcolor{gray!45}-4 \\
    \bottomrule
    \end{tabular}
\label{Row loading matrix of global factor within energy data}
\end{table}

\begin{table}[!htbp]
    \centering
    \caption{Column loading matrix of global factor within energy data}
    \begin{tabular}{lllllllllll}
    \toprule
    \textbf{Global Column Factors} & I1 & I2 & I3 & I4 & I5 & I6 & I7 & I8 \\
    \midrule
    Factor 1 & 1 & \cellcolor{gray!20}18 & \cellcolor{gray!20}15 & 0 & 0 & \cellcolor{gray!45}4 & \cellcolor{gray!45}5 & \cellcolor{gray!20}18 \\
    Factor 2 & \cellcolor{gray!20}-20 & 0 & 0 & \cellcolor{gray!20}-13 & \cellcolor{gray!20}-18 & \cellcolor{gray!45}-1 & \cellcolor{gray!45}3 & 0 \\
    \bottomrule
    \end{tabular}
\label{Column loading matrix of global factor within energy data}
\end{table}

Subsequently, we analyze the loading matrices of local factors and take the data of Consumer Discretionary for an example. 
Table \ref{Row loading matrix of local factor within consumer discretionary data} presents the factor loading matrix for stocks. Group 1: Comprising stocks S2-S5, S10-S11, and S17-S18. Group 2: Consisting of stocks S13-S14 and S20. Group 3: Encompassing stocks S15 and S19.
 Group 4: Including stocks S1, S7-S9, S12, and S16. Notably, only stock S6 exhibits identical factor loading matrix weights within both Group 2 and Group 4, allowing it to be classified within both groups. Within the Consumer Discretionary data, a single local column factor is extracted, 
 and its loading matrix is depicted in Table \ref{Column loading matrix of local factor within consumer discretionary data}.

\begin{table}[!htbp]
    \centering
    \caption{Row loading matrix of local factor within consumer discretionary data}
    \begin{tabular}{lllllllllll}
    \toprule
    \textbf{Local Row Factors} & S1 & S2 & S3 & S4 & S5 & S6 & S7 & S8 & S9 & S10 \\
    \midrule
    Factor 1 &  1.0 & \cellcolor{gray!20}-12.0 & \cellcolor{gray!20}-14.0 & \cellcolor{gray!20}-9.0 & \cellcolor{gray!20} -9.0 & -4.0 & 0.0 & 1.0 & 3.0 & \cellcolor{gray!20}-10.0 \\
    Factor 2 &  9.0 & 3.0 & -3.0 & 6.0 & -1.0 & \cellcolor{gray!45}-10.0 & -1.0 & -4.0 & -5.0 & 3.0 \\
    Factor 3 &-1.0 & -1.0 & 1.0 & -3.0 & -3.0 & 2.0 & -1.0 & 7.0 & 9.0 & 1.0 \\
    Factor 4 & \cellcolor{gray!20}10.0 & -0.0 & -1.0 & 5.0 & 2.0 & \cellcolor{gray!45}10.0 & \cellcolor{gray!20}11.0 & \cellcolor{gray!20}9.0 & \cellcolor{gray!20}11.0 & -0.0 \\
    
    \midrule
    & S11 & S12 & S13 & S14 & S15 & S16 & S17 & S18 & S19 & S20 \\
    \midrule
    Factor 1 & \cellcolor{gray!20}-5.0 & -0.0 & -1.0 & 4.0 & -2.0 & -5.0 & \cellcolor{gray!20}-10.0 & \cellcolor{gray!20}-9.0 & 2.0 & -3.0 \\
    Factor 2 & 3.0 & 2.0 & \cellcolor{gray!20}7.0 & \cellcolor{gray!20}20.0 & 6.0 & 3.0 & 0.0 & -4.0 & -0.0 & \cellcolor{gray!20}8.0 \\
    Factor 3 & -2.0 & 1.0 & 0.0 & -2.0 & \cellcolor{gray!20}25.0 & -1.0 & 1.0 & 4.0 & \cellcolor{gray!20}9.0 & 2.0 \\
    Factor 4 & 2.0 & \cellcolor{gray!20}12.0 & -1.0 & 4.0 & -6.0 & \cellcolor{gray!20}9.0 & -6.0 & 0.0 & 4.0 & -1.0 \\
    \bottomrule
    \end{tabular}
    \label{Row loading matrix of local factor within consumer discretionary data}
\end{table}

\begin{table}[!htbp]
    \centering
    \caption{Column loading matrix of local factor within consumer discretionary data}
    \begin{tabular}{lllllllllll}
    \toprule
    \textbf{Local Column Factors} & I1 & I2 & I3 & I4 & I5 & I6 & I7 & I8 \\
    \midrule
    Factor 1 & \cellcolor{gray!20}3 & \cellcolor{gray!20}16 & \cellcolor{gray!20}15 & \cellcolor{gray!20}1 & \cellcolor{gray!20}2 & \cellcolor{gray!20}10 & \cellcolor{gray!20}7 & \cellcolor{gray!20}16 \\
    \bottomrule
    \end{tabular}
\label{Column loading matrix of local factor within consumer discretionary data}
\end{table}

Finally, we incorporate a well-established evaluation method commonly used in factor models. 
We define out of sample RSS/TSS on a testing set of size $T$ as
$$
\text{RSS/TSS} \triangleq 1-\frac{\sum_{t=1}^T\left\|\boldsymbol{X}_t-\widehat{\boldsymbol{X}}_t\right\|_F^2}{\sum_{t=1}^T\left\|\boldsymbol{X}_t-\overline{\boldsymbol{X}}\right\|_F^2} \text {, }
$$
where $\overline{\boldsymbol{X}}=\frac{1}{T} \sum_{t=1}^T \boldsymbol{X}_t$ and $\widehat{\boldsymbol{X}}_t$ represents the fitted values computed through the model.
 Based on the preceding factor selection results,
   we can obtain the fitting outcomes for each industry and the entire market, alongside the model's number of factors and parameters, which are presented in Table \ref{result}.

\begin{table}[!htbp]
    \centering
    \caption{Model Evaluation Results}
    \begin{tabular}{lcccccc}
        \toprule
        Industry & Local Factor & Global Factor & RSS/TSS & \#factors & \#parameters \\
        \midrule
        Energy & (1,1) & (2,2) & 0.5904959 & 5 & 84 \\
        Materials & (2,1) & (2,2) & 0.6795855 & 6 & 104 \\
        Industrials & (1,1) & (2,2) & 0.7033221 & 5 & 84 \\
        Consumer discretionary & (4,1) & (2,2) & 0.6389581 & 8 & 144 \\
        Consumer staples & (1,1) & (2,2) & 0.6318063 & 5 & 84 \\
        Health care & (1,1) & (2,2) & 0.7057251 & 6 & 104 \\
        Financials & (2,1) & (2,2) & 0.4950737 & 5 & 84 \\
        Information technology & (1,1) & (2,2) & 0.7102152 & 5 & 84 \\
        Utilities & (1,1) & (2,2) & 0.6519537 & 5 & 84 \\
        Real estate & (1,1) & (2,2) & 0.6558825 & 5 & 84 \\
        Total & -- & -- & 0.64630185 & 55 & 940 \\
        \bottomrule
    \end{tabular}
    \label{result}
\end{table}

Table \ref{compare with wang} presents a direct comparison between matrix factor models and multilevel matrix factor models with varying sizes and numbers of factors.
 The results clearly demonstrate that our proposed model employs significantly fewer parameters in the loading matrices while achieving better performance.
 This finding further validates the efficiency of our model.
\begin{table}[!htbp]
    \centering
    \caption{Model Evaluation Results}
    \begin{adjustbox}{width=0.98\textwidth}
    \begin{tabular}{lcccccc}
        \toprule
        Model & Global Factor & Local Factor & RSS/TSS & \#factors & \#parameters \\
        \midrule
        \multirow{11}{*}{Multilevel Matrix Factor Model} & (1,1) & (1,1) & 0.8364629 & 11 & 560 \\
         & (2,2) & (2,2) & 0.636141 & 44 & 1120 \\
         & (2,2) & (3,2) & 0.6111359 & 64 & 1320 \\
         & (2,2) & (4,2) & 0.5856743 & 84 & 1520 \\
         & (2,2) & (5,2) & 0.5614321 & 104 & 1720 \\
         & (3,2) & (3,2) & 0.5826328 & 66 & 1520 \\
         & (3,2) & (4,2) & 0.561714 & 86 & 1720 \\
         & (4,2) & (3,2) & 0.5558335 & 68 & 1720 \\
         & (4,2) & (4,2) & 0.5405282 & 88 & 1920 \\
         & (3,3) & (3,3) & 0.5285275 & 99 & 1680 \\
         & (4,3) & (4,3) & 0.4770679 & 132 & 2080 \\
        \cmidrule{1-6}
        \multirow{4}{*}{Matrix Factor Model} & \multicolumn{2}{c}{(10,2)} & 0.7005498 & 20 & 2016 \\
         & \multicolumn{2}{c}{(20,2)} & 0.6277033 & 40 & 4016 \\
         & \multicolumn{2}{c}{(10,3)} & 0.6692189 & 30 & 2024 \\
         & \multicolumn{2}{c}{(20,3)} & 0.5850474 & 60 & 4024 \\
        \bottomrule
    \end{tabular}
    \end{adjustbox}
    \label{compare with wang}
\end{table}

\section{Conclusion}

In this article, we introduce a multilevel factor model for high-dimensional matrix-variate time series data. Our model comprises two types of factors:
global factors that affect all the data and the local factors which affect the data within particular groups. The estimation procedures make use of the 
covariance structure between matrix time series. The asymptotic properties for the factor loadings and the signal parts have been established. Simulation results and real data analysis show the 
performance of the proposed procedure with finite sample sizes. 

There are several interesting topics for further research. Firstly, this paper assumes no correlation between the local factors. Investigating how to estimate the multilevel matrix-variate factor model when the nonzero correlations
between local factors are permitted is an important direction. Secondly, we can extend our exploration to incorpate a multi-level structure for tensor factor models. These areas will be the foucs of our ongoing research efforts.

\vspace{0.2in}




\bibliographystyle{apalike}
\bibliography{Multilevel_Matrix_Factor_Model}

\begin{thebibliography}{}

\bibitem[Ando and Bai, 2016]{ando2016panel}
Ando, T. and Bai, J. (2016).
\newblock Panel data models with grouped factor structure under unknown group membership.
\newblock {\em Journal of Applied Econometrics}, 31(1):163--191.

\bibitem[Ando and Bai, 2017]{doi:10.1080/01621459.2016.1195743}
Ando, T. and Bai, J. (2017).
\newblock Clustering huge number of financial time series: A panel data approach with high-dimensional predictors and factor structures.
\newblock {\em Journal of the American Statistical Association}, 112(519):1182--1198.

\bibitem[Andreou et~al., 2019]{andreou2019inference}
Andreou, E., Gagliardini, P., Ghysels, E., and Rubin, M. (2019).
\newblock Inference in group factor models with an application to mixed-frequency data.
\newblock {\em Econometrica}, 87(4):1267--1305.

\bibitem[Chang et~al., 2015]{chang2015high}
Chang, J., Guo, B., and Yao, Q. (2015).
\newblock High dimensional stochastic regression with latent factors, endogeneity and nonlinearity.
\newblock {\em Journal of Econometrics}, 189(2):297--312.

\bibitem[Chen et~al., 2023]{Chen2023TimeVaryingMF}
Chen, B., Chen, E.~Y., and Chen, R. (2023).
\newblock Time-varying matrix factor model.

\bibitem[Chen and Fan, 2021]{chen2021statistical}
Chen, E.~Y. and Fan, J. (2021).
\newblock Statistical inference for high-dimensional matrix-variate factor models.
\newblock {\em Journal of the American Statistical Association}, pages 1--18.

\bibitem[Chen et~al., 2020]{doi:10.1080/01621459.2019.1584899}
Chen, E.~Y., Tsay, R.~S., and Chen, R. (2020).
\newblock Constrained factor models for high-dimensional matrix-variate time series.
\newblock {\em Journal of the American Statistical Association}, 115(530):775--793.

\bibitem[Choi et~al., 2018]{choi2018multilevel}
Choi, I., Kim, D., Kim, Y.~J., and Kwark, N.-S. (2018).
\newblock A multilevel factor model: Identification, asymptotic theory and applications.
\newblock {\em Journal of Applied Econometrics}, 33(3):355--377.

\bibitem[Choi et~al., 2023]{CHOI202322}
Choi, I., Lin, R., and Shin, Y. (2023).
\newblock Canonical correlation-based model selection for the multilevel factors.
\newblock {\em Journal of Econometrics}, 233(1):22--44.

\bibitem[Han, 2021]{han2021shrinkage}
Han, X. (2021).
\newblock Shrinkage estimation of factor models with global and group-specific factors.
\newblock {\em Journal of Business \& Economic Statistics}, 39(1):1--17.

\bibitem[He et~al., 2023]{he2023matrix}
He, Y., Kong, X., Yu, L., Zhang, X., and Zhao, C. (2023).
\newblock Matrix factor analysis: From least squares to iterative projection.
\newblock {\em Journal of Business \& Economic Statistics}, (just-accepted):1--26.

\bibitem[Lam and Yao, 2012]{10.1214/12-AOS970}
Lam, C. and Yao, Q. (2012).
\newblock {Factor modeling for high-dimensional time series: Inference for the number of factors}.
\newblock {\em The Annals of Statistics}, 40(2):694 -- 726.

\bibitem[Lam et~al., 2011]{10.1093/biomet/asr048}
Lam, C., Yao, Q., and Bathia, N. (2011).
\newblock {Estimation of latent factors for high-dimensional time series}.
\newblock {\em Biometrika}, 98(4):901--918.

\bibitem[Liu and Chen, 2022]{liu2022identification}
Liu, X. and Chen, E.~Y. (2022).
\newblock Identification and estimation of threshold matrix-variate factor models.
\newblock {\em Scandinavian Journal of Statistics}, 49(3):1383--1417.

\bibitem[Merikoski and Kumar, 2004]{merikoski2004inequalities}
Merikoski, J.~K. and Kumar, R. (2004).
\newblock Inequalities for spreads of matrix sums and products.
\newblock {\em Applied Mathematics E-Notes}, 4(150-159):9.

\bibitem[Moench et~al., 2013]{moench2013dynamic}
Moench, E., Ng, S., and Potter, S. (2013).
\newblock Dynamic hierarchical factor models.
\newblock {\em Review of Economics and Statistics}, 95(5):1811--1817.

\bibitem[Pan and Yao, 2008]{pan2008modelling}
Pan, J. and Yao, Q. (2008).
\newblock Modelling multiple time series via common factors.
\newblock {\em Biometrika}, 95(2):365--379.

\bibitem[Wang et~al., 2019]{wang2019factor}
Wang, D., Liu, X., and Chen, R. (2019).
\newblock Factor models for matrix-valued high-dimensional time series.
\newblock {\em Journal of Econometrics}, 208(1):231--248.

\bibitem[Wang, 2008]{wang2008large}
Wang, P. (2008).
\newblock Large dimensional factor models with a multi-level factor structure: identification, estimation and inference.
\newblock {\em Unpublished manuscript, New York University}.

\bibitem[Yu et~al., 2022]{yu2022projected}
Yu, L., He, Y., Kong, X., and Zhang, X. (2022).
\newblock Projected estimation for large-dimensional matrix factor models.
\newblock {\em Journal of Econometrics}, 229(1):201--217.

\end{thebibliography}

\newpage
\noindent
{\Large \bf Appendix 1: Proofs}

The following notations are used in our proofs.
\begin{eqnarray*}
		\boldsymbol{\Omega}_{x, m n, j_1 j_2} &=&\frac{1}{T} \sum_{t=1}^{T} \operatorname{Cov}\left(\boldsymbol{x}_{\cdot,j_1,mt}, \boldsymbol{x}_{\cdot,j_2,nt}\right) ,\\
		\boldsymbol{\Omega}_{g, m n, j_1 j_2} &=&\frac{1}{T} \sum_{t=1}^{T} \operatorname{Cov}\left(\boldsymbol{R}_m \boldsymbol{G}_t \boldsymbol{c}_{\cdot,j_1,m}, \boldsymbol{R}_n \boldsymbol{G}_{t} \boldsymbol{c}_{\cdot,j_2,n}\right), \\
		\boldsymbol{\Omega}_{g c, m n, j_1 j_2} &=&\frac{1}{T} \sum_{t=1}^{T} \operatorname{Cov}\left(\boldsymbol{G}_t \boldsymbol{c}_{\cdot,j_1,m}, \boldsymbol{G}_{t} \boldsymbol{c}_{\cdot,j_2,n}\right) ,\\
		\widehat{\boldsymbol{\Omega}}_{x, m n, j_1 j_2} &=&\frac{1}{T} \sum_{t=1}^{T} \boldsymbol{x}_{\cdot,j_1,mt}\boldsymbol{x}_{\cdot,j_2,nt}^{\prime} ,\\
		\widehat{\boldsymbol{\Omega}}_{g, m n, j_1 j_2} &=& \frac{1}{T} \sum_{t=1}^{T} \boldsymbol{R}_m \boldsymbol{G}_t \boldsymbol{c}_{\cdot,j_1,m} \boldsymbol{c}_{\cdot,j_2,n}^{\prime} \boldsymbol{G}_{t}^{\prime} \boldsymbol{R}_n^{\prime}   ,\\
		\widehat{\boldsymbol{\Omega}}_{f, m n, j_1 j_2} &=&\frac{1}{T} \sum_{t=1}^{T} \boldsymbol{\Gamma}_m \boldsymbol{F}_{mt} \boldsymbol{\lambda}_{\cdot,j_1,m} \boldsymbol{\lambda}_{\cdot,j_2,n}^{\prime} \boldsymbol{F}_{nt}^{\prime} \boldsymbol{\Gamma}_n^{\prime}   ,\\
		\widehat{\boldsymbol{\Omega}}_{g c, m n, j_1 j_2} &=&\frac{1}{T} \sum_{t=1}^{T} \boldsymbol{G}_t \boldsymbol{c}_{\cdot,j_1,m} \boldsymbol{c}_{\cdot,j_2,n}^{\prime} \boldsymbol{G}_{t}^{\prime} ,\\
		\widehat{\boldsymbol{\Omega}}_{g f, m n, j_1 j_2} &=& \frac{1}{T} \sum_{t=1}^{T}
		\boldsymbol{R}_m\boldsymbol{G}_t \boldsymbol{c}_{\cdot , j_1, m} \boldsymbol{\lambda}_{\cdot,j_2,n}^{\prime}  \boldsymbol{F}_{nt}^{\prime} \boldsymbol{\Gamma}_n^{\prime} , \\
		\widehat{\boldsymbol{\Omega}}_{g c, m n, j_1 j_2} &=&\frac{1}{T} \sum_{t=1}^{T} \boldsymbol{R}_m \boldsymbol{G}_t \boldsymbol{c}_{\cdot , j_1, m}  \boldsymbol{e}_{\cdot,j_2,nt}^{\prime} ,\\
		\widehat{\boldsymbol{\Omega}}_{f e, m n, j_1 j_2} &=&\frac{1}{T} \sum_{t=1}^{T} \boldsymbol{\Gamma}_m \boldsymbol{F}_{mt} \boldsymbol{\lambda}_{\cdot,j_1,m} \boldsymbol{e}_{\cdot,j_2,nt}^{\prime} ,\\
		\widehat{\boldsymbol{\Omega}}_{f g, m n, j_1 j_2} &=&\frac{1}{T} \sum_{t=1}^{T} \boldsymbol{\Gamma}_m \boldsymbol{F}_{mt} \boldsymbol{\lambda}_{\cdot,j_1,m} \boldsymbol{c}_{\cdot,j_2,n}^{\prime} \boldsymbol{G}_{t}^{\prime} \boldsymbol{R}_{m}^{\prime},\\
		\widehat{\boldsymbol{\Omega}}_{e g, m n, j_1 j_2} &=&\frac{1}{T} \sum_{t=1}^{T} 
		\boldsymbol{e}_{\cdot,j_1,mt} \boldsymbol{c}_{\cdot , j_2, n}^{\prime} \boldsymbol{G}_t^{\prime} \boldsymbol{R}_n^{\prime}     ,\\
		\widehat{\boldsymbol{\Omega}}_{e f, m n, j_1 j_2} &=&\frac{1}{T} \sum_{t=1}^{T} 
		\boldsymbol{e}_{\cdot,j_1,mt} \boldsymbol{\lambda}_{\cdot , j_2, n}^{\prime} \boldsymbol{F}_{nt}^{\prime} \boldsymbol{\Gamma}_n^{\prime}     ,\\
		\widehat{\boldsymbol{\Omega}}_{e, m n, j_1 j_2} &=&\frac{1}{T} \sum_{t=1}^{T} \boldsymbol{e}_{\cdot,j_1,mt} \boldsymbol{e}_{\cdot,j_2,nt}^{\prime}.
\end{eqnarray*}

\begin{lemma}
	Let the ij-th entry of $\boldsymbol{G}_{t}$\ and\ $\boldsymbol{F}_{nt}$\ be\ $G_{t,ij}$\ and\ $F_{nt,ij}$,\ respectively.\ Under Conditions 1-2 and 6,\ it\ holds\ that
	\begin{eqnarray*}
			 \left|\frac{1}{T} \sum_{t=1}^{T}\left(G_{t, i j} G_{t, k l}-\operatorname{Cov}\left(G_{t, i j}, G_{t, k l}\right)\right)\right|&=&O_p\left(T^{-1 / 2}\right),\\
			 \left|\frac{1}{T} \sum_{t=1}^{T}\left(G_{t, i j} F_{mt, k l} \right) \right|&=&O_p\left(T^{-1 / 2}\right),\\
			 \left|\frac{1}{T} \sum_{t=1}^{T}\left(F_{m t, i j} F_{n t, k l} \right) \right|&=& O_p\left(T^{-1 / 2}\right).
	\end{eqnarray*}
\end{lemma}
\begin{proof}
	For any $i,k = 1,...,k_1; j,l = 1,...,k_2$, by Cauchy-Schwarz inequality and Davydov inequality,
	\begin{equation}
		\begin{aligned}
			 & {\rm E} \left[ \left|\frac{1}{T} \sum_{t=1}^{T}(G_{t, i j} G_{t, k l}-\operatorname{Cov}\left(G_{t, i j}, G_{t, k l}\right)) \right|^2 \right]
			  = \frac{1}{T^2 } \sum_{t=1}^{T} {\rm E} [ \{ G_{t, i j} G_{t, k l}-{\rm E}(G_{t, i j}G_{t, k l})\}^2 ] \\ &+ \frac{1}{T^2 } \sum_{t_1 \neq t_2}
              {\rm E} [ \{ G_{t_1, i j} G_{t_1, k l}-{\rm E}(G_{t_1, i j} G_{t_1, k l})\} 
			 \{ G_{t_2, i j} G_{t_2, k l}-{\rm E}(G_{t_2, i j} G_{t_2, k l})\}] \\
			 & \leq \frac{C}{T} +\frac{C}{T^2} \sum_{t_1 \neq t_2} \alpha(|t_1 - t_2|)^{1-2/\gamma} \leq \frac{C}{T} +\frac{C}{T } \sum_{u=1}^{T} \alpha(u)^{1-2/\gamma} = O_p(T^{-1}).
		\end{aligned}
	\end{equation}
    Here $C$ denotes a constant.
\end{proof}

\begin{lemma} \label{Lemma2}
	Under Conditions 1-4 and Conditions 6-7, it holds that
	\begin{equation}
		\begin{aligned}
			\sum_{j_1=1}^p \sum_{j_2=1}^p \| \widehat{\boldsymbol{\Omega}}_{g , m n, j_1 j_2} - \boldsymbol{\Omega}_{g, m n, j_1 j_2} \|^2 \quad & =\quad O_p (N^{2-2\delta_1} p^{2-2\delta_2} T^{-1}),\\
			\sum_{j_1=1}^{p} \sum_{j_2=1}^{p}\|\widehat{\boldsymbol{\Omega}}_{g f, m n, j_1  j_2}\|_2^2 \quad & = \quad O_p (N^{1- \delta_1} N^{1 - \delta_3} p^{1- \delta_2} p^{1- \delta_4} T^{-1}),\\
			\sum_{j_1=1}^{p} \sum_{j_2=1}^{p}\|\widehat{\boldsymbol{\Omega}}_{g e, m n, j_1 j_2}\|_2^2  \quad & = \quad O_p(N^{2-\delta_1} p^{2-\delta_2} T^{-1}),\\
			\sum_{j_1=1}^{p} \sum_{j_2=1}^{p}\|\widehat{\boldsymbol{\Omega}}_{f e, m n, j_1 j_2} \|_2^2 \quad & = \quad  O_p\left(N^{2-\delta_3} p^{2-\delta_4} T^{-1}\right),\\
			\sum_{j_1=1}^{p} \sum_{j_2=1}^{p}\|\widehat{\boldsymbol{\Omega}}_{f g, m n, j_1 j_2} \|_2^2 \quad & = \quad  O_p\left(N^{1-\delta_1}N^{1-\delta_3} p^{1-\delta_2} p^{1-\delta_4} T^{-1}\right),\\
			\sum_{j_1=1}^{p} \sum_{j_2=1}^{p}\|\widehat{\boldsymbol{\Omega}}_{f, m n, j_1 j_2} \|_2^2 \quad & = \quad  O_p\left(N^{2-2\delta_3} p^{2-2\delta_4} T^{-1}\right),\\
			\sum_{j_1=1}^{p} \sum_{j_2=1}^{p}\|\widehat{\boldsymbol{\Omega}}_{e g, m n, j_1 j_2} \|_2^2 \quad & = \quad  O_p\left(N^{2-\delta_1} p^{2-\delta_2} T^{-1}\right),\\
			\sum_{j_1=1}^{p} \sum_{j_2=1}^{p}\|\widehat{\boldsymbol{\Omega}}_{ef, m n, j_1 j_2} \|_2^2 \quad & = \quad  O_p\left(N^{2-\delta_3} p^{2-\delta_4} T^{-1}\right),\\
			\sum_{j_1=1}^{p} \sum_{j_2=1}^{p}\|\widehat{\boldsymbol{\Omega}}_{e, m n, j_1 j_2} \|_2^2 \quad & = \quad  O_p\left(N^2 p^{2} T^{-1}\right).
		\end{aligned}
	\end{equation}
	
	\begin{proof}
		Firstly, we have
		\begin{equation}
			\begin{aligned}
				&\left\|\widehat{\boldsymbol{\Omega}}_{g c, m n, j_1 j_2}-\boldsymbol{\Omega}_{g c,m n, j_1 j_2}\right\|_2^2 \leq\left\|\widehat{\boldsymbol{\Omega}}_{g c,m n, j_1 j_2}-\boldsymbol{\Omega}_{g c, m n, j_1 j_2}\right\|_F^2 \\
				&=\left\|\operatorname{vec}\left(\widehat{\boldsymbol{\Omega}}_{g c,m n, j_1 j_2}-\boldsymbol{\Omega}_{g c, m n, j_1 j_2}\right)\right\|_2^2 \\
				&=\left\|\frac{1}{T} \sum_{t=1}^{T} \operatorname{vec}\left(\boldsymbol{G}_t \boldsymbol{c}_{\cdot,j_1,m} \boldsymbol{c}_{\cdot,j_2,n}^{\prime}  \boldsymbol{G}_{t}^{\prime}-\mathrm{E}\left(\boldsymbol{G}_t \boldsymbol{c}_{\cdot,j_1,m} \boldsymbol{c}_{\cdot,j_2,n}^{\prime}  \boldsymbol{G}_{t}^{\prime}\right)\right)\right\|_2^2 \\
				&=\left\|\frac{1}{T} \sum_{t=1}^{T}\left[\boldsymbol{G}_{t} \otimes \boldsymbol{G}_t-\mathrm{E}\left(\boldsymbol{G}_{t} \otimes \boldsymbol{G}_t\right)\right] \operatorname{vec}\left(\boldsymbol{c}_{\cdot,j_1,m} \boldsymbol{c}_{\cdot,j_2,n}^{\prime}\right)\right\|_2^2 \\
				&\leq \left\|\frac{1}{T} \sum_{t=1}^{T}\left(\boldsymbol{G}_{t} \otimes \boldsymbol{G}_t-\mathrm{E}\left(\boldsymbol{G}_{t} \otimes \boldsymbol{G}_t\right)\right)\right\|_2^2 \| \operatorname{vec}\left(\boldsymbol{c}_{\cdot,j_1,m} \boldsymbol{c}_{\cdot,j_2,n}^{\prime}\right) \|_2^2. \\
				&=\left\|\frac{1}{T} \sum_{t=1}^{T}\left(\boldsymbol{G}_{t} \otimes \boldsymbol{G}_t-\mathrm{E}\left(\boldsymbol{G}_{t} \otimes \boldsymbol{G}_t\right)\right)\right\|_2^2\| \boldsymbol{c}_{\cdot,j_1,m} \boldsymbol{c}_{\cdot,j_2,n}^{\prime} \|_F^2. \\
				&\leq\left\| \frac{1}{T} \sum_{t=1}^{T}\left(\boldsymbol{G}_{t} \otimes \boldsymbol{G}_t-\mathrm{E}\left(\boldsymbol{G}_{t} \otimes \boldsymbol{G}_t\right)\right)\right\|_F^2 \| \boldsymbol{c}_{\cdot,j_1,m} \|_2^2 \cdot\| \boldsymbol{c}_{\cdot,j_2,n} \|_2^2 .
			\end{aligned}
		\end{equation}
		From Condition 4 and Lemma 1, it follows that
		\begin{equation}
			\begin{aligned}
				&\sum_{j_1=1}^{p} \sum_{j_2=1}^{p}\left\|\widehat{\boldsymbol{\Omega}}_{g,m n, j_1 j_2}-\boldsymbol{\Omega}_{g, m n, j_1 j_2}\right\|_2^2=\sum_{j_1=1}^{p} \sum_{j_2=1}^{p}\left\|\boldsymbol{R}_{m}\left(\widehat{\boldsymbol{\Omega}}_{g c, m n, j_1 j_2}-\boldsymbol{\Omega}_{g c, m n, j_1 j_2}\right) \boldsymbol{R}_n^{\prime}\right\|_2^2 \\
				& \leq \quad\|\boldsymbol{R}_m\|_2^2 \| \|\boldsymbol{R}_n\|_2^2 \left\| \frac{1}{T} \sum_{t=1}^{T}\left(\boldsymbol{G}_{t} \otimes \boldsymbol{G}_t-\mathrm{E}\left(\boldsymbol{G}_{t} \otimes \boldsymbol{G}_t\right)
				\right) \right\|_F^2\left(\sum_{j_1=1}^{p}\left\|\boldsymbol{c}_{\cdot,j_1,m}\right\|_2^2\right)\left(\sum_{j_2=1}^{p}\left\|\boldsymbol{c}_{\cdot,j_2,n}\right\|_2^2\right). \\
				&=\quad\|\boldsymbol{R}_m\|_2^2 \| \|\boldsymbol{R}_n\|_2^2 \left\| \frac{1}{T} \sum_{t=1}^{T}\left(\boldsymbol{G}_{t} \otimes \boldsymbol{G}_t-\mathrm{E}\left(\boldsymbol{G}_{t} \otimes \boldsymbol{G}_t\right)\right)\right\|_F^2 \| \boldsymbol{C}_m \|_F^2 \| \boldsymbol{C}_n \|_F^2. \\
				& \leq \quad k_2^2\|\boldsymbol{R}_m\|_2^2  \|\boldsymbol{R}_n\|_2^2 \left\| \frac{1}{T} \sum_{t=1}^{T}\left(\boldsymbol{G}_{t} \otimes \boldsymbol{G}_t-\mathrm{E}\left(\boldsymbol{G}_{t} \otimes \boldsymbol{G}_t\right)\right)\right\|_F^2\| \boldsymbol{C}_m \|_2^2 \| \boldsymbol{C}_n \|_2^2 =O_p\left(N^{2- 2\delta_1}  p^{2-2 \delta_2} T^{-1}\right) .
			\end{aligned}
		\end{equation}
		For the interaction component between the global factors and the local factors, we have
		\begin{equation}
			\begin{aligned}
				& \sum_{j_1=1}^{p} \sum_{j_2=1}^{p}\left\|\widehat{\boldsymbol{\Omega}}_{g f, m n, j_1  j_2}\right\|_2^2 \leq \sum_{j_1=1}^{p} \sum_{j_2=1}^{p}\|\boldsymbol{R}_m\|_2^2\|\boldsymbol{\Gamma}_n\|_2^2\left\|\frac{1}{T} \sum_{t=1}^{T} \boldsymbol{G}_t \boldsymbol{c}_{\cdot,j_1,m} \boldsymbol{\lambda}_{\cdot,j_2,n}^{\prime}\boldsymbol{F}_{nt}^{\prime}\right\|_2^2 \\
				&  \leq \quad\|\boldsymbol{R}_m\|_2^2\|\boldsymbol{\Gamma}_n\|_2^2\left(\sum_{j_1=1}^{p}\| \boldsymbol{c}_{\cdot,j_1,m}\|_2^2 \right)\left( \left\|\frac{1}{T} \sum_{t=1}^{T} \boldsymbol{F}_{nt} \otimes \boldsymbol{G}_t\right\|_F^2\right)\left(\sum_{j_2=1}^{p}\left\|\boldsymbol{\lambda}_{ \cdot,j_2,n}\right\|_2^2\right) \\
				& \leq \quad k_2^2  \|\boldsymbol{R}_m\|_2^2\|\boldsymbol{\Gamma}_n\|_2^2  \left( \left\|\frac{1}{T} \sum_{t=1}^{T} \boldsymbol{F}_{nt} \otimes \boldsymbol{G}_t\right\|_F^2\right) \|\boldsymbol{C}_m\|_2^2\|\boldsymbol{\Lambda}_n\|_2^2
				\quad =O _p\left(N^{1- \delta_1} N^{1 - \delta_3} p^{1- \delta_2} p^{1- \delta_4} T^{-1}\right) .
			\end{aligned}
		\end{equation}
		Similarly, 
		\begin{equation}
			\begin{aligned}
				& \sum_{j_1=1}^{p} \sum_{j_2=1}^{p}\left\|\widehat{\boldsymbol{\Omega}}_{g e, m n, j_1 j_2}\right\|_2^2 \leq \sum_{j_1=1}^{p} \sum_{j_2=1}^{p}\|\boldsymbol{R}_{m}\|_2^2\left\|\frac{1}{T} \sum_{t=1}^{T} \boldsymbol{G}_t \boldsymbol{c}_{\cdot,j_1,m} \boldsymbol{e}_{\cdot, j_2,nt }^{\prime}\right\|_2^2 \\
				& \leq \quad \|\boldsymbol{R}_{m}\|_2^2\left(\sum_{j_2=1}^{p}\left\|\frac{1}{T} \sum_{t=1}^{T} \boldsymbol{e}_{\cdot,j_2,nt} \otimes \boldsymbol{G}_t\right\|_2^2\right)\left(\sum_{j_1=1}^{p}\left\|\boldsymbol{c}_{\cdot,j_1,m}\right\|_F^2\right) \\
				& = \quad O_p\left(N^{2-\delta_1} p^{2-\delta_2} T^{-1}\right),
			\end{aligned}
		\end{equation}
		and 
		\begin{eqnarray*}
				 \sum_{j_1=1}^{p} \sum_{j_2=1}^{p}\left\|\widehat{\boldsymbol{\Omega}}_{f e, m n, j_1 j_2} \right\|_2^2 \quad & = & \quad  O_p\left(N^{2-\delta_3} p^{2-\delta_4} T^{-1}\right),\\
				 \sum_{j_1=1}^{p} \sum_{j_2=1}^{p}\left\|\widehat{\boldsymbol{\Omega}}_{f g, m n, j_1 j_2} \right\|_2^2 \quad & = & \quad  O_p\left(N^{1-\delta_1}N^{1-\delta_3} p^{1-\delta_2} p^{1-\delta_4} T^{-1}\right),\\
				 \sum_{j_1=1}^{p} \sum_{j_2=1}^{p}\left\|\widehat{\boldsymbol{\Omega}}_{f, m n, j_1 j_2} \right\|_2^2 \quad & = & \quad  O_p\left(N^{2-2\delta_3} p^{2-2\delta_4} T^{-1}\right),\\
				 \sum_{j_1=1}^{p} \sum_{j_2=1}^{p}\left\|\widehat{\boldsymbol{\Omega}}_{e g, m n, j_1 j_2} \right\|_2^2 \quad & = & \quad  O_p\left(N^{2-\delta_1} p^{2-\delta_2} T^{-1}\right),\\
				 \sum_{j_1=1}^{p} \sum_{j_2=1}^{p}\left\|\widehat{\boldsymbol{\Omega}}_{ef, m n, j_1 j_2} \right\|_2^2 \quad & = & \quad  O_p\left(N^{2-\delta_3} p^{2-\delta_4} T^{-1}\right),\\
				 \sum_{j_1=1}^{p} \sum_{j_2=1}^{p}\left\|\widehat{\boldsymbol{\Omega}}_{e, m n, j_1 j_2} \right\|_2^2 \quad & = & \quad  O_p\left(N^2 p^{2} T^{-1}\right).
		\end{eqnarray*}
	\end{proof}
\end{lemma}

With the nine rates established in Lemma 2, we can show the rate of convergence for the observed covariance matrix $\widehat{\boldsymbol{\Omega}}_{x,m n,j_1 j_2}$.

\begin{lemma} \label{Lemma3}
	Under Conditions 1-4 and Conditions 6-7, it\ holds\ that
	\begin{equation}
		\sum_{j_1=1}^p \sum_{j_2=1}^p \| \widehat{\boldsymbol{\Omega}}_{x,m n,j_1 j_2} - \boldsymbol{\Omega}_{x,m n,j_1 j_2} \|_2^2 = O_p(N^2 p^{2} T^{-1}).
	\end{equation}
	\begin{proof}
		\begin{equation}
			\begin{aligned}
				\widehat{\boldsymbol{\Omega}}_{x,m n,j_1 j_2} & = \frac{1}{T} \sum_{t=1}^{T} \boldsymbol{x}_{\cdot,j_1,mt} \boldsymbol{x}_{\cdot,j_2,nt}^{\prime}\\
				& = \frac{1}{T} \sum_{t=1}^T (\boldsymbol{R}_m \boldsymbol{G}_t \boldsymbol{c}_{\cdot,j_1,m}+ \boldsymbol{\Gamma}_m \boldsymbol{F}_{mt} \boldsymbol{\lambda}_{\cdot,j_1,m} + \boldsymbol{e}_{\cdot,j_1,mt})(\boldsymbol{R}_n \boldsymbol{G}_t \boldsymbol{c}_{\cdot,j_1,n}+ \boldsymbol{\Gamma}_n \boldsymbol{F}_{nt} \boldsymbol{\lambda}_{\cdot,j_2,n} + \boldsymbol{e}_{\cdot,j_2,nt})^{\prime} \\
				& = \widehat{\boldsymbol{\Omega}}_{g, m n, j_1 j_2} + \widehat{\boldsymbol{\Omega}}_{g f, m n, j_1 j_2} + \widehat{\boldsymbol{\Omega}}_{g e, m n, j_1 j_2} + \widehat{\boldsymbol{\Omega}}_{f g, m n, j_1 j_2} +\widehat{\boldsymbol{\Omega}}_{f, m n, j_1 j_2} \\
				& + \widehat{\boldsymbol{\Omega}}_{f e, m n, j_1 j_2} + \widehat{\boldsymbol{\Omega}}_{e g, m n, j_1 j_2}
				+\widehat{\boldsymbol{\Omega}}_{e f, m n, j_1 j_2} + \widehat{\boldsymbol{\Omega}}_{e, m n, j_1 j_2}.
			\end{aligned}
		\end{equation}
		Then by Lemma 2, it follows that
		\begin{equation}
			\begin{aligned}
				& \sum_{j_1=1}^p \sum_{j_2=1}^p \| \widehat{\boldsymbol{\Omega}}_{x,m n,j_1 j_2} - {\boldsymbol{\Omega}}_{x,m n,j_1 j_2} \|_2^2\\
				& \leq 9 \sum_{j_1=1}^{p } \sum_{j_2=1}^p (\| \widehat{\boldsymbol{\Omega}}_{g, m n, j_1 j_2}- \boldsymbol{\Omega}_{g, m n, j_1 j_2} \|_2^2 + \|\widehat{\boldsymbol{\Omega}}_{g f, m n, j_1 j_2} \|_2^2 + \|\widehat{\boldsymbol{\Omega}}_{g e , m n, j_1 j_2} \|_2^2 +\|\widehat{\boldsymbol{\Omega}}_{f g, m n, j_1 j_2} \|_2^2 \\
				& \|\widehat{\boldsymbol{\Omega}}_{f, m n, j_1 j_2} \|_2^2 + \|\widehat{\boldsymbol{\Omega}}_{f e , m n, j_1 j_2} \|_2^2 + \|\widehat{\boldsymbol{\Omega}}_{e g, m n, j_1 j_2} \|_2^2 +\|\widehat{\boldsymbol{\Omega}}_{e f , m n, j_1 j_2} \|_2^2+\|\widehat{\boldsymbol{\Omega}}_{e, m n, j_1 j_2} \|_2^2 )\\
				& = O_p(N^2 p^{2} T^{-1}).
			\end{aligned}
		\end{equation}
	\end{proof}
\end{lemma}

\begin{lemma} \label{Lemma4}
Under Conditions 1-4, Conditions 6-7, and $N^{\delta_1} p^{\delta_2} T^{-1/2} = o(1)$, it\ holds\ that
	\begin{equation}
		\| \widehat{\boldsymbol{W}}_{1,m} - \boldsymbol{W}_{1,m} \|_2 = O_p (N^{2-\delta_1} p^{2-\delta_2} T^{-1/2}),\quad m=1,...,M.
	\end{equation}
	\begin{proof}
 From the definitions of $\widehat{\boldsymbol{W}}_{1,m}$ and $\boldsymbol{W}_{1,m}$ in \eqref{W1mdef} and \eqref{hatW1mdef}, it follows that
		\begin{equation}
			\begin{aligned}
				&  \| \widehat{\boldsymbol{W}}_{1,m} - \boldsymbol{W}_{1,m} \|_2 =
				\| \sum_{i=1, i \neq m}^{M} \sum_{j_1=1}^p \sum_{j_2=1}^p \left( \widehat{\boldsymbol{\Omega}}_{x,m i,j_1 j_2} \widehat{\boldsymbol{\Omega}}_{x,m i,j_1 j_2}^{\prime} - {\boldsymbol{\Omega}}_{x,m i,j_1 j_2} {\boldsymbol{\Omega}}_{x,m i,j_1 j_2}^{\prime}\right) \|_2 \\
				& \leq \sum_{i=1, i \neq m}^{M} \sum_{j_1=1}^p \sum_{j_2=1}^p \left(\|(\widehat{\boldsymbol{\Omega}}_{x,m i,j_1 j_2}-\boldsymbol{\Omega}_{x,m i,j_1 j_2})(\widehat{\boldsymbol{\Omega}}_{x,m i,j_1 j_2}-\boldsymbol{\Omega}_{x,m i,j_1 j_2})^{\prime}\|_2+2\|\boldsymbol{\Omega}_{x,m i,j_1 j_2}\|_2\|\widehat{\boldsymbol{\Omega}}_{x,m i,j_1 j_2}-\boldsymbol{\Omega}_{x,m i,j_1 j_2}\|_2\right) \\
				& \leq \sum_{i=1, i \neq m}^{M} \sum_{j_1=1}^p \sum_{j_2=1}^p \|(\widehat{\boldsymbol{\Omega}}_{x,m i,j_1 j_2}-\boldsymbol{\Omega}_{x,m i,j_1 j_2}) \|_2^2 +2 \sum_{i=1, i \neq m}^{M} \sum_{j_1=1}^p \sum_{j_2=1}^p \|\boldsymbol{\Omega}_{x,m i,j_1 j_2}\|_2\|\widehat{\boldsymbol{\Omega}}_{x,m i,j_1 j_2}-\boldsymbol{\Omega}_{x,m i,j_1 j_2}\|_2.
			\end{aligned}
		\end{equation}
		
		We have 
		\begin{equation}  \label{eq1}
			\begin{aligned}
				& \sum_{j_1=1}^p \sum_{j_2=1}^p \| \boldsymbol{\Omega}_{x,m i,j_1 j_2} \|_2^2 = \sum_{j_1=1}^p \sum_{j_2=1}^p \| \boldsymbol{R}_m \boldsymbol{\Omega}_{gc ,m i,j_1 j_2} \boldsymbol{R}_i^{\prime} \|_2^2 \leq \sum_{j_1=1}^p \sum_{j_2=1}^p \|\boldsymbol{R}_m\|_2^2
				\|\boldsymbol{R}_i\|_2^2 \| \boldsymbol{\Omega}_{gc ,m i,j_1 j_2} \|_2^2\\
				& \quad \leq \quad \|\boldsymbol{R}_m\|_2^2
				\|\boldsymbol{R}_i\|_2^2 \cdot \left\| \frac{1}{T} \sum_{t=1}^T E(\boldsymbol{G}_t \otimes
				\boldsymbol{G}_t) \right\|_2^2 \cdot \left(\sum_{j_1=1}^p \| \boldsymbol{c}_{\cdot,j_1,m} \|_2^2 \right)
				\left(\sum_{j_2=1}^p \| \boldsymbol{c}_{\cdot,j_2,n} \|_2^2 \right)\\
				& \quad = \quad O_p(N^{2-2\delta_1}  p^{2-2\delta_2}).
			\end{aligned}
		\end{equation}
		By \eqref{eq1} and Lemma 3, we have
		\begin{equation}
			\begin{aligned}
				& \left( \sum_{j_1=1}^p \sum_{j_2=1}^p \|\boldsymbol{\Omega}_{x,m i,j_1 j_2}\|_2\|\widehat{\boldsymbol{\Omega}}_{x,m i,j_1 j_2}-\boldsymbol{\Omega}_{x,m i,j_1 j_2}\|_2\right)^2 \\
				& \quad \leq \quad \left( \sum_{j_1=1}^p \sum_{j_2=1}^p \|\boldsymbol{\Omega}_{x,m i,j_1 j_2}\|_2^2 \right) \cdot \left( \sum_{j_1=1}^p \sum_{j_2=1}^p \|\widehat{\boldsymbol{\Omega}}_{x,m i,j_1 j_2}-\boldsymbol{\Omega}_{x,m i,j_1 j_2}\|_2^2\right)\\
				& \quad \leq \quad O_p (N^{2-2\delta_1}  p^{2-2\delta_2} N^2 p^2 T^{-1}) = O_p (N^{4-2\delta_1}  p^{4-2\delta_2}T^{-1}).
			\end{aligned}
		\end{equation}
		It follows that
		\begin{equation}
			\| \widehat{\boldsymbol{W}}_{1,m} - \boldsymbol{W}_{1,m} \|_2 = O_p (N^{2-\delta_1} p^{2-\delta_2} T^{-1/2}).
		\end{equation}
	\end{proof}
\end{lemma}

\begin{lemma} \label{Lemma5}
Under Conditions 2-3 and 7, we have
	\begin{equation}
		\lambda_i(\boldsymbol{W}_{1,m}) \asymp N^{2-2\delta_1} p^{2-2\delta_2},\quad i=1,2,...,k_1,\ m=1,2,...,M,
	\end{equation}
 where $\lambda_i(\boldsymbol{W}_{1,m})$ denotes the $i$-th largest  eigenvalue of $\boldsymbol{W}_{1,m}$.
	\begin{proof}
		By definition, we have
		\begin{equation*}
			\begin{aligned}
				\boldsymbol{\Omega}_{gc,mi,j_1 j_2} \quad & = \quad \frac{1}{T} \sum_{t=1}^T \text{E} [(\boldsymbol{c}_{\cdot,j_1,m} \otimes \boldsymbol{I}_{k_1}) \text{vec}(\boldsymbol{G}_t) \text{vec}(\boldsymbol{G}_t)^{\prime} (\boldsymbol{c}_{\cdot,j_2,i}^{\prime} \otimes \boldsymbol{I}_{k_1})] \\
				& = \quad (\boldsymbol{c}_{\cdot,j_1,m} \otimes \boldsymbol{I}_{k_1})
				\boldsymbol{\Sigma}_G (\boldsymbol{c}_{\cdot,j_2,i}^{\prime} \otimes \boldsymbol{I}_{k_1}).
			\end{aligned}   
		\end{equation*}
		We have
		\begin{equation}
			\begin{aligned}
				\lambda_{k_1}\left(\boldsymbol{W}_{1,m}\right) & =\lambda_{k_1}\left(\sum_{i=1,i\neq m}^{M} \sum_{j_1=1}^{p} \sum_{j_2=1}^{p} \boldsymbol{R}_m \boldsymbol{\Omega}_{g c,m i,j_1 j_2} \boldsymbol{R}_{i}^{\prime} \boldsymbol{R}_i \boldsymbol{\Omega}_{g c,m i,j_1 j_2}^{\prime} \boldsymbol{R}_{m}^{\prime}\right) \\
				&\geq \lambda_{k_1}\left( \sum_{j_1=1}^{p} \sum_{j_2=1}^{p} \boldsymbol{R}_m \boldsymbol{\Omega}_{g c,m i,j_1 j_2} \boldsymbol{R}_{i}^{\prime} \boldsymbol{R}_i \boldsymbol{\Omega}_{g c,m i,j_1 j_2}^{\prime} \boldsymbol{R}_{m}^{\prime}\right) (i \neq m)\\
				& \geq \|\boldsymbol{R}_i\|_{\min }^2 \cdot \|\boldsymbol{R}_m\|_{\min }^2 \cdot \lambda_{k_1}\left(\sum_{i=1,i \neq m}^{M} \sum_{j_1=1}^{p} \sum_{j_2=1}^{p} \boldsymbol{\Omega}_{g c,m i,j_1 j_2} \boldsymbol{\Omega}_{g c,m i,j_1 j_2}^{\prime}\right) \\
				& =\|\boldsymbol{R}_i\|_{\min }^2 \cdot \|\boldsymbol{R}_m\|_{\min }^2 \cdot \lambda_{k_1}\left( \sum_{j_1=1}^{p} \sum_{j_2=1}^{p}\left(\boldsymbol{c}_{\cdot,j_1,m} \otimes \boldsymbol{I}_{k_1}\right) \boldsymbol{\Sigma}_G\left(\boldsymbol{c}_{\cdot,j_2,i}^{\prime} \otimes \boldsymbol{I}_{k_1}\right)\left(\boldsymbol{c}_{\cdot,j_2,i} \otimes \boldsymbol{I}_{k_1}\right) \boldsymbol{\Sigma}_G^{\prime}\left(\boldsymbol{c}_{\cdot,j_1,m}^{\prime} \otimes \boldsymbol{I}_{k_1}\right)\right) \\
				& \geq \|\boldsymbol{R}_i\|_{\min }^2 \cdot \|\boldsymbol{R}_m\|_{\min }^2 \cdot \lambda_{k_1}\left( \sum_{j_1=1}^{p} \sum_{j_2=1}^{p}\left(\boldsymbol{c}_{\cdot,j_1,m} \otimes \boldsymbol{I}_{k_1}\right) \boldsymbol{\Sigma}_G \left(\boldsymbol{c}_{\cdot,j_2,i}^{\prime} \cdot \boldsymbol{c}_{\cdot,j_2,i} \otimes \boldsymbol{I}_{k_1}\right) \boldsymbol{\Sigma}_G^{\prime}\left(\boldsymbol{c}_{\cdot,j_1,m}^{\prime} \otimes \boldsymbol{I}_{k_1}\right)\right) \\
				& = \|\boldsymbol{R}_i\|_{\min }^2  \cdot \|\boldsymbol{R}_m\|_{\min }^2 \cdot \lambda_{k_1}\left( \sum_{j_1=1}^{p}\left(\boldsymbol{c}_{\cdot,j_1,m} \otimes \boldsymbol{I}_{k_1}\right) \boldsymbol{\Sigma}_G \left(\boldsymbol{C}_{i}^{\prime} \boldsymbol{C}_{i} \otimes \boldsymbol{I}_{k_1}\right) \boldsymbol{\Sigma}_G^{\prime}\left(\boldsymbol{c}_{\cdot,j_1,m}^{\prime} \otimes \boldsymbol{I}_{k_1}\right)\right) \\
				& =\|\boldsymbol{R}_i\|_{\min }^2 \cdot \|\boldsymbol{R}_m\|_{\min }^2 \cdot \lambda_{k_1}\left( \sum_{j_1=1}^{p}\left(\boldsymbol{c}_{\cdot,j_1,m} \otimes \boldsymbol{I}_{k_1}\right) \boldsymbol{\Sigma}_G\left(\boldsymbol{C}_{i}^{\prime} \otimes \boldsymbol{I}_{k_1}\right)\left(\boldsymbol{C}_{i} \otimes \boldsymbol{I}_{k_1}\right) \boldsymbol{\Sigma}_G^{\prime}\left(\boldsymbol{c}_{\cdot,j_1,m}^{\prime} \cdot \otimes \boldsymbol{I}_{k_1}\right)\right) \\
				& =\|\boldsymbol{R}_i\|_{\min }^2 \cdot \|\boldsymbol{R}_m\|_{\min }^2 \cdot \lambda_{k_1}\left( \sum_{j_1=1}^{p} \left(\boldsymbol{C}_{i} \otimes \boldsymbol{I}_{k_1}\right) \boldsymbol{\Sigma}_G^{\prime}\left(\boldsymbol{c}_{\cdot,j_1,m} \otimes \boldsymbol{I}_{k_1}\right)
				\left(\boldsymbol{c}_{\cdot,j_1,m} \otimes \boldsymbol{I}_{k_1}\right)\boldsymbol{\Sigma}_G\left(\boldsymbol{C}_{i}^{\prime} \otimes \boldsymbol{I}_{k_1}\right)\right) \\
				& =\|\boldsymbol{R}_i\|_{\min }^2 \cdot \|\boldsymbol{R}_m\|_{\min }^2 \cdot \lambda_{k_1}\left(\left(\boldsymbol{C}_i \otimes \boldsymbol{I}_{k_1}\right) \boldsymbol{\Sigma}_G^{\prime}\left(\boldsymbol{C}_m^{\prime} \boldsymbol{C}_m \otimes \boldsymbol{I}_{k_1}\right) \boldsymbol{\Sigma}_G \left(\boldsymbol{C}_i^{\prime} \otimes \boldsymbol{I}_{k_1}\right)\right).
			\end{aligned}
		\end{equation}
		Since $\boldsymbol{C}_i^{\prime}\boldsymbol{C}_i$ is a $k_2\times k_2$ symmetric positive matrix, we can find $k_2 \times k_2$ positive matrix $\boldsymbol{U}_i$, such that $\boldsymbol{C}_i^{\prime}\boldsymbol{C}_i=\boldsymbol{U}_i\boldsymbol{U}_i^{\prime}$ and $\|\boldsymbol{U}_i\|_2^2 \asymp O(p^{1-\delta_2}) \asymp \|\boldsymbol{U}_i\|_{\text{min}}^2.$ By the properties of Kronecker product, we have $\sigma_1(\boldsymbol{U}_i \otimes \boldsymbol{I}_{k_1}) \asymp O(p^{1/2-\delta_2/2}) \asymp \sigma_{k_1 k_2}(\boldsymbol{U}_i \otimes \boldsymbol{I}_{k_1}).$ Using Condition 2 and Theorem 9 in \citet{merikoski2004inequalities}, it follows that $\sigma_{k_1}(\boldsymbol{\Sigma}_G^{\prime} (\boldsymbol{U_i } \otimes \boldsymbol{I}_{k_1})) \asymp O(p^{1/2-\delta_2/2}),\ i=1,...,M$. We have
		\begin{equation}
			\begin{aligned}
				\lambda_{k_1}\left(\boldsymbol{W}_{1,m}\right) & \geq \|\boldsymbol{R}_i\|_{\min }^2 \cdot\|\boldsymbol{R}_m\|_{\min }^2 \cdot \lambda_{k_1}\left(\left(\boldsymbol{C}_i \otimes \boldsymbol{I}_{k_1}\right) \boldsymbol{\Sigma}_G^{\prime}\left(\boldsymbol{U}_m \otimes \boldsymbol{I}_{k_1}\right)\left(\boldsymbol{U}_m^{\prime} \otimes \boldsymbol{I}_{k_1}\right) \boldsymbol{\Sigma}_G\left(\boldsymbol{C}_i^{\prime} \otimes \boldsymbol{I}_{k_1}\right)\right) \\
				& =\|\boldsymbol{R}_i\|_{\min }^2 \cdot \|\boldsymbol{R}_m\|_{\min }^2 \cdot \lambda_{k_1}\left(\left(\boldsymbol{U}_m^{\prime} \otimes \boldsymbol{I}_{k_1}\right) \boldsymbol{\Sigma}_G\left(\boldsymbol{C}_i^{\prime} \boldsymbol{C}_i \otimes \boldsymbol{I}_{k_1}\right) \boldsymbol{\Sigma}_G^{\prime} \left(\boldsymbol{U}_m \otimes \boldsymbol{I}_{k_1}\right)\right) \\
				& =\|\boldsymbol{R}_i\|_{\min }^2 \cdot \|\boldsymbol{R}_m\|_{\min }^2 \cdot \lambda_{k_1}\left(\left(\boldsymbol{U}_m^{\prime} \otimes \boldsymbol{I}_{k_1}\right) \boldsymbol{\Sigma}_G\left(\boldsymbol{U}_i \otimes \boldsymbol{I}_{k_1}\right)\left(\boldsymbol{U}_i^{\prime} \otimes \boldsymbol{I}_{k_1}\right) \boldsymbol{\Sigma}_G^{\prime} \left(\boldsymbol{U}_m \otimes \boldsymbol{I}_{k_1}\right)\right) \\
				&  \geq \|\boldsymbol{R}_i\|_{\min }^2 \cdot \|\boldsymbol{R}_m\|_{\min }^2 \cdot\left[\sigma_{k_1}\left(\left(\boldsymbol{U}_i^{\prime} \otimes \boldsymbol{I}_{k_1}\right) \boldsymbol{\Sigma}_G^{\prime} \left(\boldsymbol{U}_m \otimes \boldsymbol{I}_{k_1}\right)\right)\right]^2=O\left(N^{2-2 \delta_1} p^{2-2 \delta_2}\right) .
			\end{aligned}
		\end{equation}
	\end{proof}
\end{lemma} 
\noindent{\bf Proof of Theorem 1}
\begin{proof}
	By Lemmas 1-5, and Lemma 3 in \citet{10.1093/biomet/asr048}, we have
	\begin{equation}
		\| \widehat{\boldsymbol{Q}}_{1m} - \boldsymbol{Q}_{1m} \|_2 \leq \frac{8}{\lambda_{min}(\boldsymbol{W}_{1,m})} \| \widehat{\boldsymbol{W}}_{1,m} -\boldsymbol{W}_{1,m} \|_2 = 
		O_p \left( N^{\delta_1} p^{\delta_2} T^{-1/2}\right).
	\end{equation}
\end{proof}
Proof for $\| \widehat{\boldsymbol{Q}}_{2m} - \boldsymbol{Q}_{2m}\|_2$ is similar. 

\noindent{\bf Proof of Theorem 2}
\begin{proof}
	Let $\widehat{\lambda}_{1m,j}$ and $\widehat{\boldsymbol{q}}_{1m,j}$ denote the $j$-th largest eigenvalue of $\widehat{\boldsymbol{W}}_{1,m}$ and its corresponding eigenvector. $\lambda_{1m,j}$ and $\boldsymbol{q}_{1m,j}$ denote the eigenvalue and eigenvector of matrix $\boldsymbol{W}_{1,m}$. Let $\widehat{\boldsymbol{Q}}_{1,m}=(\widehat{\boldsymbol{q}}_{1m,1},...,\widehat{\boldsymbol{q}}_{1m,k_1})$, $\boldsymbol{B}_{1,m}=(\boldsymbol{q}_{1m,k_1+1},...,\boldsymbol{q}_{1m,N_m})$, $\boldsymbol{Q}_{1,m} = (\boldsymbol{q}_{1m,1},...,\boldsymbol{q}_{1m,k_1})$, $\widehat{\boldsymbol{Q}}_{1,m}=(\widehat{\boldsymbol{q}}_{1m,1},...,\widehat{\boldsymbol{q}}_{1m,k_1})$ and $\widehat{\boldsymbol{B}}_{1,m} = (\widehat{\boldsymbol{q}}_{1m,k_1+1},...,\widehat{\boldsymbol{q}}_{1m,N_m})$. We have
	\begin{equation}
		\lambda_{1m,j}=\boldsymbol{q}_{1m,j}^{\prime} \boldsymbol{W}_{1,m} \boldsymbol{q}_{1m,j},\ \text{and} \quad \widehat{\lambda}_{1m,j}= \widehat{\boldsymbol{q}}_{1m,j}^{\prime} \widehat{\boldsymbol{W}}_{1,m} \widehat{\boldsymbol{q}}_{1m,j},\quad j=1,...,k_1.
	\end{equation}
	We can decompose $\widehat{\lambda}_{1m,j}-\lambda_{1m,j}$ by 
	\begin{equation}
		\widehat{\lambda}_{1m, j}-\lambda_{1m, j}=\widehat{\boldsymbol{q}}_{1m, j}^{\prime} \widehat{\boldsymbol{W}}_{1,m} \widehat{\boldsymbol{q}}_{1m, j}-\boldsymbol{q}_{1m, j}^{\prime} \boldsymbol{W}_{1,m} \boldsymbol{q}_{1m, j}=I_1+I_2+I_3+I_4+I_5,
	\end{equation}
	where
	\begin{equation*}
		\begin{gathered}
			I_1=\left(\widehat{\boldsymbol{q}}_{1m, j}-\boldsymbol{q}_{1m, j}\right)^{\prime}\left(\widehat{\boldsymbol{W}}_{1,m}-\boldsymbol{W}_{1,m}\right) \widehat{\boldsymbol{q}}_{1m, j}, \quad I_2=\left(\widehat{\boldsymbol{q}}_{1m, j}-\boldsymbol{q}_{1m, j}\right)^{\prime} \boldsymbol{W}_{1,m}\left(\widehat{\boldsymbol{q}}_{1m, j}-\boldsymbol{q}_{1m, j}\right), \\
			I_3=\left(\widehat{\boldsymbol{q}}_{1 m , j}-\boldsymbol{q}_{1 m , j}\right)^{\prime} \boldsymbol{W}_{1,m} \boldsymbol{q}_{1 m, j}, \quad I_4=\boldsymbol{q}_{1 m, j}^{\prime}\left(\widehat{\boldsymbol{W}}_{1,m}-\boldsymbol{W}_{1,m}\right) \widehat{\boldsymbol{q}}_{1 m , j}, \quad I_5=\boldsymbol{q}_{1 m, j}^{\prime} \boldsymbol{W}_{1,m}\left(\widehat{\boldsymbol{q}}_{1 m, j}-\boldsymbol{q}_{1 m, j}\right) .
		\end{gathered}
	\end{equation*}
	For $j=1,...,k_1$, $\| \widehat{\boldsymbol{q}}_{1m,j} - \boldsymbol{q}_{1m,j} \|_2 \leq \| \widehat{\boldsymbol{Q}}_{1,m} - \boldsymbol{Q}_{1,m} \|_2 = O_p(N^{\delta_1} p^{\delta_2} T^{-1/2}),$ and $\| \boldsymbol{W}_{1,m} \|_{2} = O_p( N^{2 - 2\delta_1} p^{2-2\delta_2})$. By Lemma 4, we have $\left\|I_1\right\|_2$ and $\left\|I_2\right\|_2$ are of order $O_p\left(N^{2} p^{2} T^{-1}\right)$ and $\left\|I_3\right\|_2,\left\|I_4\right\|_2$ and $\left\|I_5\right\|_2$ are of order $O_p\left(N^{2- \delta_1} p^{2- \delta_2} T^{-1/2}\right)$. So $\|\widehat{\lambda}_{1m, j}-\lambda_{1m, j}\|=O_p \left(N^{2-\delta_1} p^{2-\delta_2} T^{-1 / 2}\right)$.
	
	For $j=k_1+1, \ldots, N_m$, define,
	$$
	\widetilde{\boldsymbol{W}}_{1,m}=\sum_{i=1, i \neq m}^{M} \sum_{j_1=1}^p \sum_{j_2=1}^p  \widehat{\boldsymbol{\Omega}}_{x,m i,j_1 j_2} \boldsymbol{\Omega}_{x,m i,j_1 j_2}^{\prime} ,
	$$
	It can be shown that $\|\widehat{\boldsymbol{B}}_{1m}-\boldsymbol{B}_{1m}\|_2=O_p\left(N^{\delta_1} p^{\delta_2} T^{-1/2}\right)$, similar to proof of Theorem 1 with Lemma 3 in \citet{10.1093/biomet/asr048}. Hence, $\left\|\widehat{\boldsymbol{q}}_{1m, j}-\boldsymbol{q}_{1m, j}\right\|_2 \leq\|\widehat{\boldsymbol{B}}_{1m}-\boldsymbol{B}_{1m}\|_2=O_p\left(N^{\delta_1} p^{\delta_2} T^{-1/2}\right)$.
	Since $\lambda_{1m, j}=0$, for $j=k_1+1, \ldots, N_m$, consider the decomposition
	$$
	\widehat{\lambda}_{1m, j}=\widehat{\boldsymbol{q}}_{1m, j}^{\prime} \widehat{\boldsymbol{W}}_{1,m} \widehat{\boldsymbol{q}}_{1m, j}=K_1+K_2+K_3,
	$$
	where
	$$
	\begin{gathered}
		K_1=\widehat{\boldsymbol{q}}_{1m, j}^{\prime}\left(\widehat{\boldsymbol{W}}_{1,m}-\widetilde{\boldsymbol{W}}_{1,m}-\widetilde{\boldsymbol{W}}^{\prime}_{1,m}+\boldsymbol{W}_{1,m}\right) \widehat{\boldsymbol{q}}_{1m, j}, \quad K_2=2 \widehat{\boldsymbol{q}}_{1m, j}^{\prime}\left(\widetilde{\boldsymbol{W}}_{1,m}-\boldsymbol{W}_{1,m}\right)\left(\widehat{\boldsymbol{q}}_{1m, j}-\boldsymbol{q}_{1m, j}\right), \\
		K_3=\left(\widehat{\boldsymbol{q}}_{1m, j}-\boldsymbol{q}_{1m, j}\right)^{\prime} \boldsymbol{W}_{1,m}\left(\widehat{\boldsymbol{q}}_{1m, j}-\boldsymbol{q}_{1m, j}\right) .
	\end{gathered}
	$$
	By Lemma 2 and Lemma 4 ,
	$$
	\begin{aligned}
		K_1 & =\sum_{i=1, i \neq m}^{M}\left\|\sum_{j_1=1}^{p} \sum_{j_2=1}^{p}\left(\widehat{\boldsymbol{\Omega}}_{x,m i,j_1 j_2} - \boldsymbol{\Omega}_{x,m i,j_1 j_2}\right) \widehat{\boldsymbol{q}}_{1m, j}\right\|_2^2 \\ & \leq\sum_{i=1, i \neq m}^{M} \sum_{j_1=1}^{p} \sum_{j_2=1}^{p}\left\|\widehat{\boldsymbol{\Omega}}_{x,m i,j_1 j_2} - \boldsymbol{\Omega}_{x,m i,j_1 j_2}\right\|_2^2=O_p\left(N^2 p^2 T^{-1}\right), \\
		\left|K_2\right| & =O_p\left(\left\|\widetilde{\boldsymbol{W}}_{1m}-\boldsymbol{W}_{1m}\right\|_2 \cdot\left\|\widehat{\boldsymbol{q}}_{1m, j}-\boldsymbol{q}_{1m, j}\right\|_2\right)=O_p\left(\left\|\widetilde{\boldsymbol{W}}_{1m}-\boldsymbol{W}_{1m}\right\|_2 \cdot\left\|\widehat{\boldsymbol{B}}_{1m}-\boldsymbol{B}_{1m}\right\|_2\right)=O_p\left(N^2 p^2 T^{-1}\right), \\
		\left|K_3\right| & =O_p\left(\left\|\widehat{\boldsymbol{B}}_{1m}-\boldsymbol{B}_{1m}\right\|_2^2 \cdot\left\|\boldsymbol{W}_{1m}\right\|_2\right)=O_p\left(N^2 p^2 T^{-1}\right) . \\
		\text { Hence } & \widehat{\lambda}_{1m, j}=O_p\left(N^2 p^2 T^{-1}\right) .
	\end{aligned}
	$$
\end{proof}
\begin{lemma} \label{Lemma6}
	Let $\boldsymbol{U}_{mt} = \widehat{\boldsymbol{B}}_{1m} \boldsymbol{E}_{mt}$, each element of $\boldsymbol{\Sigma}_u ={\rm Cov} ({\rm vec}(\boldsymbol{U}_{mt}))$ is uniformly bounded as $N$ and $p$ increase to infinity.
	\begin{proof}
		By Lemma 1 in \citet{doi:10.1080/01621459.2019.1584899}, we can obtain this result.
	\end{proof}
 \end{lemma}
Let define some notations for the estimation of local factor loadings:
 \begin{equation*}
	\begin{aligned}
		\boldsymbol{Y}_{mt}^{(1)} &  := \widehat{\boldsymbol{B}}_{1m}^\prime \boldsymbol{X}_{mt} = \widehat{\boldsymbol{B}}_{1m}^\prime \boldsymbol{Q}_{3m} \boldsymbol{Z}_{mt} \boldsymbol{Q}_{4m}^{\prime} + \boldsymbol{U}_{mt} + \widehat{\boldsymbol{B}}_{1m}^{\prime} (\boldsymbol{Q}_{1m} - \widehat{\boldsymbol{Q}}_{1m}) \boldsymbol{S}_{mt} \boldsymbol{Q}_{2m}^{\prime}.
	\end{aligned}	
\end{equation*}
Consider the $j$th column of $\boldsymbol{Y}_{mt}^{(1)^{\prime} }$ , denote by $\boldsymbol{y}_{mt,j}:$
\begin{equation*}
	\boldsymbol{y}_{mt,j\cdot} = \boldsymbol{Q}_{4m} \boldsymbol{Z}_{mt}^{\prime} \boldsymbol{Q}_{3m}^{\prime} \boldsymbol{b}_{1m,j\cdot} + \boldsymbol{u}_{mt,j\cdot} + \boldsymbol{Q}_{2m} \boldsymbol{S}_{mt}^{\prime} (\boldsymbol{Q}_{1m} - \widehat{\boldsymbol{Q}}_{1m})^{\prime} \boldsymbol{b}_{1m,j\cdot},
\end{equation*}
where $\boldsymbol{b}_{mt,j\cdot}$ and $\boldsymbol{u}_{mt, j\cdot}$ be the $j$ column of $\widehat{\boldsymbol{B}}_{1m}$ and $\boldsymbol{U}_{mt}^{\prime}$, respectively. \\
Define:
\begin{eqnarray*}
		     \boldsymbol{y}_{mt,j\cdot}^{*} & = &  \boldsymbol{Q}_{4m} \boldsymbol{Z}_{mt}^{\prime} \boldsymbol{Q}_{3m}^{\prime} \boldsymbol{b}_{1m,j\cdot} ,\\
			 \boldsymbol{\Pi}_{ym,ij}(h)  & = & \frac{1}{T-h} \sum_{t=1}^{T-h} \text{Cov}(\boldsymbol{y}_{mt,i\cdot}^{*}, \boldsymbol{y}_{m(t+h),j\cdot}^{*}) ,\\
			 \boldsymbol{M}_{1m}  & = & \sum_{h=1}^{h_0} \sum_{i=1}^{N_m - k_1} \sum_{j=1}^{N_m-k_1}  \boldsymbol{\Pi}_{ym,ij}(h) \boldsymbol{\Pi}_{ym,ij}(h)^{\prime}, \\
		 \widehat{\boldsymbol{\Pi}}_{ym,ij}(h) & = & \frac{1}{T-h} \sum_{t=1}^{T-h} \boldsymbol{y}_{mt,i}, \boldsymbol{y}_{m(t+h),j}^{\prime} ,\\
			 \widehat{\boldsymbol{M}}_{1m}  & = & \sum_{h=1}^{h_0} \sum_{i=1}^{N_m - k_1} \sum_{j=1}^{N_m - k_1}  \widehat{\boldsymbol{\Pi}}_{ym,ij}(h) \widehat{\boldsymbol{\Pi}}_{ym,ij}(h)^{\prime}.
\end{eqnarray*}
\begin{lemma} \label{Lemma7}
	Under Conditions 1-7, $N^{\delta_1} p^{\delta_2} T^{-1/2} = o(1)$ and $N^{\delta_3} p^{\delta_4} T^{-1/2} = o(1)$, it holds that 
	\begin{equation*}
		 \sum_{i=1}^{N_m - k_1} \sum_{j=1}^{N_m - k_1}\| \boldsymbol{\Pi}_{ym,ij}(h) - \widehat{\boldsymbol{\Pi}}_{ym,ij}(h) \|_2 = O_p\left(\max \{N^{2 } p^{2} T^{-1/2}, N^{2 } p^{2} T^{-1/2}N^{\delta_1 - \delta_3} p^{\delta_2 -\delta_4}\}\right).
	\end{equation*}
\end{lemma}
\begin{proof}
	The difference between Lemma 7. and the Lemma 3. in \citet{wang2019factor} is the term $ \boldsymbol{Q}_{2m} \boldsymbol{S}_{mt}^{\prime} (\boldsymbol{Q}_{1m} - \widehat{\boldsymbol{Q}}_{1m})^{\prime} \boldsymbol{b}_{1m,j\cdot}$ exists.
	\begin{equation*}
		\begin{aligned}
			& \widehat{\boldsymbol{\Pi}}_{ym,ij}(h) 	=  \frac{1}{T-h} \sum_{t=1}^{T-h} \boldsymbol{y}_{mt,i} \boldsymbol{y}_{m(t+h),j}^{\prime} = \frac{1}{T-h} \sum_{t=1}^{T-h} [(\boldsymbol{y}_{mt,i\cdot}^{*} +\boldsymbol{u}_{mt, i\cdot} )(\boldsymbol{y}_{m(t+h),j\cdot}^{*}+ \boldsymbol{u}_{m(t+h),j\cdot})^{\prime}\\
			& + \boldsymbol{Q}_{2m} \boldsymbol{S}_{mt}^{\prime} (\boldsymbol{Q}_{1m} - \widehat{\boldsymbol{Q}}_{1m})^{\prime} \boldsymbol{b}_{1m, i \cdot} (\boldsymbol{y}_{m(t+h),j}^{*}+ \boldsymbol{u}_{m(t+h),j\cdot})^{\prime} +  (\boldsymbol{y}_{mt,i\cdot}^{*}+ \boldsymbol{u}_{mt,i\cdot}) \boldsymbol{b}_{1m, j \cdot}^{\prime} (\boldsymbol{Q}_{1m} - \widehat{\boldsymbol{Q}}_{1m}) \boldsymbol{S}_{m(t+h)} \boldsymbol{Q}_{2m} ^{\prime} \\
			& +  \boldsymbol{Q}_{2m} \boldsymbol{S}_{mt}^{\prime} (\boldsymbol{Q}_{1m} - \widehat{\boldsymbol{Q}}_{1m})^{\prime} \boldsymbol{b}_{1m, i \cdot} \boldsymbol{b}_{1m, j \cdot}^{\prime} (\boldsymbol{Q}_{1m} - \widehat{\boldsymbol{Q}}_{1m}) \boldsymbol{S}_{m(t+h)} \boldsymbol{Q}_{2m} ^{\prime}] = A_1 + A_2 + A_3 + A_4.
		\end{aligned}
	\end{equation*}
By Lemma 3. in \citet{wang2019factor} we can get that
\begin{equation*}
	\sum_{i=1}^{N_m - k_1} \sum_{j=1}^{N_m -k_1} \| A_1 - \boldsymbol{\Pi}_{ym,ij}(h) \|_2^2 = O_p(N^2 p^2 T^{-1}).
\end{equation*}
	\begin{equation*}
		\begin{aligned}
			&\sum_{i=1}^{N_m - k_1} \sum_{j=1}^{N_m -k_1} \|A_4 \|_2^2 = 
			 \sum_{i=1}^{N_m -k_1} \sum_{j=1}^{N_m -k_1} \| \frac{1}{T-h}\sum_{t=1}^{T-h}\boldsymbol{C}_{m} \boldsymbol{G}_{t}^{\prime} (\boldsymbol{R}_{m} - \widehat{\boldsymbol{R}}_{m} )^{\prime}\boldsymbol{b}_{1m,i\cdot}  \boldsymbol{b}_{1m,j\cdot}^{\prime} (\boldsymbol{R}_{m} - \widehat{\boldsymbol{R}}_{m}) \boldsymbol{G}_{t+h} \boldsymbol{C}_m^{\prime} \|_2^2\\
			& \leq \|\boldsymbol{C}_m \|_2^4  \sum_{i=1}^{N_m -k_1} \sum_{j=1}^{N_m -k_1}\| \frac{1}{T-h}\sum_{t=1}^{T-h} \boldsymbol{G}_{t}^{\prime} (\boldsymbol{R}_{m} - \widehat{\boldsymbol{R}}_{m} )^{\prime}\boldsymbol{b}_{1m,i\cdot}  \boldsymbol{b}_{1m,j\cdot}^{\prime} (\boldsymbol{R}_{m} - \widehat{\boldsymbol{R}}_{m}) \boldsymbol{G}_{t+h}  \|_F^2 \\
			& = \|\boldsymbol{C}_m \|_2^4  \sum_{i=1}^{N_m -k_1} \sum_{j=1}^{N_m -k_1} \| \frac{1}{T-h}\sum_{t=1}^{T-h} \text{vec}(\boldsymbol{G}_{t}^{\prime} (\boldsymbol{R}_{m} - \widehat{\boldsymbol{R}}_{m} )^{\prime}\boldsymbol{b}_{1m,i\cdot}  \boldsymbol{b}_{1m,j\cdot}^{\prime} (\boldsymbol{R}_{m} - \widehat{\boldsymbol{R}}_{m}) \boldsymbol{G}_{t+h} ) \|_2^2 \\
			& = \|\boldsymbol{C}_m \|_2^4 \sum_{i=1}^{N_m -k_1} \sum_{j=1}^{N_m -k_1} \| \frac{1}{T-h}\sum_{t=1}^{T-h}(\boldsymbol{G}_{t+h} \otimes \boldsymbol{G}_t ) \text{vec}((\boldsymbol{R}_{m} - \widehat{\boldsymbol{R}}_{m} )^{\prime}\boldsymbol{b}_{1m,i\cdot}\boldsymbol{b}_{1m,j\cdot}^{\prime} (\boldsymbol{R}_{m} - \widehat{\boldsymbol{R}}_{m}))\|_2^2		\\
			& \leq \|\boldsymbol{C}_m \|_2^4 \sum_{i=1}^{N_m -k_1} \sum_{j=1}^{N_m -k_1} \| \frac{1}{T-h}\sum_{t=1}^{T-h}(\boldsymbol{G}_{t+h} \otimes \boldsymbol{G}_t )\|_2^2 \| \text{vec}((\boldsymbol{R}_{m} - \widehat{\boldsymbol{R}}_{m} )^{\prime}\boldsymbol{b}_{1m,i\cdot} \boldsymbol{b}_{1m,j\cdot}^{\prime} (\boldsymbol{R}_{m} - \widehat{\boldsymbol{R}}_{m}))\|_2^2	\\
			& \leq  \|\boldsymbol{C}_m \|_2^4  \| \frac{1}{T-h}\sum_{t=1}^{T-h}(\boldsymbol{G}_{t+h} \otimes \boldsymbol{G}_t )\|_F^2  (\sum_{i=1}^{N_m -k_1}\| (\boldsymbol{R}_{m} - \widehat{\boldsymbol{R}}_{m} )^{\prime}\boldsymbol{b}_{1m,i\cdot} \|_2^2)  (\sum_{j=1}^{N_m -k_1}\| (\boldsymbol{R}_{m} - \widehat{\boldsymbol{R}}_{m})^{\prime} \boldsymbol{b}_{1m,j\cdot}\|_2^2) \\
			& =  \|\boldsymbol{C}_m \|_2^4 \| \frac{1}{T-h}\sum_{t=1}^{T-h}(\boldsymbol{G}_{t+h} \otimes \boldsymbol{G}_t )\|_F^2  \| (\boldsymbol{R}_{m} - \widehat{\boldsymbol{R}}_{m} )^{\prime} \widehat{\boldsymbol{B}}_{1m} \|_F^4 \\
			& \leq k_1^2 \|\boldsymbol{C}_m \|_2^4 \| \frac{1}{T-h}\sum_{t=1}^{T-h}(\boldsymbol{G}_{t+h} \otimes \boldsymbol{G}_t )\|_F^2  \| (\boldsymbol{R}_{m} - \widehat{\boldsymbol{R}}_{m} )^{\prime} \widehat{\boldsymbol{B}}_{1m} \|_2^4 \\
			& \leq k_1^2 \|\boldsymbol{C}_m \|_2^4 \| (\boldsymbol{R}_{m} - \widehat{\boldsymbol{R}}_{m} ) \|_2^4 \| \frac{1}{T-h}\sum_{t=1}^{T-h}(\boldsymbol{G}_{t+h} \otimes \boldsymbol{G}_t )\|_F^2 \\
			& = O_p(p^{2-2\delta_2}) O_p(N^{2+2\delta_1} p^{4\delta_2}T^{-2}) O_p(1) = O_p(N^2 p^2 T^{-1} (N^{\delta_1}p^{\delta_2}T^{-1/2})^2).
		\end{aligned}
	\end{equation*}

	\begin{equation*}
		\begin{aligned}
		& \sum_{i=1}^{N_m - k_1} \sum_{j=1}^{N_m -k_1} \|A_2 \|_2^2  \leq  \sum_{i=1}^{N_m -k_1} \sum_{j=1}^{N_m -k_1} \| \frac{1}{T-h}\sum_{t=1}^{T-h}\boldsymbol{C}_{m} \boldsymbol{G}_{t}^{\prime} (\boldsymbol{R}_{m} - \widehat{\boldsymbol{R}}_{m} )^{\prime}\boldsymbol{b}_{1m,i\cdot}  \boldsymbol{b}_{1m,j\cdot}^{\prime} \boldsymbol{\Gamma}_{m} \boldsymbol{F}_{m (t+h)} \boldsymbol{\Lambda}_m^{\prime} \|_2^2 \\
		& +	\sum_{i=1}^{N_m -k_1} \sum_{j=1}^{N_m -k_1} \| \frac{1}{T-h}\sum_{t=1}^{T-h}\boldsymbol{C}_{m} \boldsymbol{G}_{t}^{\prime} (\boldsymbol{R}_{m} - \widehat{\boldsymbol{R}}_{m} )^{\prime}\boldsymbol{b}_{1m,i\cdot}  \boldsymbol{u}_{m(t+h),j\cdot}^{\prime} \|_2^2 .
		\end{aligned}	
	\end{equation*}
	Where
	\begin{equation*}
		\begin{aligned}
			& \sum_{i=1}^{N_m -k_1} \sum_{j=1}^{N_m -k_1} \| \frac{1}{T-h}\sum_{t=1}^{T-h}\boldsymbol{C}_{m} \boldsymbol{G}_{t}^{\prime} (\boldsymbol{R}_{m} - \widehat{\boldsymbol{R}}_{m} )^{\prime}\boldsymbol{b}_{1m,i\cdot}  \boldsymbol{u}_{m(t+h),j\cdot}^{\prime} \|_2^2\\
			&  \leq k_1\|\boldsymbol{C}_m \|_2^2 \| \boldsymbol{R}_{m} - \widehat{\boldsymbol{R}}_{m} \|_2^2 (\sum_{j=1}^{N_m - k_1} \left\|\frac{1}{T-h} \sum_{t-1}^{T-h} \left(\boldsymbol{u}_{m(t+h),j \cdot} \otimes \boldsymbol{G}_t\right)\right\|_F^2) = O_p(p^{1-\delta_2}) O_p(N^{1+\delta_1} p^{2\delta_2}T^{-1})  O_p(Np T^{-1}) \\
			& = O_p(N^2 p^2 T^{-1} ) O_p(N^{\delta_1} p^{\delta_2} T^{-1}).\\
			& \sum_{i=1}^{N_m -k_1} \sum_{j=1}^{N_m -k_1} \| \frac{1}{T-h}\sum_{t=1}^{T-h}\boldsymbol{C}_{m} \boldsymbol{G}_{t}^{\prime} (\boldsymbol{R}_{m} - \widehat{\boldsymbol{R}}_{m} )^{\prime}\boldsymbol{b}_{1m,i\cdot}  \boldsymbol{b}_{1m,j\cdot}^{\prime} \boldsymbol{\Gamma}_{m} \boldsymbol{F}_{m (t+h)} \boldsymbol{\Lambda}_m^{\prime} \|_2^2\\
			& \leq k_1 r_{m_3}\|\boldsymbol{C}_m \|_2^2 \|\boldsymbol{\Lambda}_m \|_2^2 \|\boldsymbol{R}_{m} - \widehat{\boldsymbol{R}}_{m} \|_2^2  \|\widehat{\boldsymbol{B}}_{1m}^{\prime} \boldsymbol{\Gamma}_m \|_2^2 \left\|\frac{1}{T-h} \sum_{t-1}^{T-h} \left(\boldsymbol{F}_{m(t+h)} \otimes \boldsymbol{G}_t\right)\right\|_F^2 \\
			& = O_p(N^2 p^2 T^{-1} ) O_p(N^{\delta_1 - \delta_3} p^{\delta_2 -\delta_4}).
		\end{aligned}
	\end{equation*}
	So \begin{equation*}
		\sum_{i=1}^{N_m - k_1} \sum_{j=1}^{N_m -k_1} \|A_2 \|_2^2 = O_p(N^2 p^2 T^{-1} N^{\delta_1 - \delta_3} p^{\delta_2 -\delta_4}).
	\end{equation*}
	Similarly, $ \sum_{i=1}^{N_m - k_1} \sum_{j=1}^{N_m -k_1} \|A_3 \|_2^2 = O_p(N^2 p^2 T^{-1} N^{\delta_1 - \delta_3} p^{\delta_2 -\delta_4}) $. Then we get the conclusion.
\end{proof}
\begin{lemma} \label{Lemma8}
	Under Conditions 1-7, $N^{\delta_1} p^{\delta_2} T^{-1/2} = o(1)$ and $N^{\delta_3} p^{\delta_4} T^{-1/2} = o(1)$, it holds that 
	\begin{equation*}
	\| \boldsymbol{M}_{1m} - \widehat{\boldsymbol{M}}_{1m}\|_2 = O_p\left({\rm max} \{N^{2 -\delta_3} p^{2-\delta_4} T^{-1/2},N^{2-\delta_3 } p^{2- \delta_4} T^{-1/2}N^{\delta_1 - \delta_3} p^{\delta_2 -\delta_4}\}\right).
	\end{equation*}
\end{lemma}
\begin{proof}
	From Lemma \ref{Lemma7} and Lemma 4 in \citet{wang2019factor}, we have this result.
\end{proof}
\begin{lemma} \label{Lemma9}
	Under Conditions 2 and 3, we have 
	\begin{equation*}
	\lambda_i (\boldsymbol{M}_{1m} ) \asymp N^{2-2\delta_3} p^{2-2\delta_4} , i= 1,...,r_{m_2}, 
	\end{equation*}
	where $\lambda_i (\boldsymbol{M}_{1m} )$ denotes the $i$-th largest eigenvalue of $\boldsymbol{M}_{1m}$.
\end{lemma}
\begin{proof}
	From the Lemma 5 in \cite{wang2019factor}, we get the conclusion.
\end{proof}
\noindent{\bf Proof of Theorem 3}
\begin{proof}
	By Lemmas \ref{Lemma8}-\ref{Lemma9}, and Lemma 3 in \citet{10.1093/biomet/asr048}, we have
	\begin{equation}
		\| \widehat{\boldsymbol{Q}}_{4m} - \boldsymbol{Q}_{4m} \|_2 \leq \frac{8}{\lambda_{min}(\boldsymbol{M}_{1m})} \| \widehat{\boldsymbol{M}}_{1m} -\boldsymbol{M}_{1m} \|_2 = \text{max}
		\{O_p ( N^{\delta_3} p^{\delta_4} T^{-1/2}), O_p ( N^{\delta_1} p^{\delta_2} T^{-1/2})\} .
	\end{equation}
\end{proof}
Proof for $\| \widehat{\boldsymbol{Q}}_{3m} - \boldsymbol{Q}_{3m}\|_2$ is similar. 
\begin{lemma} \label{Lemma10}
Under Condition 4 we have, 
\begin{equation*}
  \| \boldsymbol{A}_m \|_{\text{min}}^2 \asymp O(1), \quad m=1,...,M.
  \end{equation*}
\begin{proof}
Through the QR decomposition, we get 
	\begin{equation*}
	N^{1- \delta_3}	 \asymp \| \boldsymbol{\Gamma}_m \|_2^2  = \| \boldsymbol{Q}_{3m} \boldsymbol{K}_{3m} \|_2^2 = \|\boldsymbol{K}_{3m} \|_2^2.
	\end{equation*}
Similarly, 
\begin{equation}
	N^{1- \delta_3}	 \asymp \| \boldsymbol{\Gamma}_m \|_{\text{min}}^2  = \|\boldsymbol{K}_{3m} \|_{\text{min}}^2.
\end{equation}
By the Theorem 7 in \citet{merikoski2004inequalities}, we have
\begin{equation}
\begin{aligned}
	N^{1-\delta_3} \asymp	\| \widehat{\boldsymbol{B}}_{1m}^{\prime} \boldsymbol{Q}_{3m} \boldsymbol{K}_{3m} \|_{\text{min}}^2  & \leq \lambda_{r_{m1}}(\boldsymbol{Q}_{3m}^{\prime}\widehat{\boldsymbol{B}}_{1m} \widehat{\boldsymbol{B}}_{1m}^{\prime} \boldsymbol{Q}_{3m}) \lambda_1 (\boldsymbol{K}_{3m} \boldsymbol{K}_{3m}^{\prime}) = \|\widehat{\boldsymbol{B}}_{1m}^{\prime} \boldsymbol{Q}_{3m}\|_{\text{min}}^2 \| \boldsymbol{K}_{3m} \|_2^2 ,\\
	N^{1-\delta_3} \asymp	\| \widehat{\boldsymbol{B}}_{1m}^{\prime} \boldsymbol{Q}_{3m} \boldsymbol{K}_{3m} \|_{\text{min}}^2  & \geq \lambda_{r_{m1}}(\boldsymbol{Q}_{3m}^{\prime}\widehat{\boldsymbol{B}}_{1m} \widehat{\boldsymbol{B}}_{1m}^{\prime} \boldsymbol{Q}_{3m}) \lambda_{r_{m1}} (\boldsymbol{K}_{3m} \boldsymbol{K}_{3m}^{\prime}) = \|\widehat{\boldsymbol{B}}_{1m}^{\prime} \boldsymbol{Q}_{3m}\|_{\text{min}}^2 \| \boldsymbol{K}_{3m} \|_{\text{min}}^2 .
\end{aligned}	   
\end{equation}
Then we get the conclusion.
\end{proof}
\end{lemma}
\begin{lemma} \label{Lemma11}
Under Conditions 3-4, it follows that
\begin{equation*}
	\|(\widehat{\boldsymbol{Q}}_{4m}^{\prime} \otimes \boldsymbol{A}_m (\boldsymbol{A}_{m}^{\prime} \boldsymbol{A}_{m} )^{-1} ) {\rm vec}(\boldsymbol{U}_{mt}) \|_2 = O_p(1) , m=1,...,M.
\end{equation*}
\begin{proof}
	Let the $i$th row of $\widehat{\boldsymbol{B}}_{1m}^{\prime}$, $\widehat{\boldsymbol{Q}}_{4m}$ , $\boldsymbol{A}_m (\boldsymbol{A}_{m}^{\prime} \boldsymbol{A}_{m} )^{-1}$ and $\widehat{\boldsymbol{Q}}_{4m}^{\prime} \otimes \boldsymbol{A}_m (\boldsymbol{A}_{m}^{\prime} \boldsymbol{A}_{m} )^{-1} $ be $\boldsymbol{b}_{i\cdot}$, $\boldsymbol{q}_{i\cdot}$, $\boldsymbol{a}_{i \cdot}$ and $\boldsymbol{c}_{i \cdot}$, respectively. Consider the length of $\boldsymbol{a}_{i\cdot}$ and by Lemma \ref{Lemma10}:
	\begin{equation}
		\begin{aligned}
			\boldsymbol{a}_{i\cdot} \boldsymbol{a}_{i\cdot}^{\prime} & = \boldsymbol{b}_{i\cdot} \widehat{\boldsymbol{Q}}_{3m} (\boldsymbol{A}_{m}^{\prime} \boldsymbol{A}_{m} )^{-1} (\boldsymbol{A}_{m}^{\prime} \boldsymbol{A}_{m} )^{-1} \widehat{\boldsymbol{Q}}_{3m}^{\prime} \boldsymbol{b}_{i\cdot}^{\prime} \leq \lambda_{\text{max}} (\widehat{\boldsymbol{Q}}_{3m} (\boldsymbol{A}_{m}^{\prime} \boldsymbol{A}_{m} )^{-1} (\boldsymbol{A}_{m}^{\prime} \boldsymbol{A}_{m} )^{-1} \widehat{\boldsymbol{Q}}_{3m}^{\prime})\\
			& = \|\widehat{\boldsymbol{Q}}_{3m} (\boldsymbol{A}_{m}^{\prime} \boldsymbol{A}_{m} )^{-1}\|_2^2 \leq \|(\boldsymbol{A}_{m}^{\prime} \boldsymbol{A}_{m})^{-1} \|_2^2 = (\frac{1}{\|\boldsymbol{A}_m \|_{\text{min}}})^2 = O(1), i = 1,...,r_{m1}.
		\end{aligned}			
	\end{equation}
Because each row of $\widehat{\boldsymbol{Q}}_{4m}$ has length 1, the length of $\boldsymbol{c}_{i\cdot}(i=1,..., r_{m1}r_{m2})$ is also bounded. Under Lemma \ref{Lemma6} we get
\begin{equation}
	\|(\widehat{\boldsymbol{Q}}_{4m}^{\prime} \otimes \boldsymbol{A}_m (\boldsymbol{A}_{m}^{\prime} \boldsymbol{A}_{m} )^{-1}  ) \text{vec}(\boldsymbol{U}_{mt}) \|_2 \leq \sum_{i=1}^{r_{m_1} r_{m_2}}  (\boldsymbol{c}_{i\cdot} \text{vec}(\boldsymbol{U}_{mt}) )^2 =O_p(1).
\end{equation}
\end{proof}
\end{lemma}
\noindent{\bf Proof of Theorem 4}
\begin{proof}
We decompose the difference of the local signal part and the estimated local signal part into five parts:
	\begin{equation}
		\begin{aligned}
		&	\widehat{\boldsymbol{Q}}_{3m}\widehat{\boldsymbol{Z}}_{mt} \widehat{\boldsymbol{Q}}_{4m}^{\prime} - \boldsymbol{Q}_{3m}\boldsymbol{Z}_{mt} \boldsymbol{Q}_{4m}^{\prime}  = 	\widehat{\boldsymbol{Q}}_{3m}(\boldsymbol{A}_{m}^{\prime} \boldsymbol{A}_{m} )^{-1}\boldsymbol{A}_{m}^{\prime}\boldsymbol{Y}_{mt}^{(1)} \widehat{\boldsymbol{Q}}_{4m} \widehat{\boldsymbol{Q}}_{4m}^{\prime} -  {\boldsymbol{Q}}_{3m}\boldsymbol{Z}_{mt} {\boldsymbol{Q}}_{4m}^{\prime}\\
			& =\widehat{\boldsymbol{Q}}_{3m} (\boldsymbol{A}_{m}^{\prime} \boldsymbol{A}_{m} )^{-1}\boldsymbol{A}_{m}^{\prime}(\widehat{\boldsymbol{B}}^{\prime}_{1m}\boldsymbol{Q}_{3m} \boldsymbol{Z}_{mt}  \boldsymbol{Q}_{4m}^{\prime}+\boldsymbol{U}_{mt} + \widehat{\boldsymbol{B}}_{1m}^{\prime} (\boldsymbol{Q}_{1m} - \widehat{\boldsymbol{Q}}_{1m}) \boldsymbol{S}_{mt} \boldsymbol{Q}_{2m}^{\prime})\widehat{\boldsymbol{Q}}_{4m} \widehat{\boldsymbol{Q}}_{4m}^{\prime} \\
			& -{\boldsymbol{Q}}_{3m} {\boldsymbol{Q}}_{3m}^{\prime} {\boldsymbol{Q}}_{3m} \boldsymbol{Z}_{mt} {\boldsymbol{Q}}_{4m}^{\prime} {\boldsymbol{Q}}_{4m} {\boldsymbol{Q}}_{4m}^{\prime}\\
			& = \widehat{\boldsymbol{Q}}_{3m} (\boldsymbol{A}_{m}^{\prime} \boldsymbol{A}_{m} )^{-1}\boldsymbol{A}_{m}^{\prime}\widehat{\boldsymbol{B}}^{\prime}_{1m}\boldsymbol{Q}_{3m} \boldsymbol{Z}_{mt} \boldsymbol{Q}_{4m}^{\prime}(\widehat{\boldsymbol{Q}}_{4m}\widehat{\boldsymbol{Q}}_{4m}^{\prime} - \boldsymbol{Q}_{4m} \boldsymbol{Q}_{4m}^{\prime}) + \widehat{\boldsymbol{Q}}_{3m}(\boldsymbol{A}_{m}^{\prime} \boldsymbol{A}_{m} )^{-1}\boldsymbol{A}_{m}^{\prime} \boldsymbol{U}_{mt} \widehat{\boldsymbol{Q}}_{4m} \widehat{\boldsymbol{Q}}_{4m}^{\prime} \\
			& + \widehat{\boldsymbol{Q}}_{3m} (\boldsymbol{A}_{m}^{\prime} \boldsymbol{A}_{m} )^{-1}\boldsymbol{A}_{m}^{\prime}\widehat{\boldsymbol{B}}_{1m}^{\prime} (\boldsymbol{Q}_{1m} - \widehat{\boldsymbol{Q}}_{1m}) \boldsymbol{S}_{mt} \boldsymbol{Q}_{2m}^{\prime}\widehat{\boldsymbol{Q}}_{4m} \widehat{\boldsymbol{Q}}_{4m}^{\prime} +(\widehat{\boldsymbol{Q}}_{3m}  -
			\boldsymbol{Q}_{3m}) \boldsymbol{Z}_{mt} \boldsymbol{Q}_{4m}^{\prime}\\
			& + \widehat{\boldsymbol{Q}}_{3m} (\boldsymbol{A}_{m}^{\prime} \boldsymbol{A}_{m} )^{-1}\boldsymbol{A}_{m}^{\prime}\widehat{\boldsymbol{B}}^{\prime}_{1m} 
			(\boldsymbol{Q}_{3m}-\widehat{\boldsymbol{Q}}_{3m}) \boldsymbol{Z}_{mt} \boldsymbol{Q}_{4m}^{\prime} \\
			& = I_1 + I_2 + I_3 + I_4 + I_5,
		\end{aligned}			
	\end{equation}
where 
\begin{equation}
	\|\boldsymbol{A}_m \|_2 \leq \|\widehat{\boldsymbol{B}}_{1m}^{\prime} \|_2 \|\widehat{\boldsymbol{Q}}_{3m} \|_2 = 1.
\end{equation}
By Theorem \ref{T1}, Theorem \ref{T3} and Lemma \ref{Lemma11}, we have
	\begin{equation}
		\begin{aligned}
			\| I_1 \|_2 & \leq 2 \|(\boldsymbol{A}_m^{\prime} \boldsymbol{A}_m )^{-1} \|_2 \|\boldsymbol{Z}_{mt} \|_2 \| \widehat{\boldsymbol{Q}}_{4m} - \boldsymbol{Q}_{4m}  \|_2   \\
			& = \frac{2}{\|\boldsymbol{A}_m \|_{\text{min}}}  O_p(N^{1/2 -\delta_3/2}p^{1/2 -\delta_4/2}) \text{max}\{O_p(N^{\delta_1 } p^{\delta_2} T^{-1/2}), O_p(N^{\delta_3 } p^{\delta_4} T^{-1/2}) \} \\
			& =O_p(N^{1/2 +\delta_1/2}p^{1/2 +\delta_2/2}T^{-1/2}) + O_p(N^{1/2 +\delta_3/2}p^{1/2 +\delta_4/2}T^{-1/2}) ,\\
			\| I_2 \|_2 & \leq \| (\boldsymbol{A}_{m}^{\prime} \boldsymbol{A}_{m} )^{-1}\boldsymbol{A}_{m}^{\prime} \boldsymbol{U}_{mt} \widehat{\boldsymbol{Q}}_{4m} \|_2 = \|(\widehat{\boldsymbol{Q}}_{4m}^{\prime} \otimes \boldsymbol{A}_m (\boldsymbol{A}_{m}^{\prime} \boldsymbol{A}_{m} )^{-1}  ) \text{vec}(\boldsymbol{U}_{mt}) \|_2 = O_p(1), \\
			\| I_3 \|_2 &  \leq \|(\boldsymbol{A}_m^{\prime} \boldsymbol{A}_m )^{-1} \|_2 \|\widehat{\boldsymbol{Q}}_{1m} - \boldsymbol{Q}_{1m}  \|_2  \| \boldsymbol{S}_{mt}\|_2 = O_p(N^{1/2 +\delta_1/2}p^{1/2 +\delta_2/2}T^{-1/2}), \\
			\| I_4 \|_2 & \leq 2\| \boldsymbol{Q}_{3m} - \widehat{\boldsymbol{Q}}_{3m} \|_2 \| \boldsymbol{Z}_{mt} \|_2  = O_p(N^{1/2 +\delta_1/2}p^{1/2 +\delta_2/2}T^{-1/2}) + O_p(N^{1/2 +\delta_3/2}p^{1/2 +\delta_4/2}T^{-1/2}),\\
			\| I_5 \|_2 & \leq \|(\boldsymbol{A}_m^{\prime} \boldsymbol{A}_m )^{-1} \|_2 \| \widehat{\boldsymbol{Q}}_{3m}
			- \boldsymbol{Q}_{3m}  \|_2 \|\boldsymbol{Z}_{mt} \|_2  \\
			& =(\frac{2}{\|\boldsymbol{A}_m \|_{\text{min}}})\| \widehat{\boldsymbol{Q}}_{3m}
			- \boldsymbol{Q}_{3m}  \|_2 \|\boldsymbol{Z}_{mt} \|_2 = O_p(N^{1/2 +\delta_1/2}p^{1/2 +\delta_2/2}T^{-1/2}) + O_p(N^{1/2 +\delta_3/2}p^{1/2 +\delta_4/2}T^{-1/2}) . 
		\end{aligned}
	\end{equation}
Then we get the conclusion.
 
\end{proof}

\noindent{\bf Proof of Theorem 5}
\begin{proof}
    \begin{equation}
		\begin{aligned}
			\widehat{\boldsymbol{\Phi}}_{mt} - \boldsymbol{\Phi}_{mt} & = \widehat{\boldsymbol{Q}}_{1m} \widehat{\boldsymbol{Q}}_{1m}^{\prime} (\boldsymbol{X}_{mt} - \widehat{\boldsymbol{Q}}_{3m}\widehat{\boldsymbol{Z}}_{mt} \widehat{\boldsymbol{Q}}_{4m}^{\prime}) \widehat{\boldsymbol{Q}}_{2m} \widehat{\boldsymbol{Q}}_{2m}^{\prime} - \boldsymbol{Q}_{1m} \boldsymbol{S}_{mt} \boldsymbol{Q}_{2m}^{\prime} \\
			& = [\widehat{\boldsymbol{Q}}_{1m} \widehat{\boldsymbol{Q}}_{1m}^{\prime} (\boldsymbol{Q}_{1m} \boldsymbol{S}_{mt} \boldsymbol{Q}_{2m}^{\prime} + \boldsymbol{E}_{mt})\widehat{\boldsymbol{Q}}_{2m}^{\prime} \widehat{\boldsymbol{Q}}_{2m}- \boldsymbol{Q}_{1m} \boldsymbol{S}_{mt} \boldsymbol{Q}_{2m}^{\prime}]+ \widehat{\boldsymbol{Q}}_{1m} \widehat{\boldsymbol{Q}}_{1m}^{\prime}(\widehat{\boldsymbol{\Psi}}_{mt} - \boldsymbol{\Psi}_{mt})\widehat{\boldsymbol{Q}}_{2m} \widehat{\boldsymbol{Q}}_{2m}^{\prime}\\
			& = I_1  + I_2.
		\end{aligned}
	\end{equation}
	By the proof of Theorem 3 in \cite{wang2019factor}, we have 
	\begin{equation}
		\begin{aligned}
			\|I_1\|_2 & \leq O_p(N^{1/2 +\delta_1/2}p^{1/2 +\delta_2/2}T^{-1/2} + 1). \\
			\|I_2\|_2 & \leq \|\widehat{\boldsymbol{\Psi}}_{mt} - \boldsymbol{\Psi}_{mt}\|_2 = O_p(N^{1/2 +\delta_1/2}p^{1/2 +\delta_2/2}T^{-1/2} +N^{1/2 +\delta_3/2}p^{1/2 +\delta_4/2}T^{-1/2} + 1).
		\end{aligned}   		 
	\end{equation}
Then
\begin{equation}
	N^{-1/2} p^{-1/2} \| \widehat{\boldsymbol{\Phi}}_{mt} - \boldsymbol{\Phi}_{mt} \|_2 = O_p(N^{\delta_1/2}p^{ \delta_2/2}T^{-1/2} + N^{\delta_3/2}p^{\delta_4/2}T^{-1/2} + N^{-1/2} p^{-1/2}).
\end{equation}
\end{proof}

\noindent
{\Large \bf Appendix 2: More results and images in the real data analysis}
Figure \ref{materials}-\ref{real estate} show the time series plot of other industries.
\begin{figure}[htbp]
	\centering
	\includegraphics[width=0.8\linewidth]{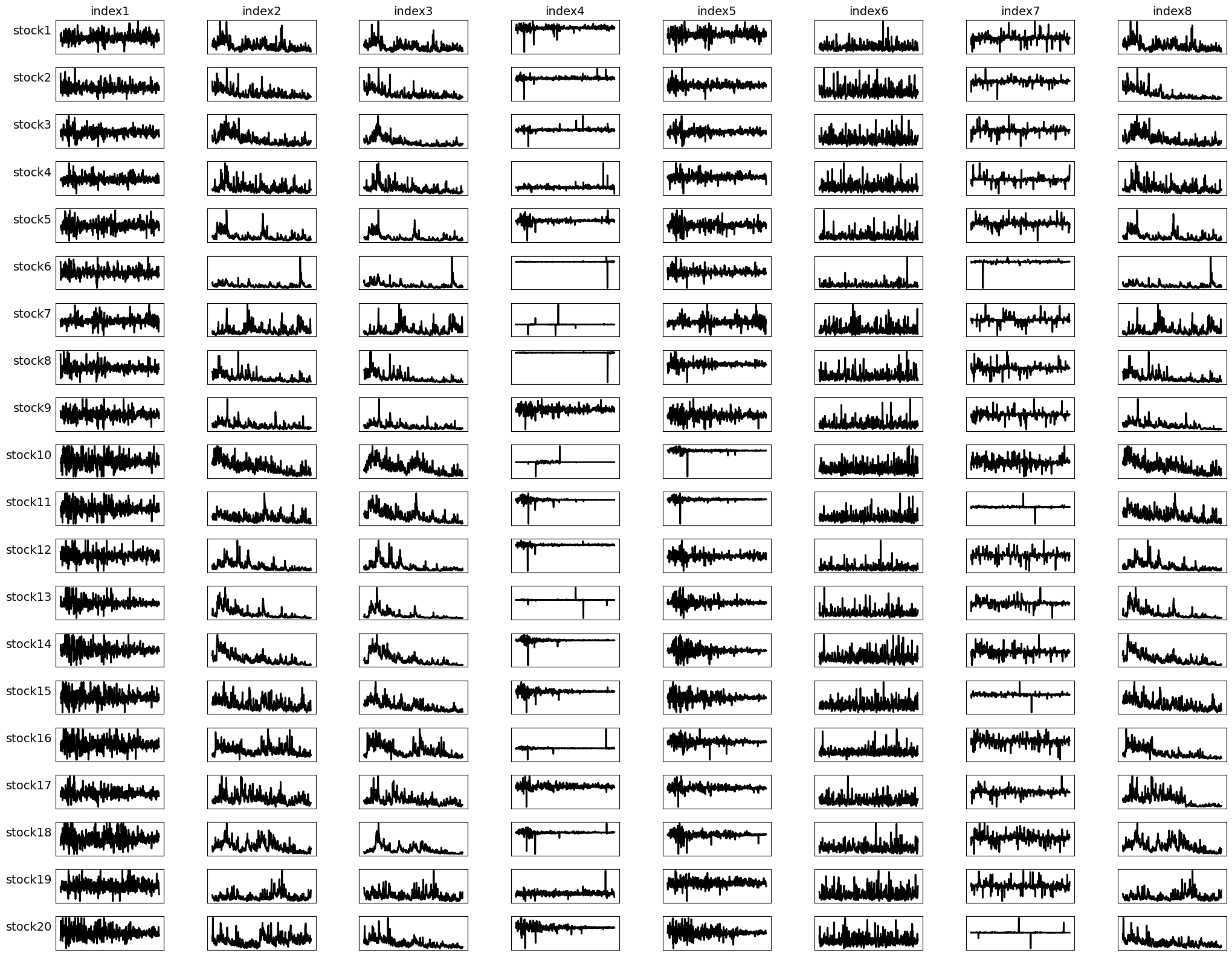}
	\caption{Time series plot of materials 20 by 8 series (after data transformation)}
	\label{materials}
\end{figure}

\begin{figure}[htbp]
	\centering
	\includegraphics[width=0.8\linewidth]{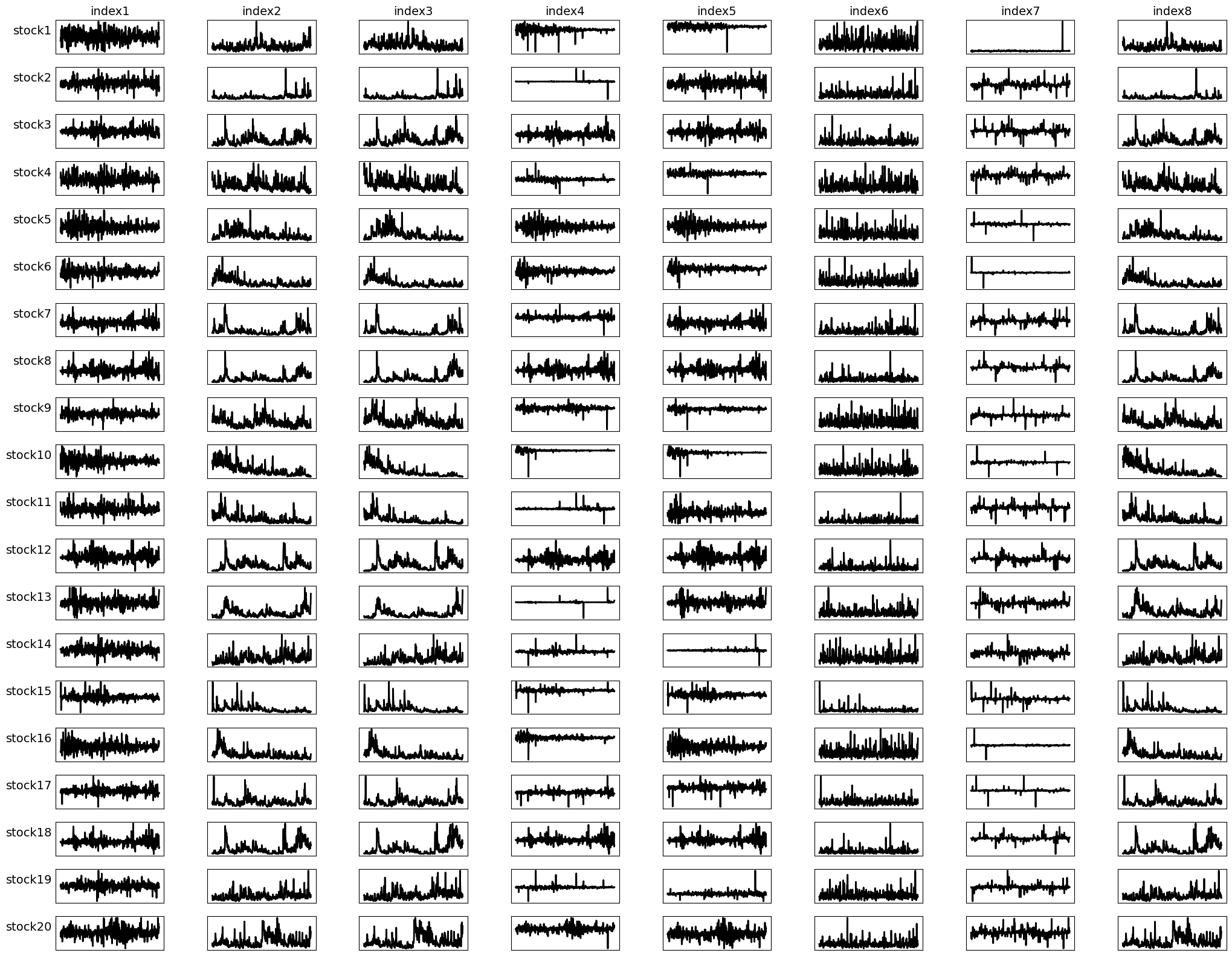}
	\caption{Time series plot of industrials 20 by 8 series (after data transformation)}
\end{figure}

\begin{figure}[htbp]
	\centering
	\includegraphics[width=0.8\linewidth]{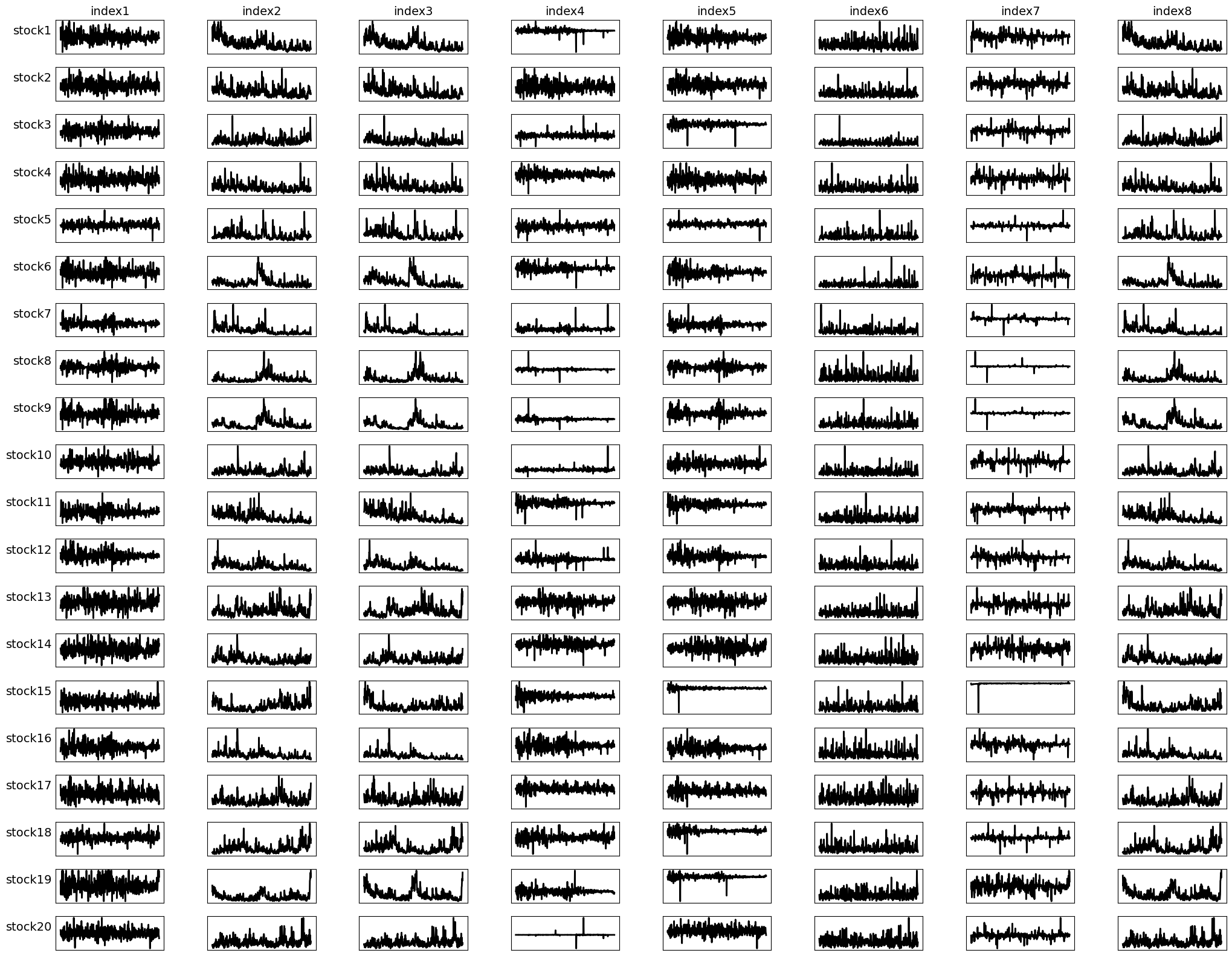}
	\caption{Time series plot of consumer discretionary 20 by 8 series (after data transformation)}
\end{figure}

\begin{figure}[htbp]
	\centering
	\includegraphics[width=0.8\linewidth]{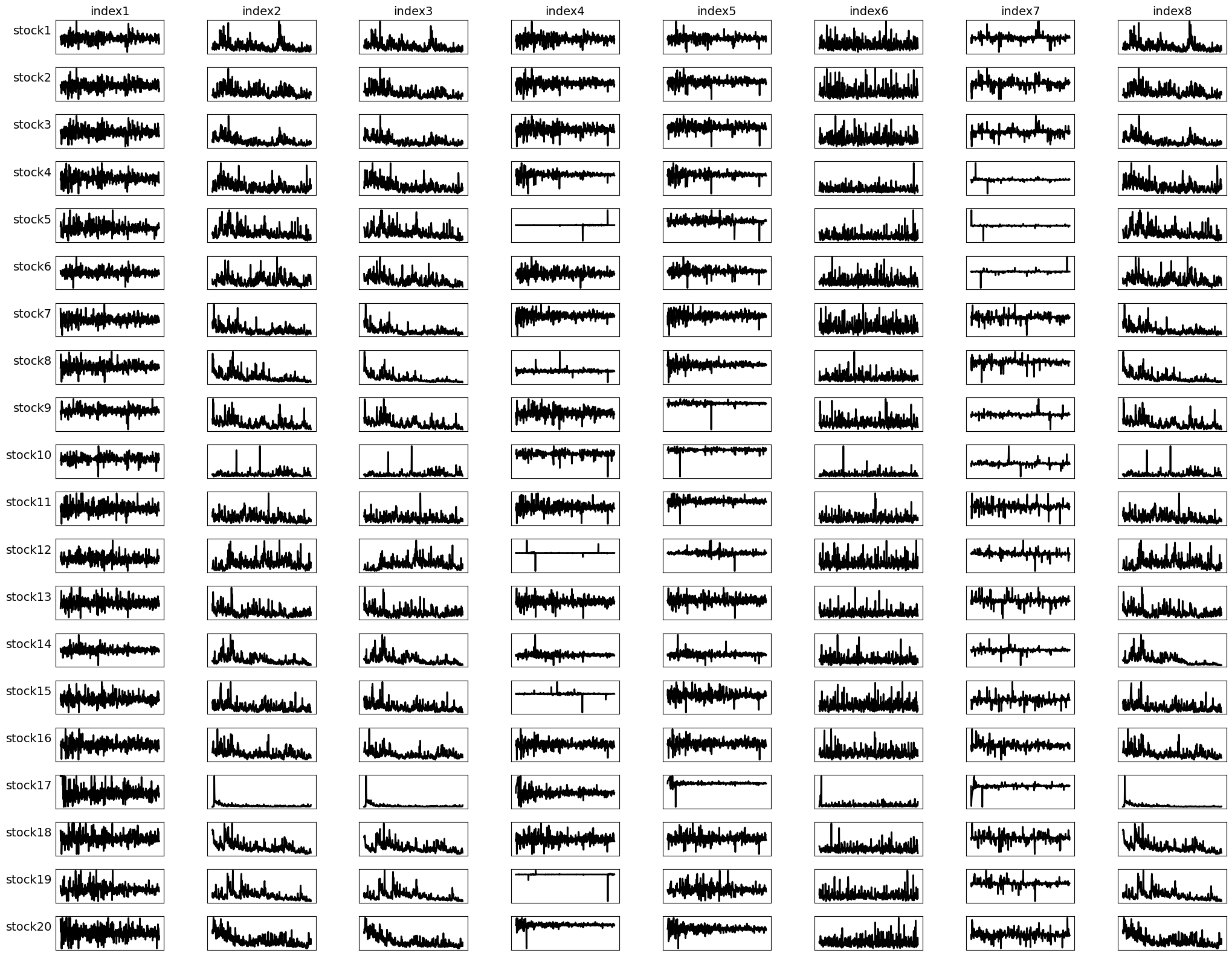}
	\caption{Time series plot of consumer staples 20 by 8 series (after data transformation)}
\end{figure}

\begin{figure}[htbp]
	\centering
	\includegraphics[width=0.8\linewidth]{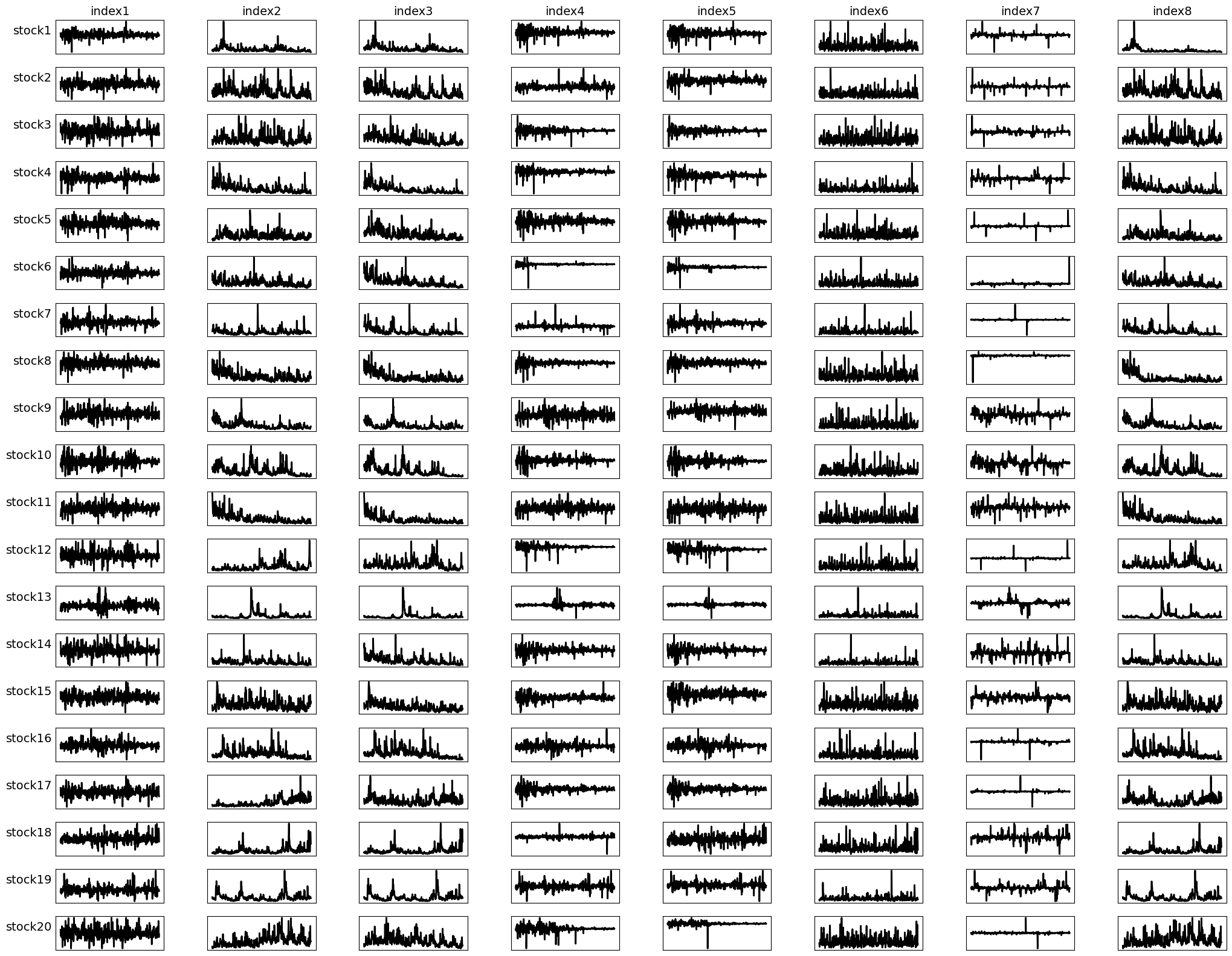}
	\caption{Time series plot of health care 20 by 8 series (after data transformation)}
\end{figure}

\begin{figure}[htbp]
	\centering
	\includegraphics[width=0.8\linewidth]{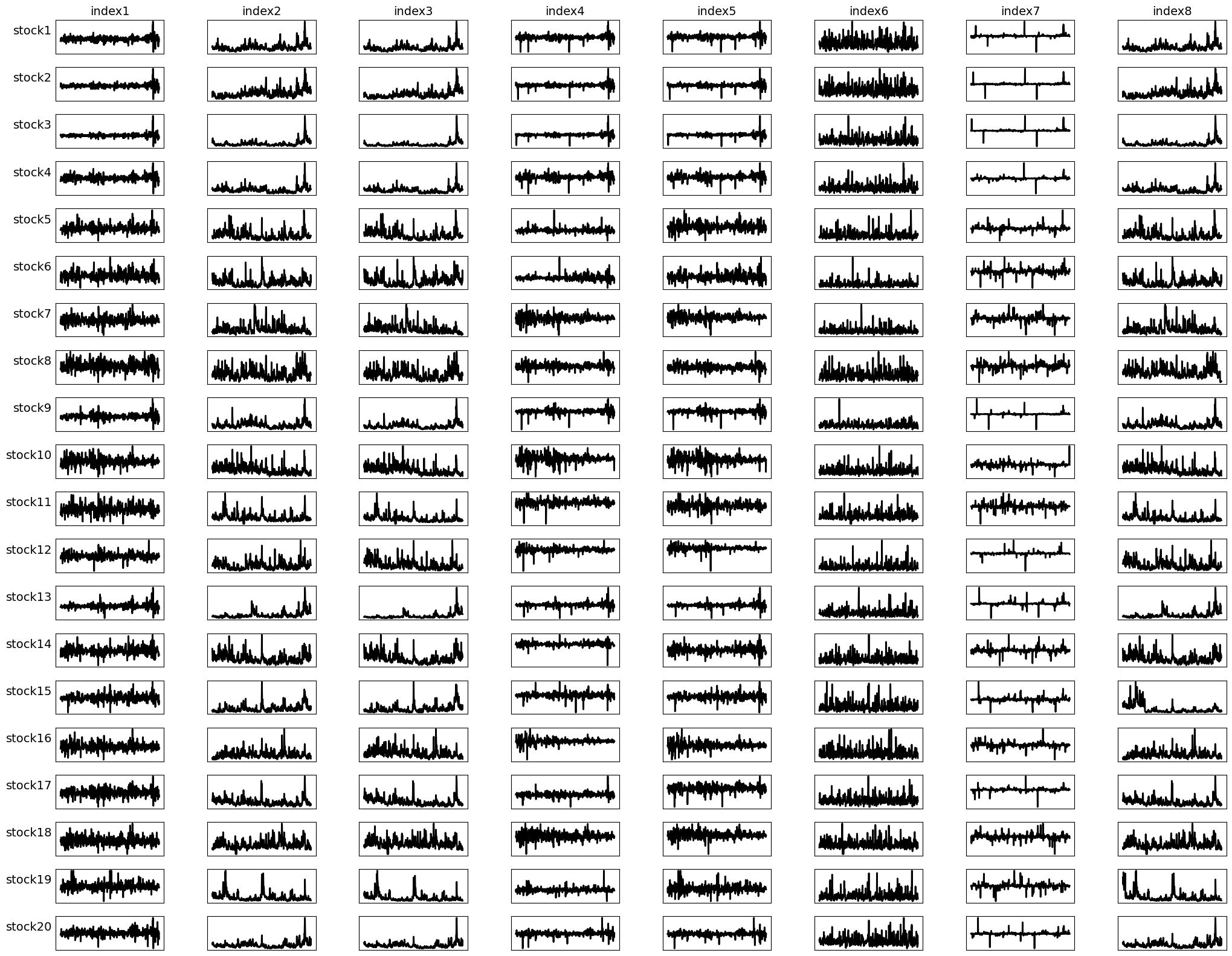}
	\caption{Time series plot of Financials 20 by 8 series (after data transformation)}
\end{figure}

\begin{figure}[htbp]
	\centering
	\includegraphics[width=0.8\linewidth]{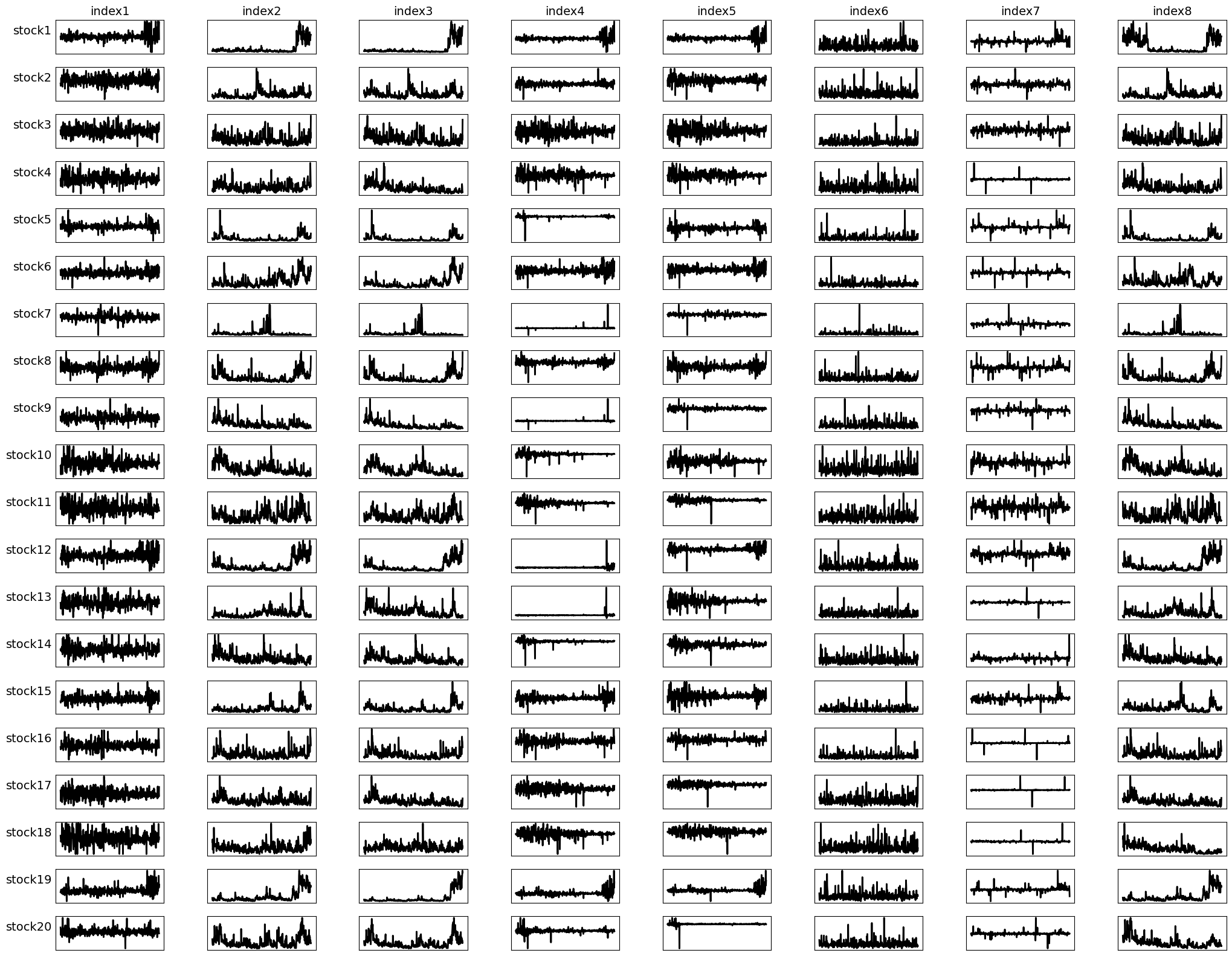}
	\caption{Time series plot of information technology 20 by 8 series (after data transformation)}
\end{figure}

\begin{figure}[htbp]
	\centering
	\includegraphics[width=0.8\linewidth]{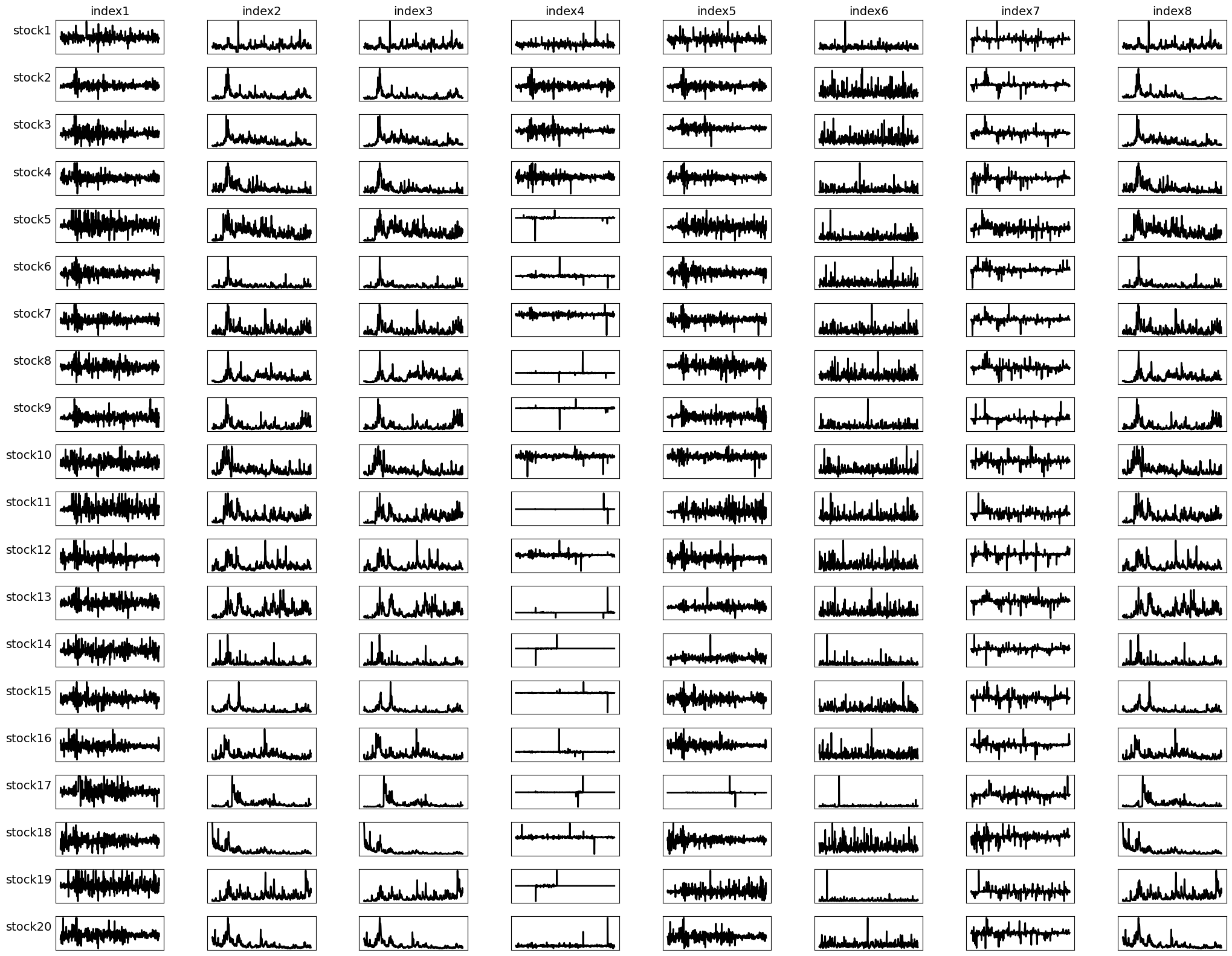}
	\caption{Time series plot of utilities 20 by 8 series (after data transformation)}
\end{figure}

\begin{figure}[htbp]
	\centering
	\includegraphics[width=0.8\linewidth]{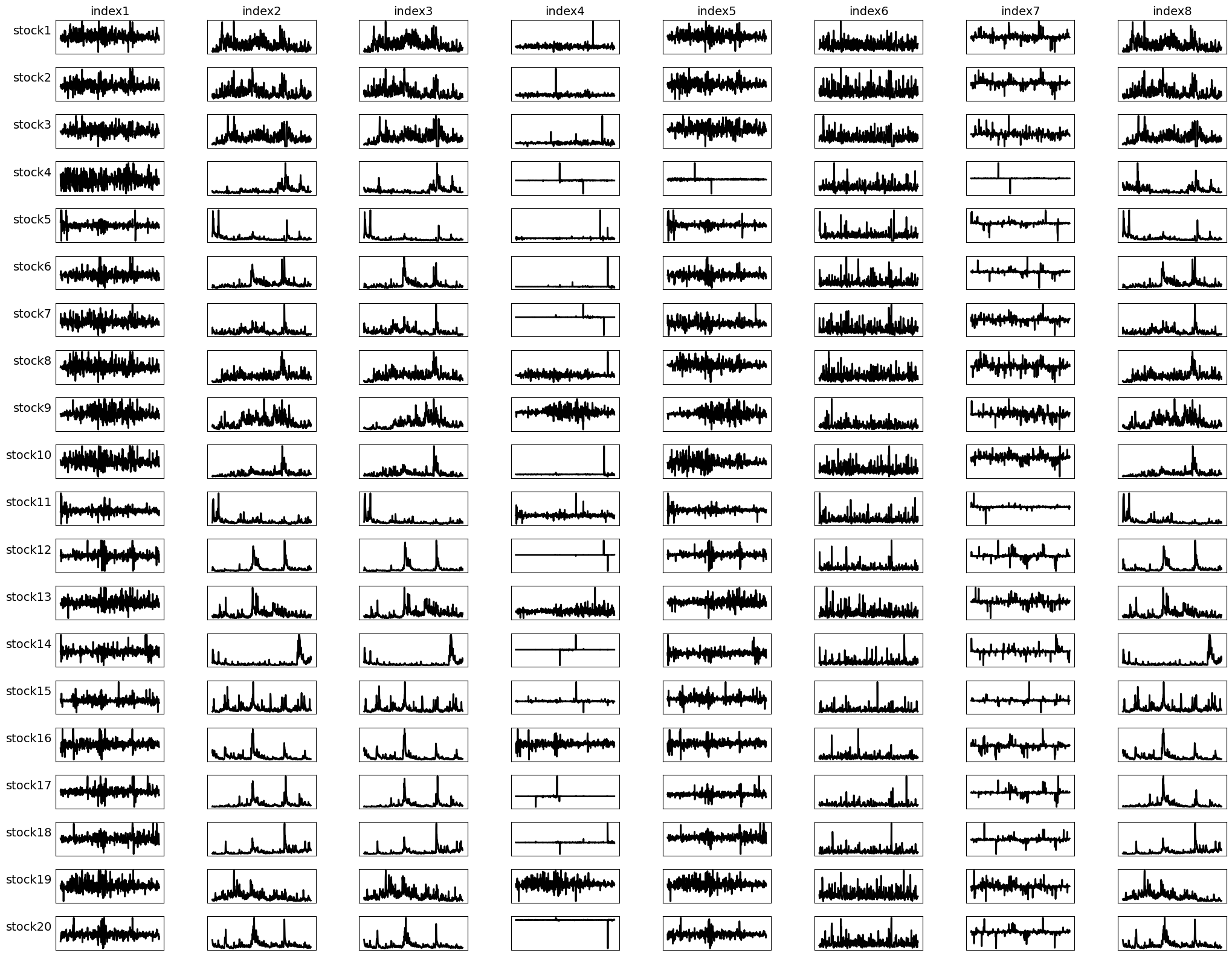}
	\caption{Time series plot of real estate 20 by 8 series (after data transformation)}
	\label{real estate}
\end{figure}

Figure \ref{materials1}-\ref{real estate1} show the logarithms and ratios
of eigenvalues of other industries.
\begin{figure}[htbp]
	\centering
        \includegraphics[width=0.95\linewidth]{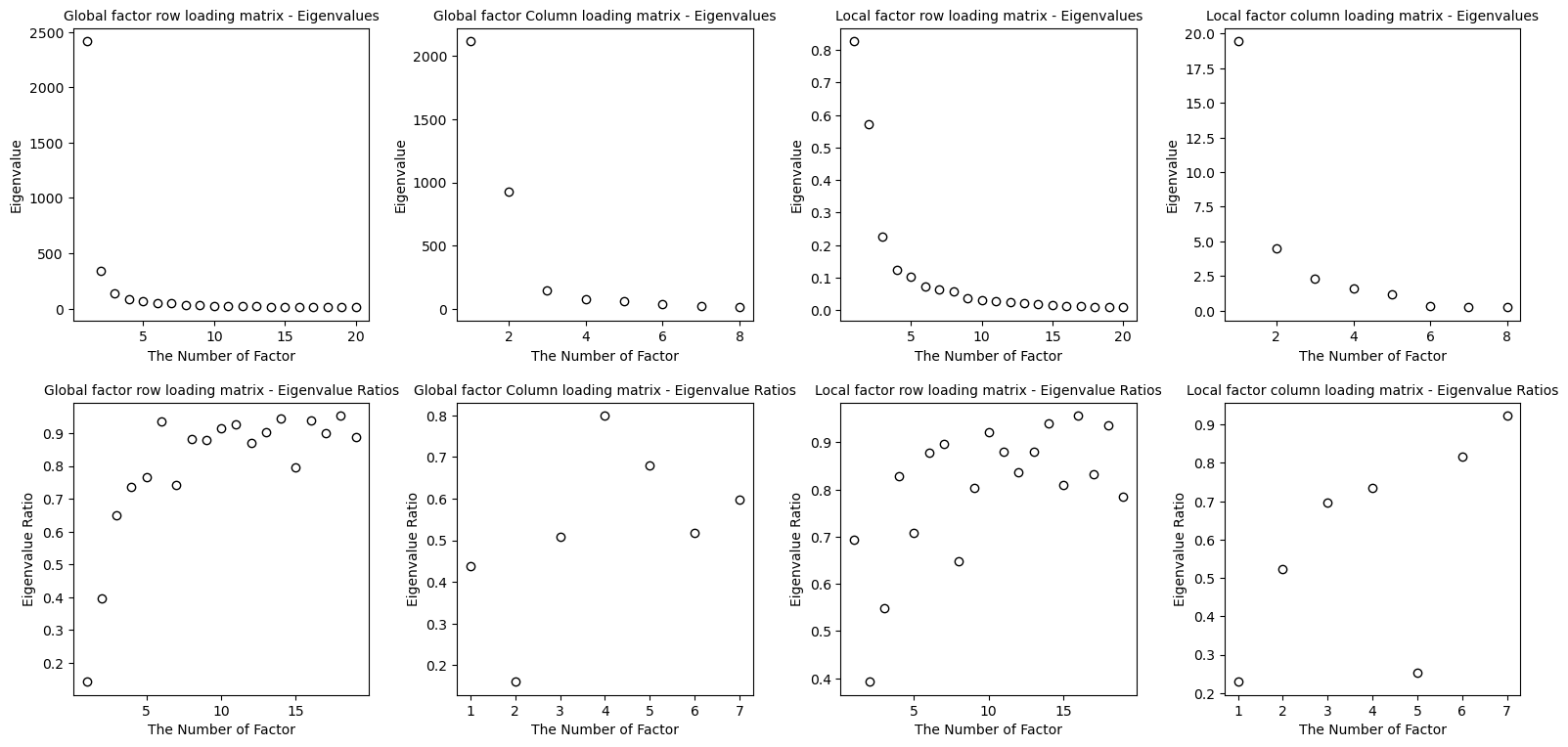}
	\caption{Materials series: eigenvalues and eigenvalues ratio for determining the number of global and local factors }
	\label{materials1}
\end{figure}

\begin{figure}[htbp]
	\centering
        \includegraphics[width=0.95\linewidth]{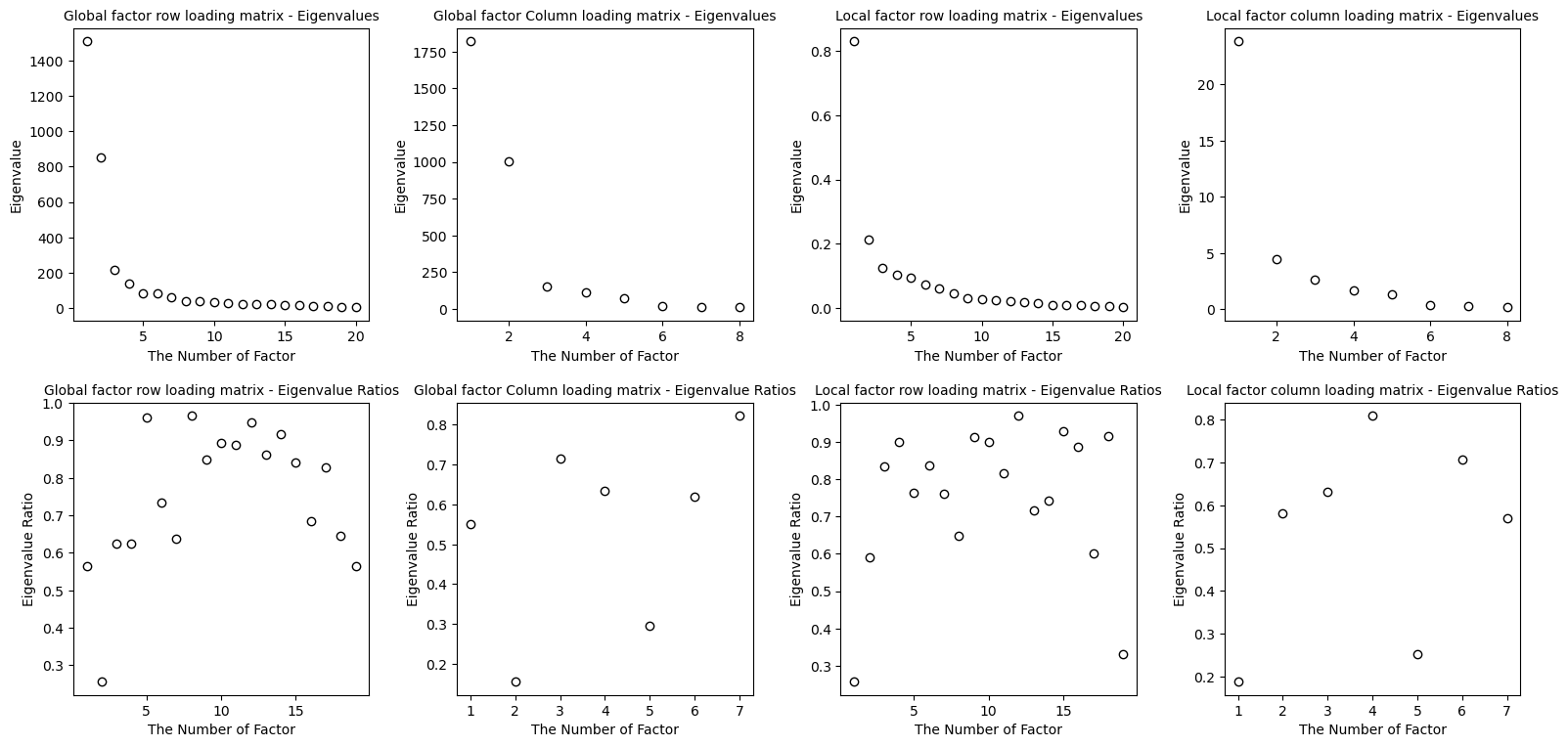}
	\caption{Industrials series: eigenvalues and eigenvalues ratio for determining the number of global and local factors }
\end{figure}

\begin{figure}[htbp]
	\centering
        \includegraphics[width=0.95\linewidth]{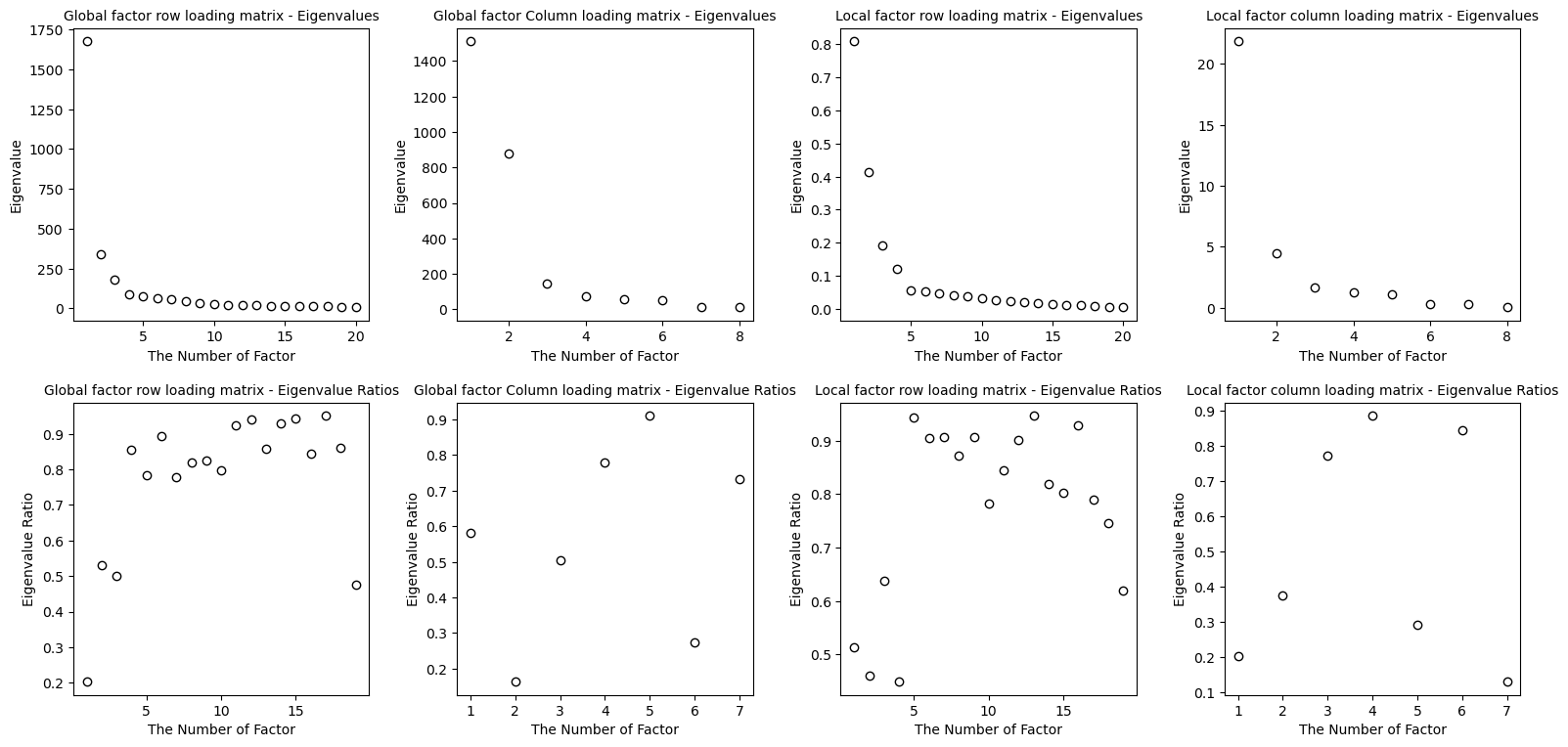}
	\caption{Consumer discretionary series: eigenvalues and eigenvalues ratio for determining the number of global and local factors }
\end{figure}

\begin{figure}[htbp]
	\centering
        \includegraphics[width=0.95\linewidth]{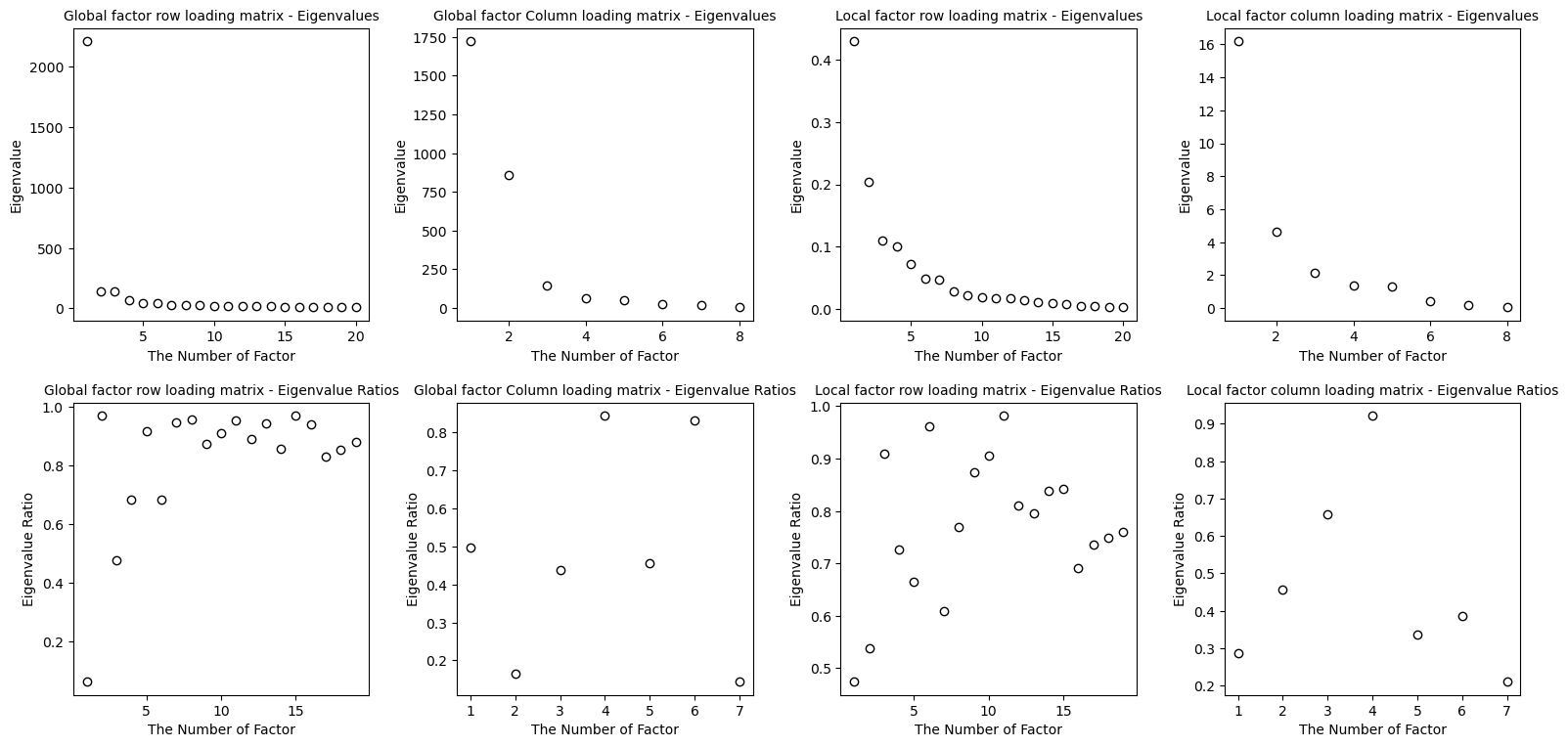}
	\caption{Consumer staple series: eigenvalues and eigenvalues ratio for determining the number of global and local factors }
\end{figure}
\begin{figure}[htbp]
	\centering
        \includegraphics[width=0.95\linewidth]{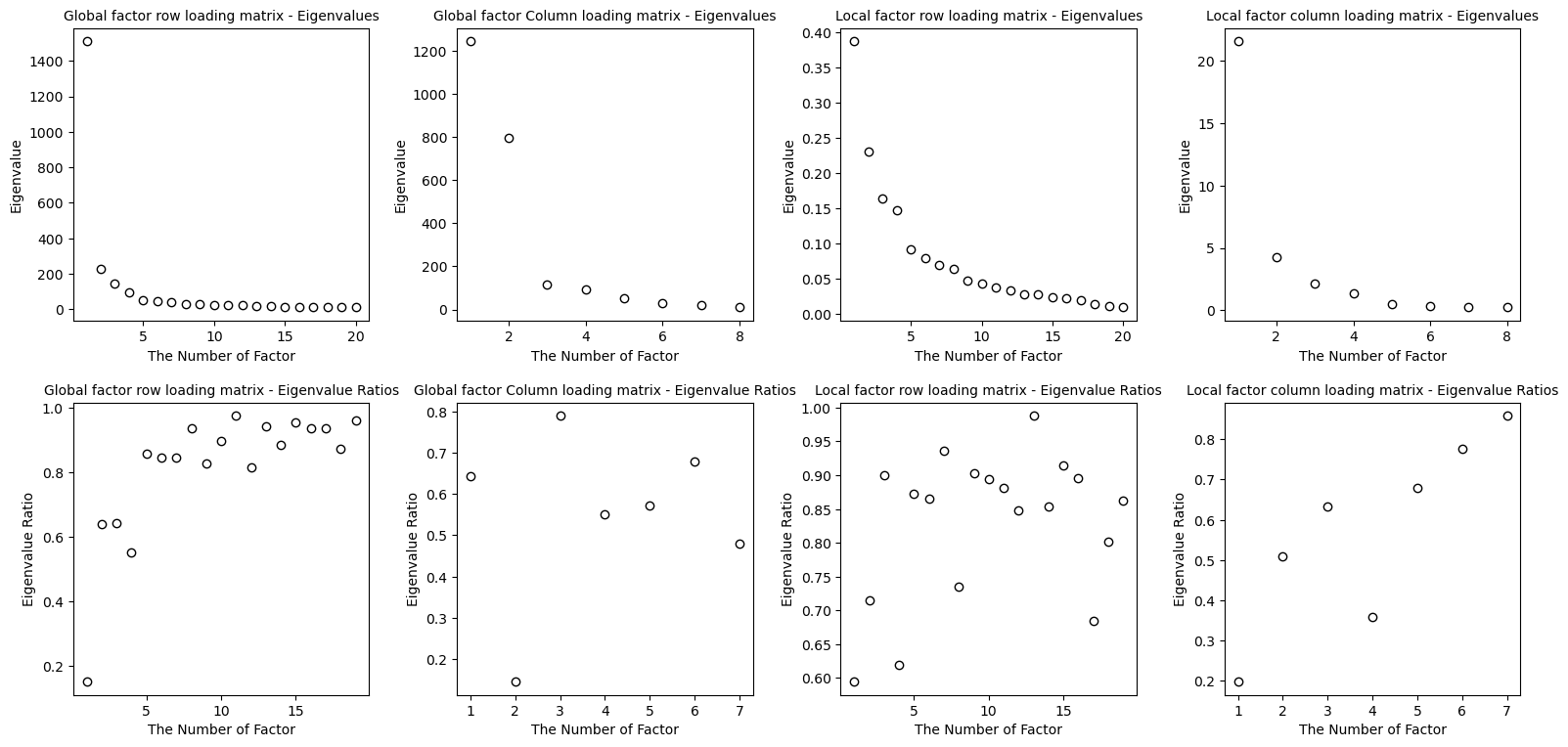}
	\caption{Health care series: eigenvalues and eigenvalues ratio for determining the number of global and local factors }
\end{figure}
\begin{figure}[htbp]
	\centering
        \includegraphics[width=0.95\linewidth]{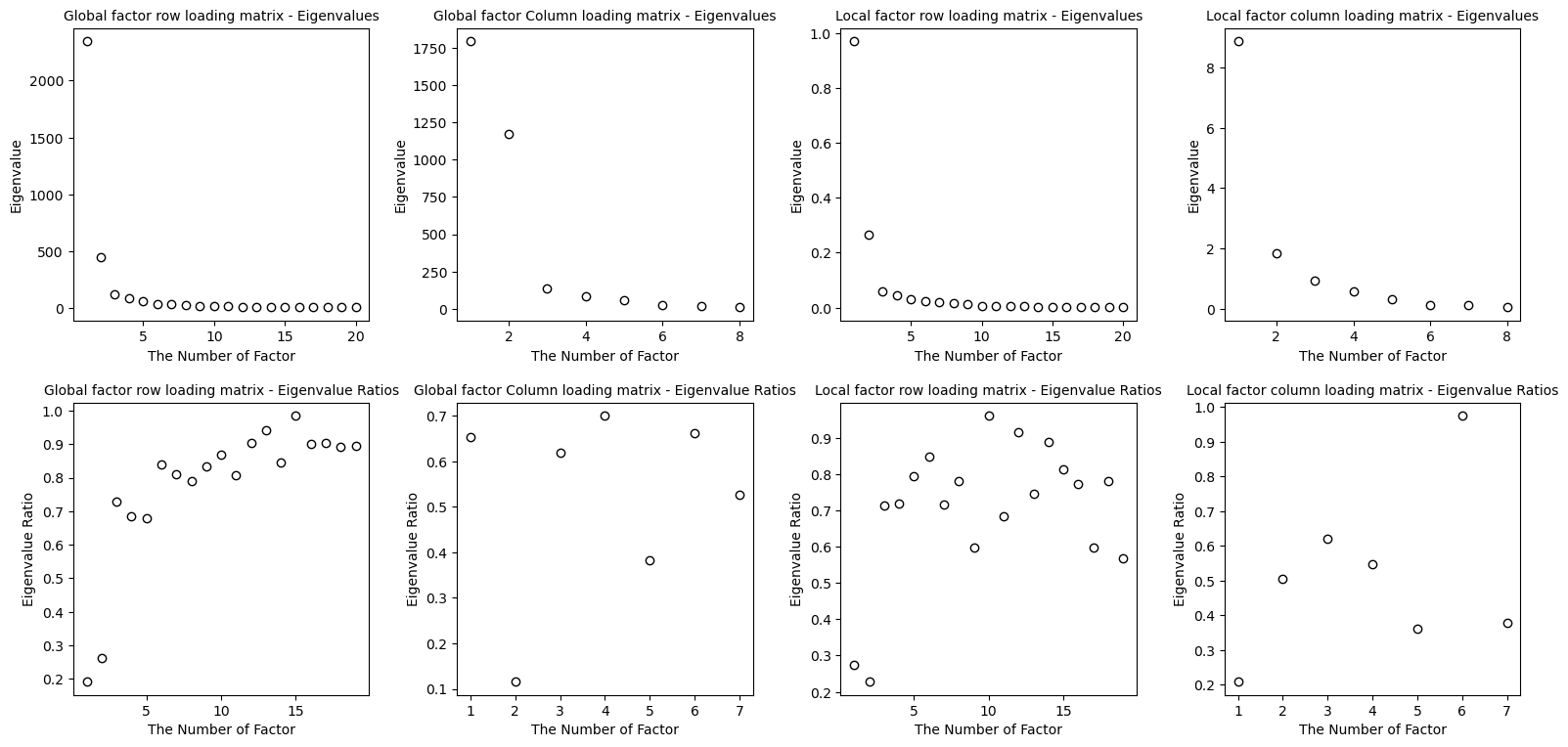}
	\caption{Financials series: eigenvalues and eigenvalues ratio for determining the number of global and local factors }
\end{figure}
\begin{figure}[htbp]
	\centering
        \includegraphics[width=0.95\linewidth]{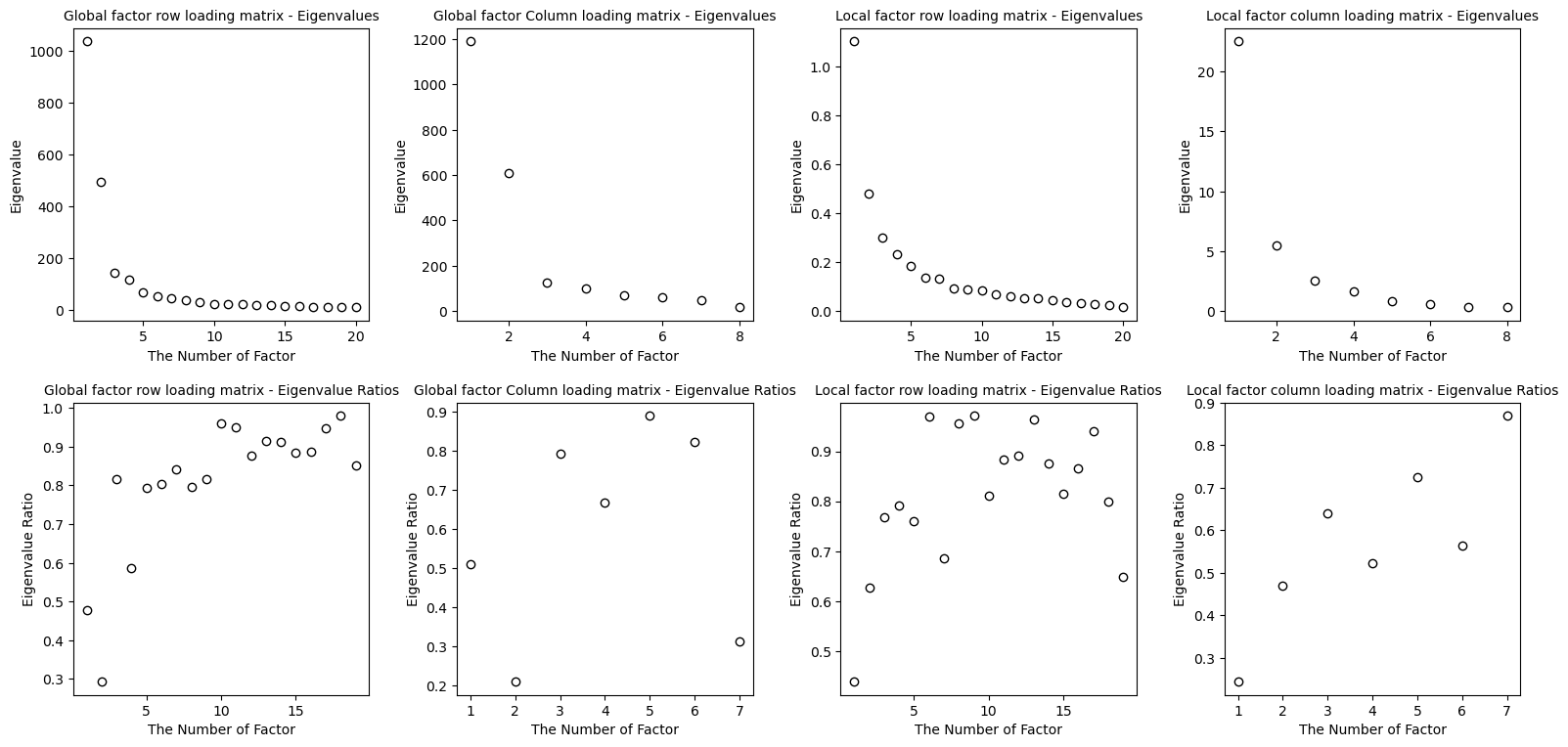}
	\caption{Information technology series: eigenvalues and eigenvalues ratio for determining the number of global and local factors }
\end{figure}
\begin{figure}[htbp]
	\centering
        \includegraphics[width=0.95\linewidth]{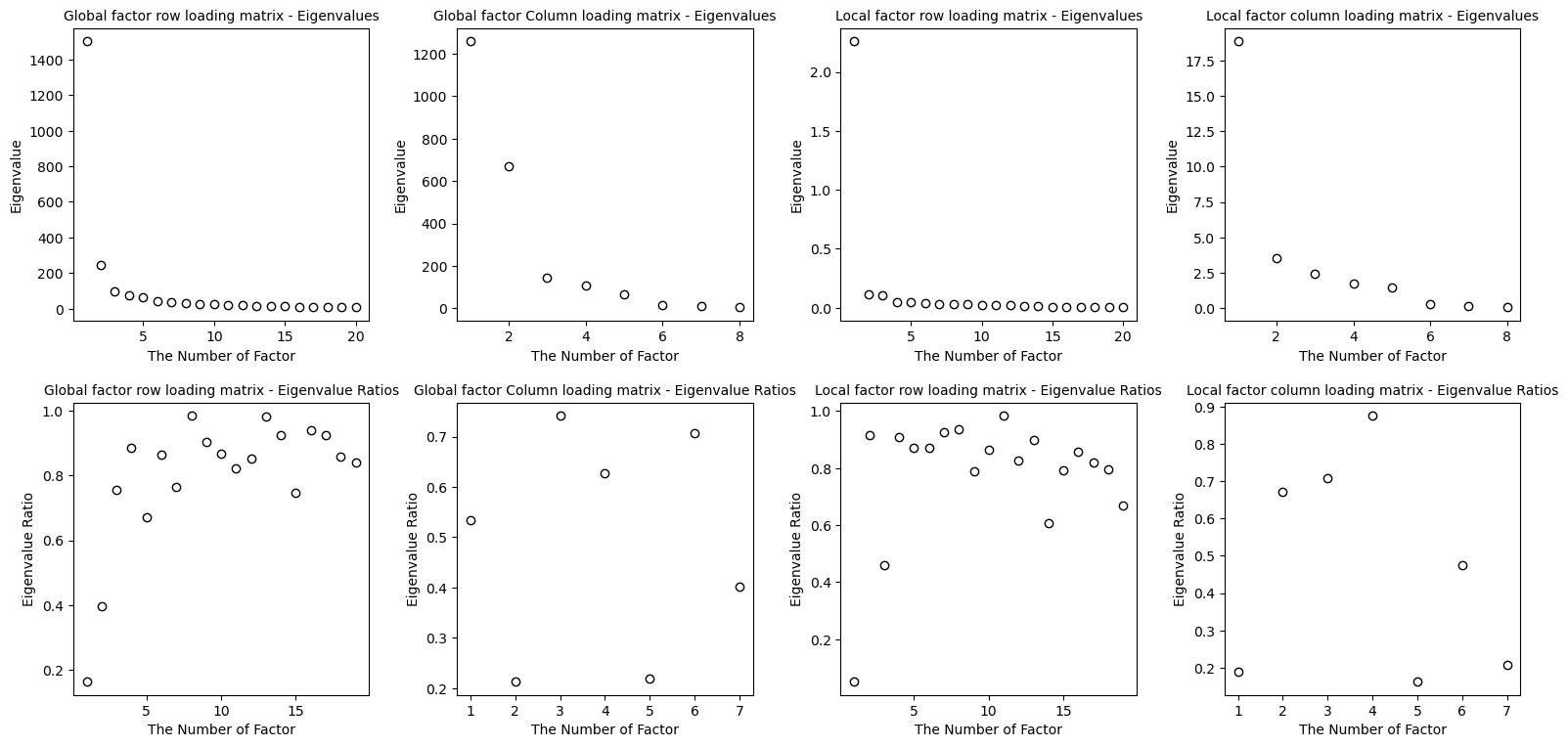}
	\caption{Utilities series: eigenvalues and eigenvalues ratio for determining the number of global and local factors }
\end{figure}
\begin{figure}[htbp]
	\centering
        \includegraphics[width=0.95\linewidth]{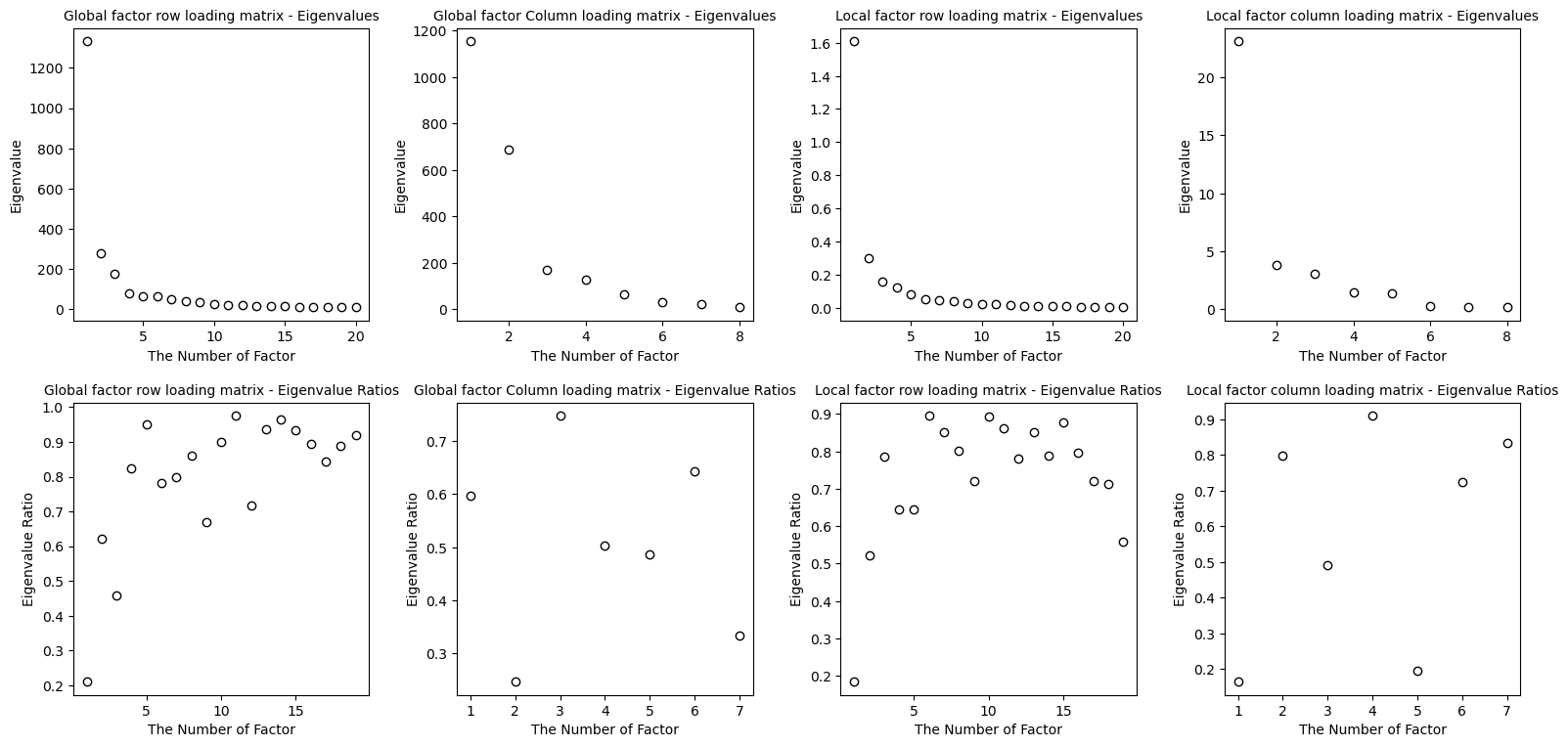}
	\caption{Real estate series: eigenvalues and eigenvalues ratio for determining the number of global and local factors }
	\label{real estate1}
\end{figure}

\end{document}